\documentclass[aps,prx,twocolumn, notitlepage,longbibliography,superscriptaddress,nofootinbib]{revtex4-2}
\usepackage{amsmath}
\usepackage{amsthm}
\usepackage{thm-restate}
\usepackage{amsfonts}
\usepackage{amssymb}
\usepackage{dsfont}
\usepackage{graphicx}
\usepackage{physics}
\usepackage{tikz}
\usepackage[hidelinks]{hyperref}
\usepackage{mathtools}
\usepackage{comment}
\usepackage{multirow}
\usepackage{ulem}
\usepackage{pst-node}
\usepackage{tikz-cd}
\usepackage[shortlabels]{enumitem}
\usepackage{mathtools}
\usepackage[ruled,vlined,linesnumbered]{algorithm2e}
\usepackage[dvipsnames]{xcolor}
\usepackage[english]{babel}

\newtheorem{theorem}{Theorem}

\newtheorem{definition}{Definition}
\newtheorem{lemma}{Lemma}

\newtheorem{statement}{Statement}

\usepackage{pst-all}
\usepackage{leftidx}

\def\H{\mathcal{H}}
\def\Hs{\mathsf{H}}
\def\M{\mathcal{M}}
\def\Z{\mathbb{Z}}
\def\C{\mathcal{C}}

\def\B{\mathcal{B}}
\def\R{\mathcal{R}}
\def\N{\mathcal{N}}
\def\L{\mathcal{L}}
\def\T{\mathcal{T}}

\newcommand{\ZZ}{\mathbb{Z}}

\newcommand{\FF}
{\mathbb{F}}

\newcommand{\CC}{\mathbb{C}}

\newcommand{\bket}[2]{\langle \, #1 \,|\, #2 \, \rangle}
\newcommand{\boket}[3]{\langle\, #1 \,|\, #2 \,|\, #3 \,\rangle}

\newcommand{\be}{\begin{equation}}
\newcommand{\ee}{\end{equation}}
\newcommand{\bc}{\begin{center}}
\newcommand{\ec}{\end{center}}
\newcommand{\nin}{\noindent}

\newcommand{\lo}{\overline}

\newcommand{\as}{\mathsf{a}}
\newcommand{\bs}{\mathsf{b}}
\newcommand{\cs}{\mathsf{c}}
\newcommand{\ds}{\mathsf{d}}

\newcommand{\fs}{\mathsf{f}}

\newcommand{\ts}{\mathsf{\tau}}

\newcommand{\cC}{\mathcal{C}}

\newcommand{\cL}{\mathcal{L}}
\newcommand{\cM}{\mathcal{M}}

\newcommand{\cR}{\mathcal{R}}
\newcommand{\cS}{\mathcal{S}}\newcommand{\cT}{\mathcal{T}}

\newcommand{\cX}{\mathcal{X}}
\newcommand{\cZ}{\mathcal{Z}}

\definecolor{mydarkgreen}{rgb}{0.0, 0.37, 0.1}

\renewcommand{\red}[1]{\textcolor{red}{#1}}
\renewcommand{\blue}[1]{\textcolor{blue}{#1}}
\renewcommand{\green}[1]{\textcolor{mydarkgreen}{#1}}

\newcommand{\rd}{\textcolor{red}{\texttt{r}}}
\newcommand{\bl}{\textcolor{blue}{\texttt{b}}}
\newcommand{\gr}{\textcolor{mydarkgreen}{\texttt{g}}}

\usetikzlibrary{arrows.meta}
\makeatletter
\newlength\lrvec@height
\newlength\lrvec@width
\newif\iflrvec@same@height
\def\lrvec{\@ifstar\slrvec@\lrvec@}
\newcommand{\slrvec@}[2][.4ex]{
  \lrvec@same@heighttrue
  \mathpalette\lrvec@@{{#1}{#2}}
}
\newcommand{\lrvec@}[2][.4ex]{
  \lrvec@same@heightfalse
  \mathpalette\lrvec@@{{#1}{#2}}
}
\def\lrvec@@#1#2{\lrvec@@@#1#2}
\def\lrvec@@@#1#2#3{%
  \iflrvec@same@height
    \settoheight{\lrvec@height}{$\m@th#1 \mathbf{T}#3$}
  \else
    \settoheight{\lrvec@height}{$\m@th#1#3$}
  \fi
  \settowidth{\lrvec@width}{$\m@th#1#3$}
  \kern.08em
  \raisebox{#2}{\raisebox{\lrvec@height}{\rlap{%
    \kern-.05em
    \begin{tikzpicture}[<-> /.tip={To[width=.4em, length=.2em]}]
      \draw [<->] (-.05em,0)--(\lrvec@width+.05em,0);
    \end{tikzpicture}%
  }}}%
  #3
  \kern.08em
}
\makeatother

\usepackage[bottom]{footmisc}

\begin{document}

\title{Non-Abelian qLDPC: TQFT Formalism, Addressable Gauging Measurement and Application to Magic State Fountain on 2D Product Codes}

\author{Guanyu Zhu}
\email{guanyu.zhu@ibm.com}
\affiliation{IBM Quantum, IBM T.J. Watson Research Center, Yorktown Heights, NY 10598 USA}

\author{Ryohei Kobayashi}
\email{ryok@ias.edu} 
\affiliation{School of Natural Sciences, Institute for Advanced Study, Princeton, NJ 08540, USA}

\author{Po-Shen Hsin}
\email{po-shen.hsin@kcl.ac.uk}

\affiliation{Department of Mathematics, King’s College London, Strand, London WC2R 2LS, UK}

\begin{abstract}
A fundamental problem of fault-tolerant quantum computation with quantum low-density parity-check (qLDPC) codes is the tradeoff between connectivity and universality.  It is widely believed that in order to perform native logical non-Clifford gates, one needs to resort to 3D  product-code constructions.   In this work, we extend Kitaev's framework of non-Abelian topological codes on manifolds to non-Abelian qLDPC codes (realized as  Clifford-stabilizer codes) and the corresponding combinatorial topological quantum field theories (TQFT) defined on Poincar\'e CW complexes and certain types of general chain complexes.   We also construct the spacetime path integrals as topological invariants on these complexes.  Remarkably, we show that native non-Clifford logical gates can be realized using constant-rate 2D hypergraph-product codes and their Clifford-stabilizer variants. This is achieved by a spacetime path integral effectively implementing the addressable gauging measurement of a new type of \textit{0-form subcomplex symmetries}, which correspond to addressable transversal Clifford gates and become \textit{higher-form} symmetries when lifted to higher-dimensional CW complexes or manifolds.  Building on this structure, we apply the gauging protocol to the \textit{magic state fountain} scheme for parallel preparation of $O(\sqrt{n})$ disjoint CZ magic states  with code distance of $O(\sqrt{n})$, using a total number of $n$ qubits.

\end{abstract}

\maketitle

\tableofcontents

\section{Introduction}
Understanding the fundamental limit of spacetime overhead for universal fault-tolerant quantum computation is one of the key theoretical problems in quantum error correction.
Significant progress has been made in the study of \textit{quantum low-density parity-check} (qLDPC) codes in recent years in terms of low-overhead quantum information storage \cite{fiberbundlecode21, hastingswr21,pkldpc22,Quantum_tanner,lh22,guefficient22,dhlv23,lzdecoding23,gusingleshot23, Bravyi:2024wc}  with the culmination of achieving asymptotically good qLDPC codes \cite{pkldpc22}.  Nevertheless, ongoing efforts continue to focus on developing logical gates for qLDPC codes.

Until now, a large portion of the works in this category have focused on constructing logical Clifford gates using techniques involving lattice surgery (see e.g.~\cite{cohen22,  cross2024improved, williamson2024low, swaroop2024universal}),  homomorphic measurements \cite{huang2023homomorphic, xu2024fast},  fold-transversal gates \cite{breuckmann2022fold}, and so on.  In order to achieve universality, it is also desirable to construct logical non-Clifford gates, which can reduce the large overhead of magic state distillation.  Another important challenging problem is how to perform addressable and parallelizable logical gates, since unlike topological codes, a large number of logical qubits are encoded into a single code block and it is very hard to address them individually. 

Some recent progress has been made in directly applying transversal non-Clifford gates to qLDPC codes using 3D product codes  corresponding to a 4-term chain complex \cite{zhu2023non, scruby2024quantum, golowich2024quantum, lin2024transversal, breuckmann2024cups, Hsin2024:classifying, zhu2025topological, zhu2025transversal, menon2025magic, li2025transversal, golowich2025constant, tan2025single, gulshen2025quantum}. In particular, Refs.~\cite{zhu2023non, zhu2025topological, zhu2025transversal} have proposed the `\textit{magic state fountain}' scheme to use transversal non-Clifford gates to fault-tolerantly prepare a large number of magic states in parallel in a single shot. This type of approaches are consistent with the Bravyi-K\"onig bound \cite{Bravyi:2013dx} derived in the context of topological codes stating that the spatial dimension needs to be at least $n$ for transversal logical gates in the $n^\text{th}$-level of Clifford hierarchy.  Attempts have also been made to extend such bound to the context of $n$-dimensional hypergraph-product codes \cite{fu2025no}.   Therefore, one would expect that in order to get logical non-Clifford gates, one has to resort to 3D product codes. These codes typically require more demanding  long-range connectivity compared to the 2D product codes for the physical realization on a 2D layout, which is the constraint for most of the quantum computing architecture.    A fundamental question along this direction is hence the tradeoff between universality and connectivity (or dimensionality).  In this paper, we will show that 3D product codes are not necessary for achieving native logical non-Clifford gates, at least when not restricting to transversal gates and Pauli stabilizer codes.

Looking back in history, at the same time as Kitaev introduced the famous toric code more than two decades ago,  he also introduced a large class of 2D non-Abelian topological codes in the form of quantum double models \cite{kitaev2003}.  The non-Abelian codes can admit universal topological quantum computation through braiding anyons.   A natural question to ask is whether such non-Abelian codes can be generalized to qLDPC codes which can have much lower space overhead and larger distance.  

Several fundamental and technical barriers exist along the way. First of all, the non-Abelian codes are always associated with an underlying topological quantum field theory (TQFT).  However, a TQFT is typically defined on a manifold, with its combinatorial version on the triangulation of a manifold.  Moreover, the path integral of a TQFT, also called a state sum (e.g.~\cite{walker2021universalstatesum}), is a topological invariant that is only determined by the manifold topology and is invariant under re-triangulation or equivalently local geometric deformation. Some well-known examples of such manifold invariants include the Turaev-Viro invariants \cite{Turaev:1992hq} and the Dijkgraaf-Witten invariants \cite{Dijkgraaf:1989pz}.  On the other hand, a qLDPC code is often defined on a general chain complex based on expander graphs to boost the rate and distance. The general chain complex is typically not a triangulation or simplicial complex. The other difficulty is that the conventional way of doing topological quantum computation via braiding anyons is not particularly suitable for qLDPC codes,  since the large distance and high rate feature belongs to the ground-state subspace.

In this work, we solve these problems by generalizing the framework of combinatorial TQFT beyond the manifold triangulation and extend it to the Poincar\'e complex,  a special type of CW complex (cellular complex) that admits Poincar\'e duality, and also to a special type of more general chain complexes. In particular,  the path integral of a higher-form generalization of the Dijkgraaf-Witten twisted gauge theory (called `\textit{cubic theory}') previously developed for the study of non-Abelian self-correcting memory \cite{Hsin2024_non-Abelian} has been generalized to a Poincar\'e complex. Such a complex is built from mapping a \textit{skeleton classical code} to a higher-dimensional manifold using a generalized version \cite{zhu2025topological} of the Freedman-Hastings mapping \cite{freedman:2020_manifold_from_code}, and subsequently deformation retract to a higher-dimensional CW complex, which automatically preserves Poincar\'e duality.  The CW complex hence admits cohomology operations such as cup product.   By taking the tensor product of two higher-dimensional CW complexes built from the Tanner graphs of the skeleton classical codes,  we obtain a \textit{`thickened' 2D hypergraph-product code}.
The corresponding Poincar\'e CW complex is then used to construct the path integral of the higher-form twisted gauge theory,  which is a topological invariant of the complex, a generalization of manifold invariants.
In particular, it is a cohomology invariant under the addition of  a coboundary on the cochain (gauge field) configuration, which physically corresponds to the gauge invariance under a gauge transformation.  
Moreover, it is also an invariant under the  subdivision of the CW complex, resembling the re-triangulation invariance in the manifold case. The corresponding code of this twisted gauge theory is a Clifford stabilizer code \cite{Hsin2024_non-Abelian, Davydova:2025ylx, Kobayashi:2025cfh,Warman:2025hov}, with $X$-stabilizer dressed by extra CZ gates. Even more interestingly, we can pull back the path integral invariant defined with cup product and the corresponding Clifford stabilizer code to the skeleton 2D (3-term) chain complexes formed directly by the product of two skeleton classical codes. The Clifford stabilizer codes defined on the CW complexes or the skeleton chain complexes are called  \textit{twisted 2D hypergraph-product codes}.

\begin{figure}
    \centering \includegraphics[width=1\linewidth]{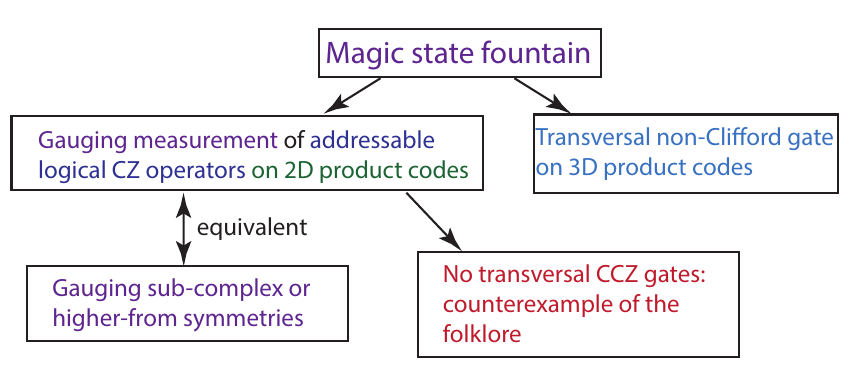}
    \caption{A diagram summarizing the relations between various concepts appeared in this paper.}
  \label{fig:conceptual_diagram}
\end{figure}

Equipped with the above higher-form twisted gauge theory, we are then able to construct the logical non-Clifford gates on 2D hypergraph-product codes \cite{Tillich:2014_hyergraph_product}.  In particular, we generalize the gauging measurement protocol for 2D topological codes in Ref.~\cite{Davydova:2025ylx} to 2D hypergraph-product codes along with the modified construction of the magic state fountain with in Ref.~\cite{zhu2025topological,zhu2025transversal}.   We start with a 2D hypergraph-product code (either the `thickened' or the skeleton version) with constant rate and $\sqrt{n}$ distance,  with the logical qubits initialized in $\lo{\ket{+}}$ or $\lo{\ket{0}}$ states. Next, we perform a gauging measurement of the transversal CZ operators with $O(d)$ rounds of error correction. This leads to a twisted (gauged) 2D hypergraph-product code that is a Clifford stabilizer code corresponding to the higher-form twisted gauge theory.   We then ungauge the twisted code back to the original untwisted hypergraph-product code, while the gauging measurement effectively projects the input logical state into a hypergraph magic state introduced in Refs.~\cite{zhu2023non, zhu2025topological, zhu2025transversal}. 
Subsequent gate teleportation can be used to convert these resource states into logical non-Clifford gate.  

Interestingly, the whole protocol can be derived from the \textit{first principle} via the path integral of the twisted gauge theory introduced above, and the logical action can also be derived from that.   In fact, the path integral is a powerful formalism for studying fault-tolerant logical gates,  with its connection to error correction in the context of topological codes first established in Refs.~\cite{Bauer:2023awl,Bauer:2024qpc,Bauer:2024alh}, and also applied to the gauging measurement protocol for non-Abelian topological codes in Ref.~\cite{Davydova:2025ylx}.    The path integral can be interpreted as an imaginary-time evolution corresponding to a post-selected spacetime history with all the stabilizer measurement outcome being $+1$.   In the more general situation of the error correction process where the measurement has random $\pm 1$ outcome, one can describe the history via an imaginary-time path integral in the presence of topological defects. Based on the syndromes, the decoder then outputs a recovery operation to close the defects.  If the decoder succeeds, then the defects will be closed in a homologically trivial way, which results in a \textit{defect-decorated imaginary-time path integral} equaling to the defect-free path integral due to topological invariance of defects. 
In fact, Kitaev's original way of doing topological quantum computation can also be described by
the defect-decorated path integral, where the defects correspond to worldlines of the non-Abelian anyons.  Nevertheless, a large class of fault-tolerant operations go beyond the braiding picture, including the gauging measurement protocol which instead uses the \textit{spacetime domain walls} to implement logical action, which is more suitable for the case of qLDPC codes.    
This paper generalizes the spacetime path integral from topological codes to qLDPC codes, and applies it to the gauging measurement protocol. Even more generally, one can also view the path integral as a natural language for describing a fault-tolerant Clifford or  non-Clifford circuit implementing logical actions  which can also be considered as a \textit{spacetime code}. For example, a Floquet qLDPC code is an example of such spacetime codes \cite{Hastings:2021ptn, Bauer:2023awl,Bauer:2024qpc,Bauer:2024alh}.

A crucial question when generalizing the gauging measurement to qLDPC codes is the addressability and parallelizability.  The gauging measurement, as well as the recently proposed magic state cultivation \cite{gidney2024magic}, are both  based on the measurement of transversal logical Clifford gates, which seems to be hard due to the large weight.  The cultivation scheme uses a GHZ state as a logical ancilla to perform controlled logical Clifford gate for the measurement, which is not scalable and relies on post selection in the low-distance regime. The 
essence of the gauging measurement is to factorize the global transversal gate corresponding to a global symmetry into a product of local Gauss's law operators corresponding to local gauge symmetries which are dressed X-stabilizer in our example. These local operators have $O(1)$ weight, which can hence be measured fault-tolerantly and is hence scalable to arbitrary distance. Now the issue for generalizing gauging measurement to qLDPC codes is the lack of  addressability since the conventional transversal gate corresponds to a 0-form global symmetry that act on the entire code and hence all logical qubits are involved in the measurement which typically results in a highly entangled magic state without individual addressability. 

Nevertheless, 
it has been realized in Ref.~\cite{zhu2023non} and elaborated in Ref.~\cite{Hsin:2024nwc} that the mechanism of a large class of addressable transversal gates is that they correspond to higher-form ($k$-form) symmetries defined on a codimension-$k$ submanifold as illustrated in Fig.~\ref{fig:subcomplex_symmetry}(a), including transversal Clifford and non-Clifford gates.  On the other hand, there 
exists a folklore (see e.g.~Refs.~\cite{JochymOConnor:2021ih, lin2024transversal, he2025quantum}) in the quantum error correction community that in order to get addressable (targeted) transversal $\text{C}^{n-1}\text{Z}$ gate, one needs to conjugate logical X with a global logical $\text{C}^n\text{Z}$ gate. This folklore seems to suggest that in order to get addressable logical CZ gate, we have to go to a 3D product code
to have transversal CCZ, and hence addressable gauging measurement of logical CZ is impossible for a 2D product code. Surprisingly, this folklore statement is only true for topological codes defined on 2D and 3D manifold, but not the case for 2D product codes defined on general 2D (3-term) chain-complex.  This is because a general 2D chain complex is unlike a 2D manifold which only has a unique top-dimensional  cycle (2-cycle) supported on the entire 2D manifold, instead it can have a large number of top-dimensional cycles (2-cycles).  Therefore,  there exist a large number of possible pairings between 2-cycles and 2-cocycles addressing only a subset of logical qubits (O(1) in our example).  These pairings correspond to \textit{subcomplex symmetries} supported on a subcomplex of the 2D chain complex, providing a more general framework for addressable logical gates that includes higher-form symmetries as special cases. In particular, the symmetries we get here are \textit{0-form subcomplex symmetries} that act on a codimension-0 (top dimensional) subcomplex corresponding to a 2-cycle $\bar\eta^*_2$ as illustrated in Fig.~\ref{fig:subcomplex_symmetry}(b), rather than a subdimensional complex.  Interestingly, when mapping the skeleton 2D chain complex to a higher-dimensional Poincar\'e CW complex or triangulated manifold as mentioned above, the subcomplex symmetry becomes a higher-form ($k$-form) symmetry ($k=2$ in our case), suggesting that the 2D general chain complex is more like a higher-dimensional manifold.

Now the essence of our LDPC gauging measurement protocol is hence the gauging of subcomplex or higher-form symmetries.  Although the gauging operation is applied on the entire complex, where the local Gauss's law operators (dressed X-stabilizers) are measured everywhere, the product of a selected subset of the measurement outcome of the dressed X-stabilizer 
supported on a particular homology class gives rise to a subcomplex (higher-form) symmetry of that homology class. The complete set of dressed X-stabilizer measurement outcomes in the twisted (gauged) code can be combined to obtain the measurement value of the subcomplex (higher-form) symmetries supported on any homology class. This is the essence how independent projective measurements can be obtained for all the transversal CZ operators without additional scaling overhead and in an ancilla-free way, which can simultaneously prepare a large number of disjoint logical magic states (`\textit{magic state fountain}').   In our construction, the fountain can prepare $\Theta(\sqrt{n})=\Theta(\sqrt{k})$ disjoint logical CZ magic states with distance $\Omega(\sqrt{n})$ in $O(d)$ rounds of error correction.  

Besides the application to the magic state fountain, exploring non-Abelian qLDPC codes broadens the landscape of highly-entangled quantum matter.  Here, we also explore the fundamental physics problem of understanding the non-Abelian topological order of qLDPC codes or more generally $k$-local quantum state of matter with long-range connectivity.  With the mapping to higher-dimensional CW complex,   we can again understand the non-Abelian fusion and braiding statistics in the language of twisted gauge theory.  In particular, the Borromean ring braiding statistics, which serves as the signature of the non-Abelian topological order of our twisted (gauged) qLDPC codes, can be derived from the spacetime path integral mentioned above.  Interestingly, the Borromean ring statistics is deeply related to the triple intersection of the world-sheet of magnetic excitations in the spacetime picture, which is also the underlying mechanism of the logical action introduced above.
 
Moreover, qLDPC codes provide a stronger form of topological ordered states: all states below certain energy density contain no trivial states, i.e. no low energy trivial states (NLTS), which was conjectured to exist in Ref.~\cite{Freedman:2013zfj} and proven in Ref.~\cite{Anshu:2022hsn}.  In contrast, most topological codes only have nontrivial topological orders in the ground states. In addition, the low energy states of qLDPC codes give rise to topological quantum spin glasses \cite{Placke:2024wey}, where constant rate gives rise to exponential numbers of stable, topologically ordered Gibbs states that persist at finite temperatures. This is a quantum analogue of the classical spin glasses with complex free energy landscapes that provide a mechanism for long-lived memory \cite{ANDERSON1970549,PhysRevLett.43.1754,Binder:1986zz,568530b2-62d5-3d43-9e2f-adfdf424006b}.  Studying non-Abelian generalization of the above types of phases of matter will hence be an interesting future direction.

\begin{figure}[t]
    \centering
    \includegraphics[width=1\linewidth]{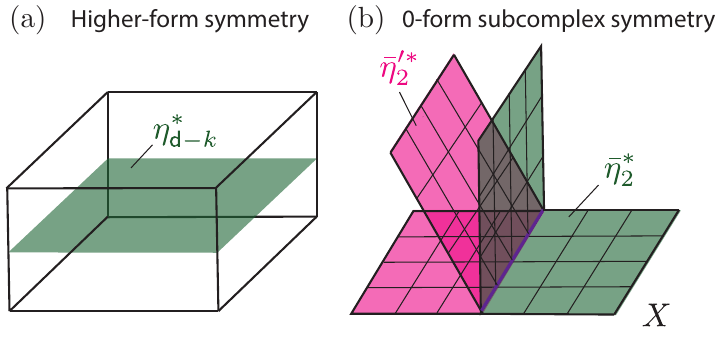}
    \caption{(a) Illustration of a higher-form ($k$-form) symmetry in a $\ds$-dimensional  manifold or Poincar\'e CW complex. The symmetry operator is supported on a codimension-$k$ submanifold/subcomplex corresponding to a $(\ds-k)$-cycle $\eta^*_{\ds-k}$. (b) Illustration of 0-form subcomplex symmetries on a 2D square complex $X$ which has a `book-like' structure: a purple edge (book hinge) are adjacent to multiple squares (book pages).  There can exist a large number of codimension-0 (top dimensional) subcomplexes corresponding to 2-cycles such as $\bar\eta^*_2$ and $\bar\eta'^*_2$.  Each of these 2-cycles can support a symmetry operator, leading to addressable transversal Clifford gate such as $\widetilde{\text{CZ}}_{\bar\eta^*_2}$.  }
    \label{fig:subcomplex_symmetry}
\end{figure}

\subsection{Summary of results}

We summarize our main results and the flow of this paper as follows:

First of all, we construct the first non-Abelian qLDPC codes and the corresponding TQFTs on certain types of general chain complexes and Poincar\'e CW complexes that are beyond the usual manifold scenario. The qLDPC codes in this case have non-Abelian topological order. In particular, we construct the spacetime path integral of a twisted higher-form gauge theory defined on a Poincar\'e CW complex in 
Eq.~\eqref{eq:path_inegral_cup} of Sec.~\ref{sec:Feynman_path_integral} using cup product, and show that it is a topological invariant in Lemma \ref{lemma:invariant} and \ref{lemma:subdivision}.  An explicit construction with non-trivial triple intersection in the path integral is also given in Sec.~\ref{sec:triple_intersection_construction} which guarantees non-trivial logical action in the gauging measurement protocol.  We then present the corresponding Clifford stabilizer code, also called \textit{twisted code} in Eq.~\eqref{eq:stabilizer_summary} of Sec.~\ref{sec:Clifford_stabilizer_codes}, along with other basic properties such as stabilizer commutation relations in Lemma \ref{lemma:stabilizer_commutation} and \ref{lemma:stabilizer_magnetic_commutation} and the expression of the logical operators.  The derivation from gauging a higher-form symmetry-protected topological (SPT) phase is presented in Appendix \ref{app:gauging_derivation}.  

We then show in Sec.~\ref{sec:charge_parity_operator} that the charge parity operators in the twisted code defined on the $d$-dimensional CW complex, which are equivalent to the addressable transversal CZ gates in the untwisted code, have the form of higher-form ($k$-form) symmetries supported on codimension-$k$ subcomplexes, or equivalently $(\ds -k)$-cycles $\eta^*_{\ds-k}$, as illustrated in Fig.~\ref{fig:subcomplex_symmetry}(a). In particular, these higher-form ($k$-form) symmetries ($k>0$), as well as its generalization of subcomplex symmetries ($k \ge 0$) 
provide a general mechanism for addressable transversal gates, including Clifford and non-Clifford gates.

Next, we study the explicit construction of the \textit{ thickened 2D hypergraph-product (HGP) code}, both the untwisted and twisted versions in Sec.~\ref{sec:rate_and_distance}, which are defined on a 16D Poincar\'e CW complex as a tensor product of two 8D CW complexes built from the skeleton classical codes as reviewed in Sec.~\ref{sec:code_to_manifold}.   For both the untwisted and twisted cases, we show in Lemma \ref{lemma:rate_untwisted} to \ref{lemma:distance_twisted} that the codes have constant encoding rate $k=\Theta(n)$ and polynomial subsystem-code distance $d=\Omega(\sqrt{n})$.  Here, we consider the subsystem-code distance in order to ignore the spurious short (co)cycles in the construction by treating them as logical operator support of gauge qubits, using Lemma \ref{lemma:subsystem} which is proven in Appendix \ref{app:subsystem}.   Note that the distance bound for the twisted codes is obtained assuming the well-known Statement \ref{statement:TQFT} about condensation descendants from TQFT holds in our case.  Since the codes we construct are also just equivalent to a combinatorial TQFT, Statement \ref{statement:TQFT} is unlikely to be violated.

In Sec.~\ref{sec:pullback_skeleton}, we then consider the pullback of the twisted (untwisted) thickened 2D HGP code defined on the 16D Poincar\'e CW complex to the \textit{twisted (untwisted) skeleton 2D HGP code} defined on the skeleton 2D general chain complex formed directly by the product of the skeleton classical codes. This minimizes the constant overhead in the construction which makes the construction particularly realistic for near-term implementation. Remarkably, the cup product structure and path integral invariants are also pulled back to the skeleton codes defined on the 2D general chain complexes.  In this way, one can also see how to define the path integral invariant without using cup product but just the chain-cochain pairing in the skeleton codes.   The code parameter scaling of the skeleton code is the same as the thickened code as stated in Lemma \ref{lemma:scaling_skeleton}.

Even more interestingly, in Sec.\ref{sec:0-form}, the higher-form symmetries in the higher-dimensional CW complex are pulled back to a new type of 0-form subcomplex symmetries supported on a codimension-0 (top-dimensional) 2-cycle of a general 2D chain complex. As illustrated in the toy example of a 2D square complex in Fig.~\ref{fig:subcomplex_symmetry}(b),  there is a book-like structure which is typical for qLDPC codes where more than two faces (pages) are adjacent to a single edge (hinge in purple), which is not allowed in the celllulation of a 2-manifold.  Therefore, in general there are many inequivalent 2-cycles $\bar \eta^*_2$ (illustrated by the green and yellow membranes in Fig.~\ref{fig:subcomplex_symmetry}(b)) which can support the subcomplex symmetry operators. These operators correspond to addressable transversal Clifford gates, such as the transversal CZ gates $\widetilde{\text{CZ}}_{\bar\eta^*_2}$, which gives the counterexample of the folklore statement since there is no transversal CCZ gates in the 2D HGP codes. 

We then present the addressable gauging measurement protocol for gauging the higher-form or subcomplex symmetries in Sec.~\ref{sec:parallel_gauging}, generalizing the protocol in Ref.~\cite{Davydova:2025ylx} for topological codes. With the specific triple-intersection structure chosen in Sec.~\ref{sec:triple_intersection_construction} and proper initialization of the logical qubits, we are able to realize the \textit{magic state fountain} that prepares $\Theta(\sqrt{n})$ logical CZ magic states with $\Omega(\sqrt{n})$ distance in parallel  using $O(d)$ rounds of error corrections, as stated in Theorem \ref{theorem:protocol}.   We further derive the logical action of the gauging measurement protocol via the spacetime path integral in Sec.~\ref{sec:spacetime_path_integral}.  We also explain the tensor-network picture of the spacetime path integral and the connection to imaginary-time evolution and error correction.   
The protocol is then also adapted to the skeleton HGP codes via the pullback described above in Sec.~\ref{sec:pullback_protocol}.

Finally, we show that the twisted qLDPC codes we have constructed have non-Abelian topological order by deriving the non-Abelian fusion rules and the Borromean ring braiding statistics in Sec.~\ref{sec:nonabelianfusionbraiding}.

\section{Review: mapping classical LDPC codes to CW complexes}\label{sec:code_to_manifold}

In this section, we first review the CW complex construction from the skeleton classical codes in Ref.~\cite{zhu2025topological}.  This paves the way for the later construction of twisted 2D hypergraph-product code in Sec.~\ref{sec:non-Abelian_LDPC}.

\begin{figure*}[t]
    \centering
    \includegraphics[width=1\linewidth]{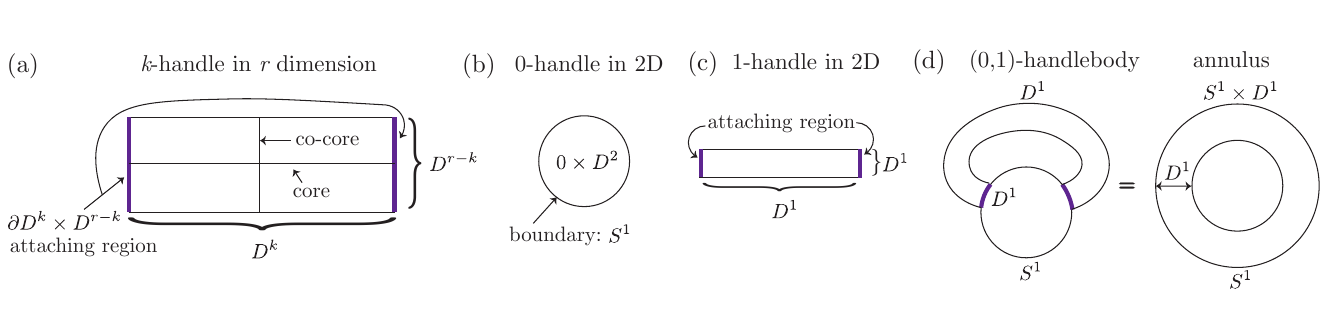}
    \caption{ (a) Anatomy of a \(k\)-handle in \(r\) dimensions. (b) A 0-handle in 2D, \(D^{0}\times D^{2}\), whose boundary is \(S^{1}\). (c) A 1-handle in 2D, \(D^{1}\times D^{1}\), with attaching region \(D^{1}\) (purple). (d) Attaching the 1-handle to the 0-handle along \(D^{1}\) produces a \((0,1)\)-handlebody homeomorphic to the annulus \(S^{1}\times D^{1}\).}
    \label{fig:handle_introduction}
\end{figure*}

\subsection{Mapping classical codes to manifolds}
We begin with the mapping of a classical code to an $r$-dimensional manifold $\cM^r$. For simplicity, we choose $r \ge 8$ as in \cite{zhu2025topological} to ensure the separation of the dimensions of the logical and spurious (co)cycles.   This starts with a skeleton classical code $\bar{\C}_{\text{c}}$ associated with the $\ZZ_2$ ($\FF_2$) chain complex $X$:
\begin{equation}\label{eq:classical_chain complex}
  \bar C_1 \xrightarrow[]{\partial_{1}=\mathsf{H}} \bar C_{0}.   
\end{equation}
In order to construct a manifold, one needs to lift it to a chain complex $\hat{X}$ over $\ZZ$ coefficients:
\begin{equation}\label{eq:lifted_chain_complex}
  \hat{\bar C}_1 \xrightarrow[]{\hat{ \partial}_{1}=\hat{\mathsf{H}}} \hat{\bar C}_{0},   
\end{equation}
where we have the lifted boundary map and parity-check matrix $\hat{\partial}_{1}=\hat{\mathsf{H}}$ over $\ZZ$. The choice of lift is not unique, however, for the classical code a straightforward choice is the `naive lift' \cite{freedman:2020_manifold_from_code}:  0 mod 2 $\rightarrow $ 0,  1 mod 2 $\rightarrow$ 1.   We now use the handle construction to build a manifold $\cM^r$ from the classical code, which corresponds to the following handle chain complex $\cL_h$ (for $r=8$):
\begin{widetext}
\begin{align}\label{eq:long_chain_CW}
& C_8 \rightarrow C_7 \rightarrow C_{6} \xrightarrow[]{\hat{\partial}^T = \hat{\mathsf{H}}^T}  C_5 \rightarrow 0  \rightarrow C_3 \xrightarrow[]{\hat{\partial} = \hat{\mathsf{H}}} C_{2} \rightarrow C_1 \rightarrow C_0, \cr
&\qquad \qquad \qquad  \qquad \qquad \qquad \qquad \ \ \ \ \text{bit}  \qquad \text{check} 
\end{align}
\end{widetext}
where $C_k=\text{span}_\ZZ(k\text{-}\text{handles})$. Note that since we have used a naive lift, the above handle chain complex $\cL_h$ can be viewed as chain complex either over $\ZZ$ or $\ZZ_2$.   Here, a  $k$-handle in $r$ dimensions ($r=8$ in this case) is a pair
\begin{equation}\label{eq:handle_expression}
  h_k=(D^k \times D^{r-k}, \partial D^k \times D^{r-k}),  
\end{equation}
which represents an $r$-dimensional manifold $D^k \times D^{r-k}$ together with its attaching region $\partial D^k \times D^{r-k} = S^{k-1} \times D^{r-k}$, as illustrated in Fig.~\ref{fig:handle_introduction}(a).  Here $D^k$ and $S^{k-1}$ represent a $k$-dimensional ball and a ($k-1)$-dimensional sphere respectively.  The two factors in the product $D^k$ and $D^{r-k}$ are the \textit{core} and \textit{co-core}  of a $k$-handle respectively. 

The handle chain complex $\L_h$ specifies how the $k$-handle is attached to $(k-1)$-handles for all $k \le r$.   In particular, the attaching map between the 3-handle and 2-handle is specified by the lifted boundary map $\hat{\partial}_1= \hat{\mathsf{H}}$ of the skeleton classical code according to Eq.~\eqref{eq:lifted_chain_complex}.   Note that the left portion of the handle chain complex $\cL_h$ is isomorphic to the right portion up to a transpose, which is a reflection of the underlying Poincar\'e duality of the manifold.  In particular, the chain group with dual dimensions $k$ and $r-k$ are isomorphic, e.g., $C_3$ and $C_5 \cong C^*_3$, as well as $C_2$ and $C_6 \cong C^*_2$. Here $C^*_k$ is the dual $k$-chain group of the dual complex $\cL^*_h$, which is spanned by the dual $k$-handles $h^*_k$, and defined as follows:
\begin{align}\label{eq:long_chain_dual}
 \cdots \leftarrow C_1^* \leftarrow C_{2}^* \xleftarrow[]{\hat{\partial} \sim {\mathsf{H}}}  C^*_3 \leftarrow 0  \leftarrow C^*_5 \xleftarrow[]{\hat{\partial}^T  \sim {\mathsf{H}}^T} C^*_{6} & \leftarrow C^*_7 \cdots \cr 
\end{align}
The Poincar\'e duality corresponds to the isomorphism of the homology and cohomology group \cite{Hatcher:2001ut}: 
\be
H_k(\cL_h) \cong H^{r-k}(\cL^*_h),
\ee
where we have $r=8$ in this construction.

Now we explain the details of the handle construction.  Besides the standard handle expressed in Eq.~\eqref{eq:handle_expression}, we also need to introduce the dressed $k$-handles, which are themselves composed of handles with the same index ($k$) and lower indexes (denoted by ``$k$-handle" from now on for simplicity following the convention in Ref.~\cite{freedman:2020_manifold_from_code}). In $r$ dimensions,  each ``$k$-handle" has the form 
\begin{equation}\label{eq:dressed_handle}
\tilde{h}_k=(N_k \times D^{r-k}, \partial N_k \times D^{r-k} ),     
\end{equation}
where $D^{m}$ represents the $m$-dimensional disk (ball) and $N_k$ is the dressed core (composed of genuine cores with the same or lower indexes).  From now on we focus on the $r=8$ case.

\begin{figure*}[t]
    \centering
    \includegraphics[width=1\linewidth]{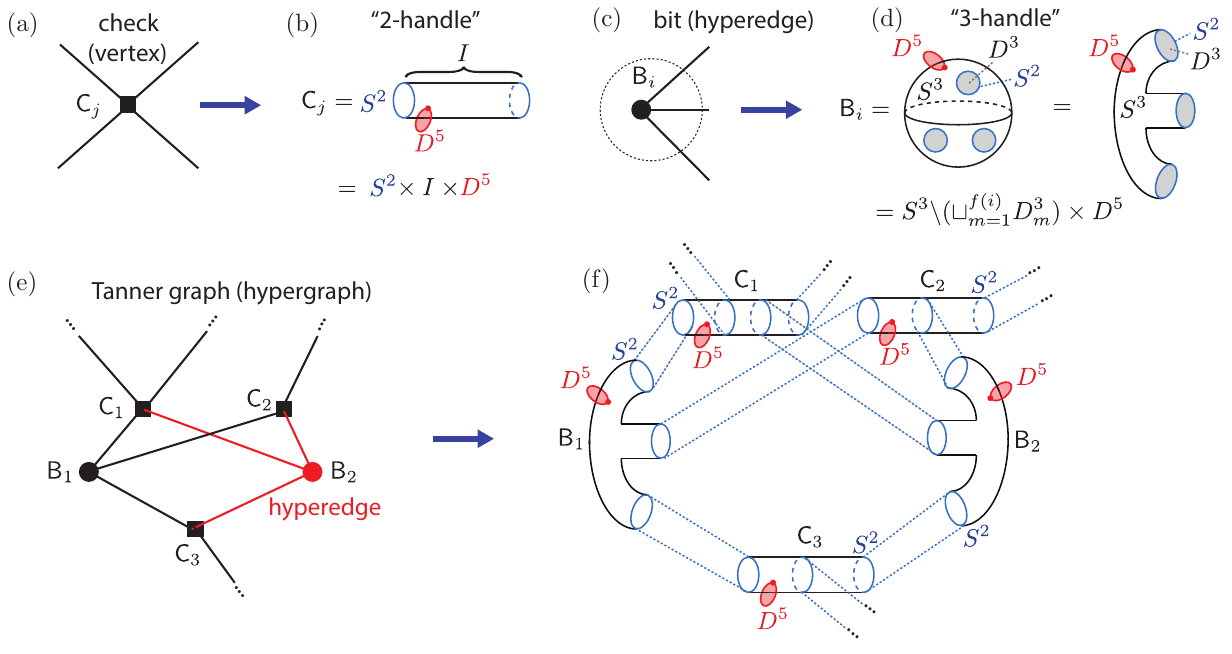}
    \caption{ Plumber’s view of classical codes. (a) A parity check \(\mathsf{C}_j\) represented as a vertex (square). (b) The check \(\mathsf{C}_j\) is mapped to a “2-handle,” \(\mathsf{C}_j = S^{2} \times I \times D^{5}\). (c) A bit \(\mathsf{B}_i\) associated with a hyper-edge consisting of multiple edges incident on a single circle. (d) The bit \(\mathsf{B}_i\) is mapped to a “3-handle,” obtained from a 3-sphere \(S^{3}\) with several 3-disks \(D^{3}\) removed and then thickened by \(D^{5}\); its multiple legs attach to neighboring 2-handles. (e) A bipartite Tanner graph with check nodes (squares) and bit nodes (circles), equivalently viewed as a hypergraph with checks on vertices and bits on hyper-edges (red). (f) Thickening the Tanner graph yields a handlebody in which 3-handles are attached to adjacent 2-handles according to the boundary map of the Tanner graph, by gluing along the attaching regions \(S^{2} \times D^{5}\) of the 3-handles.}
    \label{fig:dictionary}
\end{figure*}

We begin by associating a 0-handle \(h_0 = 0 \times D^{8} \cong D^{8}\) with each check of the classical code. For each such 0-handle, we then attach a 2-handle \(h_2 = D^{2} \times D^{6}\) along its boundary [the 2D analogue of attaching a 1-handle to a 0-handle is illustrated in Fig.~\ref{fig:handle_introduction}(b--d)]. The attaching region of \(h_2\) is \(\partial D^{2} \times D^{6} = S^{1} \times D^{6}\), with attaching map
\be
S^{1} \times D^{6} \hookrightarrow \partial D^{8} = S^{7}.
\ee
This construction yields a \((0,2)\)-handlebody
\be
\sqcup_{j=1}^{\bar{n}_c} (S^{2} \times D^{6})_j \;\equiv\; \sqcup_{j=1}^{\bar{n}_c} \mathsf{C}_j,
\ee
where \(\sqcup\) denotes a disjoint union, \(\mathsf{C}_j= S^{2} \times D^{6}\) is the dressed ``2-handle" associated with the \(j^\text{th}\) check, \(\bar{n}_c\) is the total number of checks, and the dressed core defined in Eq.~\eqref{eq:dressed_handle} is $N^{2} = S^{2}$.   The mapping between the check and the ``2-handle" $\mathsf{C}_j$ is illustrated in Fig.~\ref{fig:dictionary}(a,b),  where we have re-written the ``2-handle" as $\mathsf{C}_j= S^{2} \times I \times  D^{5}$ and $I \equiv D^1$ represents an interval.  In this way, one can interpret $\mathsf{C}_j$ as $S^2 \times I$ thickened by $D^5$ in extra dimensions.   Note that both $D^5$ and $S^2$ along extra dimensions for the skeleton classical code have $O(1)$ size in our construction. Note that each dressed ``2-handle" has a handle decomposition in terms of one 2-handle and one 0-handle as described above.

\begin{figure*}[t]
    \centering
    \includegraphics[width=1\linewidth]{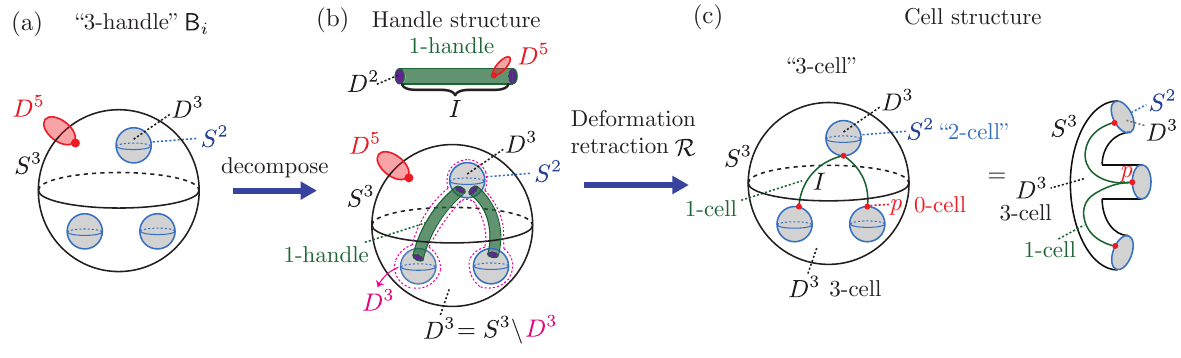}
    \caption{Anatomy of the dressed ``3-handle.''
(a) Realistic three-dimensional representation of the ``3-handle". The 3-sphere is viewed as the interior of a 3-ball whose boundary is identified to a single point. Three thickened 3-balls \(D^{3}\times D^{5}\) (gray) are removed, leaving thickened 2-sphere boundaries \(S^{2}\times D^{5}\).
(b) On the dressed core, the ``3-handle" contains two 1-handles connecting the three thickened \(S^{2}\) boundaries; together with the three removed \(D^{3}\) regions, these form a single 3-ball \(D^{3}\) (highlighted by pink dashed lines). The complement in the 3-sphere is another 3-ball, \(D^{3}=S^{3}\setminus D^{3}\), which becomes the 3-handle upon thickening by \(D^{5}\).
(c) Deformation retraction of the two 1-handles and the 3-handle to two 1-cells and one 3-cell, respectively, with the dressed ``3-handle" retracting to its core. The resulting cell connects to three neighboring ``2-cells" \(S^{2}\) and three 0-cells \(p\) on adjacent ``2-handles". The drawing may be further deformed into the abstract representation shown on the right, matching the style of Fig.~\ref{fig:dictionary}(d).}
    \label{fig:retraction_3-handle}
\end{figure*}

  We then associate each bit with a  ``3-handle"  $\mathsf{B}_i$$=$$N_3$$\times$$D^5$, where the dressed core $N_3=(S^3 \backslash \sqcup_{m=1}^{f(i)}D^3_m)$ is a punctured 3-sphere and $f(i)$ is the number of checks that the bit $\mathsf{B}_i$ couples to, as illustrated in Fig.~\ref{fig:dictionary}(c,d).   The attaching region of the ``3-handles" is $\partial N_3 \times D^6 = \sqcup_{m=1}^{f(i)} (S^2 \times D^5)_{i,m}$, i.e., 2-sphere thickened along $D^5$ in the extra dimensions.  
  
  The handle decomposition structure of the ``3-handle" is illustrated in Fig.~\ref{fig:retraction_3-handle}(a,b) (with $f(i)=3$),  using a more realistic three-dimensional representation of the dressed core \(N_i\), as opposed to the abstract illustration shown in Fig.~\ref{fig:dictionary}(d).  One can think of the 3-sphere $S^3$ being the entire 3D space inside a 3-ball $D^3$ with the boundary $S^2$ being identified to a single point.  We begin with \(f(i)\) disconnected thickened \(S^{2}\) boundary components, each of
the form \(S^{2} \times D^{5}\), and attach handles in increasing index, as
illustrated in Fig.~\ref{fig:retraction_3-handle}(b). We explain the handle decomposition as follow:

We first introduce
\((f(i)-1)\) 1-handles
\[
h_1
=
\bigl(
I \times (D^{2} \times D^{5}),\;
S^{0} \times (D^{2} \times D^{5})
\bigr),
\]
where \(S^{0}\) denotes the two endpoints of the interval \(I\). The attaching
region therefore consists of two copies of the thickened 2-disk
\(D^{2} \times D^{5}\) (shown in purple). Attaching these \((f(i)-1)\) 1-handles
connects the \(f(i)\) initially disconnected thickened \(S^{2}\) boundaries.
Each 1-handle joins two distinct thickened \(S^{2}\) components, which can
equivalently be viewed as attachments to the boundaries of the dressed
``2-handles'' \(\mathsf{C}_j\) associated with the dressed ``3-handle''
\(\mathsf{B}_i\).

Within the 3-sphere \(S^{3}\) containing the dressed core, the union of the
\((f(i)-1)\) 1-handles and the \(f(i)\) removed 3-balls \(D^{3}\) (shown in gray)
forms a single connected component homeomorphic to a 3-ball \(D^{3}\),
highlighted by the pink dashed curves. Removing this component from \(S^{3}\)
leaves the complement $S^{3} \setminus D^{3} \cong D^{3}$,
i.e., another 3-disk. Thickening this complement by \(D^{5}\) produces a 3-handle
\[
h_3
=
\bigl(
D^{3} \times D^{5},\;
S^{2} \times D^{5}
\bigr),
\]
whose attaching region is the thickened 2-sphere
\(S^{2} \times D^{5}\) (indicated by the purple dashed region). Altogether, the
dressed ``3-handle'' \(\mathsf{B}_i\) consists of one 3-handle and
\((f(i)-1)\) 1-handles.

  \begin{figure*}[t]
    \centering
    \includegraphics[width=1\linewidth]{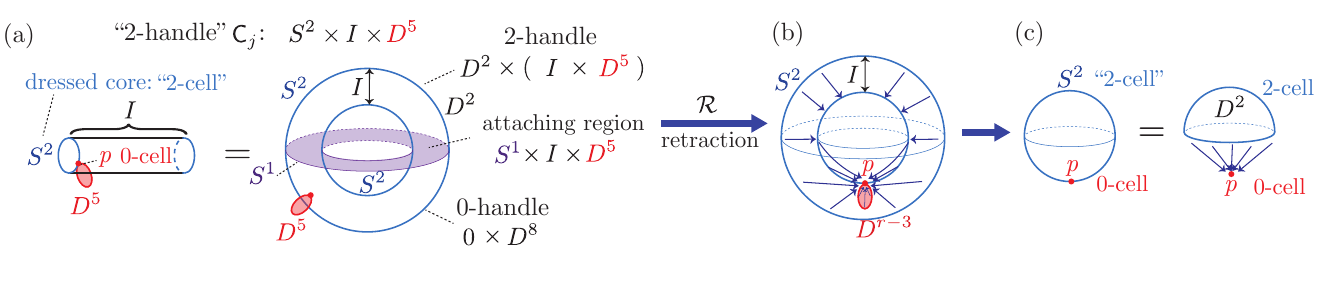}
    \caption{Anatomy of the dressed ``2-handle''. (a) The ``2-handle" has a dressed core: the ``2-cell" $S^2$.  It is formed by attaching a 2-handle $D^2 \times (I \times D^5)$ to the 0-handle $0 \times D^8$.  Here, the 0-handle is deformed into a south-facing half-shell \(D^{2}\times I\), thickened by \(D^{5}\).  The 2-handle is deformed into a north-facing thickened half-shell, whose attaching region (purple) is a thickened annulus along the equator. Attaching the 2-handle to the 0-handle along the equatorial annulus yields the dressed ``2-handle", a full shell \(S^{2}\times I\) thickened by \(D^{r-3}\). 
(b,c)  Deformation retracting the 0-handle to a single point \(p\), forming a 0-cell (vertex \(p\))  in the south pole, while
retracting the 2-handle to a 2-cell  \(D^{2}\) (disk) which is attached to the 0-cell and form a dressed ``2-cell".}
    \label{fig:retraction_2-handle}
\end{figure*}
  
  Now since each ``3-handle" and ``2-handle" only consists of one 3-handle and 2-handle respectively, the ``3-handles" are attached to ``2-handles" according to the lifted boundary map $\hat{\partial}=\hat{\mathsf{H}}$ as shown in Eq.~\eqref{eq:long_chain_CW}. We then attach the ``3-handles"  along the attaching region to the boundary of ``2-handles" $\partial (S^2 \times D^6)=S^2 \times S^5$  with the following attaching map:
\begin{equation}
(S^2 \times D^5)_{i, m}  \hookrightarrow (S^2 \times S^5)_j. 
\end{equation}
We have hence constructed a 3-handlebody $H$ corresponding to the right portion of the handle chain complex $\cL_h$ in Eq.~\eqref{eq:long_chain_CW}. The mapping from the Tanner graph (hypergraph) of a classical code $\bar{\C}_{\text{c}}$ and the 3-handlebody $H$ is illustrated in Fig.~\ref{fig:dictionary}(e, f).  

We now obtain the left portion by taking an identical upside-down copy of the 3-handlebody $H^*$, which is composed of the dual $k$-handles in the dual handle complex $\cL_h^*$ equivalent to the $(8-k)$-handles in the original complex $\cL_h$, i.e.,  $h^*_k = h_{8-k}$.   In particular, the 3-handle and 2-handle corresponding to the bit and check become the 5-handle and 6-handle in the upside-down copy $H^*$.    
We can hence build the closed manifold $\cM^{8}=\mathcal{D}H$ as the double of the 3-handlebody $H$ by gluing the two copies $H$ and $H^*$ along their common boundary $\partial H$ with an identity map $id_{\partial H}$, i.e., 
\begin{equation}
\mathcal{D}H=H \cup_{id_{\partial H}} H^*.  
\end{equation}
The closed manifold corresponds to the entire handle chain complex in Eq.~\eqref{eq:long_chain_CW}. 

\begin{figure*}[t]
\centering	\includegraphics[width=1\textwidth]{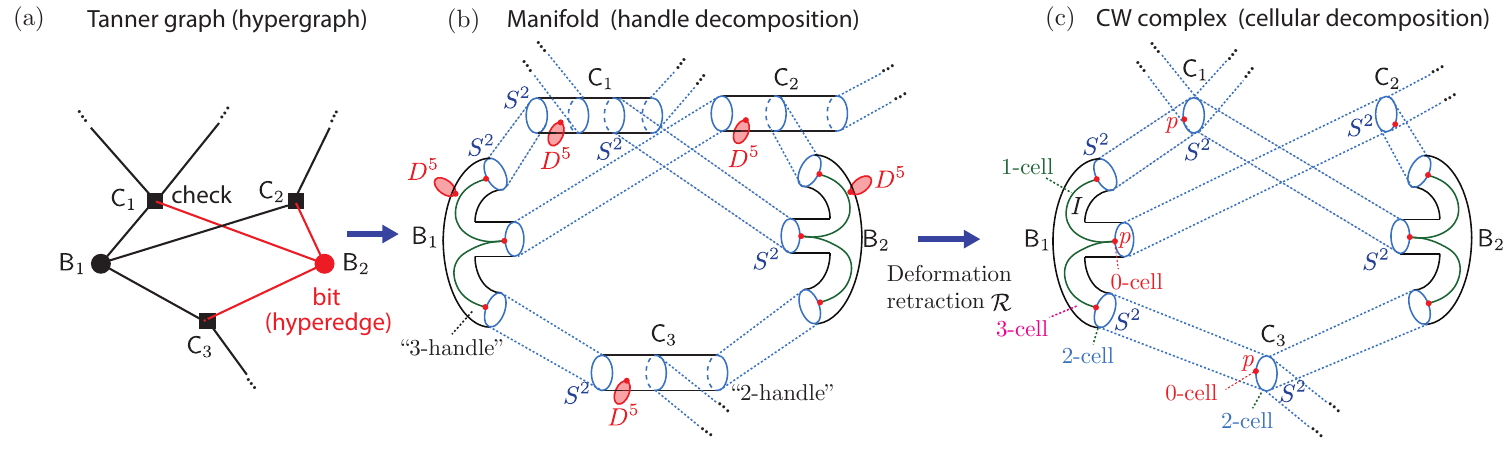}
    \caption{
(a,b) Mapping the input Tanner graph to a handle decomposition of a 3-handlebody \(H\) for an \(r\)-manifold, with cells shown on the dressed cores of the handles.
(c) Deformation retraction to a CW complex. Each ``3-handle'' retracts to its dressed core, yielding a 3-cell with one 3-cell, \((f(i)-1)\) 1-cells, and \(f(i)\) \(S^{2}\) boundary components, each equipped with a 0-cell \(p\).
Each ``2-cell'' contains a single \(S^{2}\) and a 0-cell \(p\), identified with the corresponding \(S^{2}\) boundaries and 0-cells on neighboring 3-cells.
A single 2-cell may attach to more than two 3-cells, demonstrating that the resulting CW complex is not a discretization of a manifold.}
    	\label{fig:CW_construction}
\end{figure*}

We have hence obtained a classical code defined on the triangulation $\L_\Delta$ of the 8-manifold $\M^8$ from the skeleton classical code $\bar\cC$, with bits and checks defined on the 3- and 2-simplices respectively.

\subsection{Deformation retraction to the hidden Poincar\'e CW complex}
In order to further simplify the construction and lower the constant overhead, we show that one can \textit{deformation retract} the manifold $\cM^r$ to a \textit{CW complex} $\cL_c$  \cite{Hatcher:2001ut} as hinted in Ref.~\cite{freedman:2020_manifold_from_code} and elaborated in Ref.~\cite{guemard2025lifting} and \cite{zhu2025topological}.  

We have the following formal definition for a CW complexes (also called a cellular complex):
\begin{definition}
A CW complex is built inductively by attaching 
$k$-cells (copies of open 
$k$-balls $D^k$) via attaching maps from their boundaries 
$S^{k-1}$ into the $(k-1)$-skeleton with the following conditions:
\begin{enumerate}
\item Closure-finite (C): The closure of each cell meets only finitely many other cells.
\item	
Weak topology (W): A set is closed if and only if its intersection with each cell-closure is closed.
\end{enumerate}
\end{definition}
\nin Note that simplicial complexes are special type of CW complexes where the cells are all simplices.

The CW complex we obtain is isomorphic to the handle chain complex in Eq.~\eqref{eq:long_chain_CW}, i.e., $\cL_c \cong \cL_h$. Under the deformation retraction $\cR$, each $k$-handle is retracted to its core---the $k$-cell $D^k$, i.e.,
\begin{equation}
  \cR:  h_k=D^k \times D^{r-k} \rightarrow D^k
\end{equation}
The dressed ``$k$-handle" $\tilde{h}_k$ is retracted to the dressed core $N_k$ as a dressed ``$k$-cell", i.e.,
\begin{equation}
 \cR:   \tilde{h}_k=N_k \times D^{r-k} \rightarrow N_k.
\end{equation}

Therefore, each dressed ``2-handle" associated with a check $\mathsf{C}_j$ is retracted to its dressed core: the dressed ``2-cell" $N_2=S^2$. This is illustrated by Fig.~\ref{fig:retraction_2-handle}, where we show in the right panel of (a) that the ``2-handle" $\mathsf{C}_j=S^2 \times I \times D^5$ can be viewed as a 3D shell $S^2 \times I$ thickened by $D^5$ along extra dimensions.  As we know above, the dressed ``2-handle" can be decomposed into one 2-handle $h_2=D^2 \times (I \times D^5)$ and one 0-handle $h_0=0 \times D^8$, corresponding to the upper and lower thickened half-shells respectively.  The attaching region of the 2-handle is a thickened annulus $S^1 \times I \times D^5$ highlighted in purple, which is then attached to the boundary of the 0-handle.  Now we apply the deformation retraction $\R$, as illustrated in Fig.~\ref{fig:retraction_2-handle}(b,c).  The 0-handle (lower thickened half-shell) is retracted to a single 0-cell (vertex) $p$, while the 2-handle (upper thickened half-shell) is retracted to a 2-cell $D^2$.  Overall, the dressed ``2-handle" is retracted to a dressed ``2-cell" $S^2$, which can be considered as a 0-cell $p$ and a 2-cell $D^2$ attached to the 0-cell (equivalent to identifying the boundary of $D^2$ into a single point that gives rise to $S^2$), as illustrated in Fig.~\ref{fig:retraction_2-handle}(c).  

Next, each dressed ``3-handle" associated with a bit $\mathsf{B}_j$ is retracted to its dressed core:  a dressed ``3-cell" $N_3=(S^3 \backslash \sqcup_{m=1}^{f(i)}D^3_m)$, as illustrated in Fig.~\ref{fig:retraction_3-handle}(c).  Throughout this procedure, each 1-handle
\(h_1 = I \times (D^{2} \times D^5)\) is retracted onto its core \(I\) (shown in green),
which becomes a single 1-cell (edge). At the same time, the 3-handle
\(h_3 = D^{3} \times D^{5}\) is retracted to a single 3-cell \(D^{3}\),
corresponding to the complement in the 3-sphere \(S^{3}\) of the
\((f(i)-1)\) 1-cells together with the \(f(i)\) embedded 3-balls \(D^{3}\).  The handle decomposition of the dressed ``3-handle" hence gives rise to the cell decomposition of the dressed ``3-cell":  each ``3-cell" can be further decomposed into $(f(i)-1)$ 1-cells and one 3-cell $D^3$ as illustrated in Fig.~\ref{fig:retraction_3-handle}(c).   

We now summarize our entire construction in Fig.~\ref{fig:CW_construction}.   We start with a Tanner graph  (hypergraph) of the skeleton classical code $\bar{\C}_{\text{c}}$.  We then use the handle construction to map the Tanner graph to a manifold $\M^8$ as illustrated in (b).  Finally, we apply the deformation retraction $\R$ which retracts the dressed ``2-handles" $S^2 \times D^5$ to ``2-cell" $S^2$ with one 0-cell $p$ placed on it, and retracts the dressed ``3-handle" $N_3 \times D^5$ to its dressed core $N_3$ as a dressed ``3-cell" composed of one 3-cell $D^3$ and $(f(i)-1)$ 1-cells, as illustrated in (c).  Note that the boundary of the ``3-cell" are  $f(i)$ ``2-cells" $S^2$ which are identified with the neighboring ``2-cells" corresponding to the connected checks. 

Due to the isomorphism between a portion of the CW complex involving $C_2$ and $C_3$ [Eq.~\eqref{eq:long_chain_CW}] and the classical code $\bar{\cC}$ [Eq.~\eqref{eq:classical_chain complex}], the classical code defined on the CW complex is completely the same as the skeleton code $\bar{\cC}$, since we only need 3-cells and 2-cells to define bits and checks. Meanwhile, the hidden CW complex structure (including cells other than dimension 2 and 3) has mathematically well-defined cup product which can be used to define the path integrals as topological invariants as will be introduced in Sec.~\ref{sec:Feynman_path_integral}. Note that although there additional cell structures in the CW complex, the information about them is already contained in the input Tanner graph (hypergraph) of the skeleton classical code.  For example, each 0-cell $p$ (\red{red}) is associated with a check node in the Tanner graph.  The 1-cells (\green{green}) come from dividing the hyper-edge in the hypergraph of the classical code into individual edges, and the minimal number of the edges is $f(i)-1$. The illustration in the rightmost panel of Fig.~\ref{fig:retraction_3-handle}(c) containing $f(i)$ divided edges is even more intuitive to understand.   Note that the hypergraph-product code $\cC$ obtained from tensor product of the CW complexes as will be introduced in Sec.~\ref{sec:non-Abelian_LDPC} is a non-topological code,  since $\cL_c$ is no longer a discretization of the manifold like the triangulation. This compact realization makes the near-term implementation practical and has minimal overhead.    Following the convention in Ref.~\cite{zhu2025topological}, we call the classical code defined on the CW complex $\cL_c$ the thickened classical code from now on.

Interestingly, the isomorphism between $\cL_c$ and $\cL_h$ preserves the Poincar\'e duality, i.e., and $\cL_c$ is hence a \textit{Poincar\'e complex}, which is defined below:
\begin{definition}
An $r$-dimensional CW complex is a Poincar\'e complex if and only if there exists a Poincar\'e duality isomorphism 
\be
H_k(\cL_c) \cong H^{r-k}(\cL^*_c),
\ee
where $\L^*_c$ is the dual CW complex.
\end{definition}
\nin Note that the property of the Poincar\'e complex makes it much easier to define a path integral as a topological invariant as opposed to the general CW or simplicial complexes, as will be discussed in Sec.~\ref{sec:Feynman_path_integral}.

We have the following lemma:
\begin{lemma} \cite{zhu2025topological}
From the Tanner graph of any input skeleton classical LDPC code $\bar{\cC}=\text{Ker}(\partial)=\text{Ker}(\Hs)$ that is asymptotically good, one can define the same asymptotically good classical LDPC code on an $r$-dimensional CW complex $\cL_c$ (with bits and checks placed on the 3- and 2-cells) (for $r\ge 8$), with the code space being the 3rd homology group $H_3(\cL_c; \ZZ_2)$.
\end{lemma}
\begin{proof}
As pointed out above, there exists an isomorphism between the chain complex $X$ of the $[\bar{n}, \bar{k}, \bar{d}]$ skeleton classical code $\bar{\C}_{\text{c}}$ and the portion in the CW complex $\L_c$ involving 2-cells and 3-cells.  In particular, there is an isomorphism between the $\ZZ_2$ homology groups of $X$ and $\L_c$:
\be
H_1(X; \ZZ_2) \cong H_3(\L_c; \ZZ_2),  \quad 
H_0(X; \ZZ_2) \cong H_2(\L_c; \ZZ_2).
\ee
The dimension $\bar{k}$ of the skeleton  classical code $\bar{\C}_{\text{c}} $$=$$H_1(X; \ZZ_2)$ is equal to the first $\ZZ_2$ Betti number $\bar{b}_1$ of $X$ which in turn equals the third Betti number $b_3$ of $\L_c$ and hence the dimension $k$ of the $[n, k, d]$ thickened classical code $\C= H_3(\L_c; \ZZ_2)$ defined on $\L_c$, i.e.,
\be
\bar{k} = \bar{b}_1 \equiv \dim(H_1(X; \ZZ_2)) = b_3  \equiv  \dim(H_3(\L_c; \ZZ_2)) = k.
\ee
Since the skeleton code $\bar{\C}_{\text{c}}$ is good, it has linear dimension $\bar{k}=\Theta(\bar{n})$.  Due to the isomorphism between $X$ and the portion of $\L_c$, we have $\bar{n}=n $. Therefore, we also have linear dimension or equivalently constant rate for the thickened classical code $\C$ defined on $\L_c$: $k=\Theta(n)$. 

Moreover, due to the isomorphism of chain complexes, we have the following mapping between any 1-cycle representative $\bar{\as}_1 \in H_1(X;\ZZ_2)$ in $X$ and a 3-cycle representative  $\as_3 \in H_3(\L_c; \ZZ_2)$:
\be
\bar{\as}_1 \rightarrow \as_3,
\ee
where the mapping preserves the Hamming weight of the $\ZZ_2$ cycle, i.e.,
\be\label{eq:Hamming_weight_preservation}
|\bar{\as}_1|=|\as_3|.
\ee
The asymptotically good skeleton classical code $\bar{\C}_{\text{c}}$ has linear distance:
\be
\bar{d} = \min \{ |\bar{\as}_1|: \ \bar{\as}_1 \neq 0 \in H_1(X; \ZZ_2) \} = \Theta(\bar{n}).
\ee
The distance of the thickened classical code $\C$ is defined as
\be
d =\min \{ |\as_3|: \ \as_3 \neq 0 \in H_3(\L_c; \ZZ_2) \}. 
\ee
According to Eq.~\eqref{eq:Hamming_weight_preservation}, we have the distance of the two codes equaling to each other:
\be
d=\bar{d} = \Theta(\bar{n}) = \Theta(n).
\ee
The thickened code $\C$ hence have both constant rate and linear distance, and is hence asymptotically good.
\end{proof}

Similarly, if the transposed code of the input skeleton code (also called a 0-cocycle code)  $\bar{\C}^*_c=\text{Ker}(\mathsf{H^T})= \text{Ker} \  d$ (where $d \equiv\partial^T$ is the coboundary map) is asymptotically good, then we also get a good classical code defined on the CW complex $\L_c$ with the code space being the 2rd cohomology group $H^2(\L_c; \ZZ_2)$.

We emphasize that this scheme is applicable to arbitrary skeleton classical code without requiring any specific structure.

Finally, as pointed out in Ref.~\cite{freedman:2020_manifold_from_code} and also discussed in Refs.~\cite{zhu2025topological, zhu2025transversal}, there exist spurious 1-cycles $\fs_1$ and 1-cocycles $\fs^1$, as well as Poincar\'e dual 7-cocycles ${\fs^*}^7 \equiv \text{PD}(\fs_1)$ and dual 7-cycles  $\fs^*_7 \equiv \text{PD}(\fs^1)$ in the manifold and CW-complex construction. These spurious cycles and cocycles could have very small minimal Hamming weight, as small as $O(1)$. For example, one can see from Fig.~\ref{fig:CW_construction}(c) that there can be short 1-cycles supported on the green edges (1-cells), which is not controlled by the distance of the skeleton classical code.   
When taking tensor product of the CW complexes, there could be short logical (co)cycles due to the spurious (co)cycles.   So one needs to use the subsystem-code idea developed in Refs.~\cite{zhu2025topological, zhu2025transversal} to treat these short logical (co)cycles as the logicals of gauge qubits and ignore them.  We will further elaborate on this point in Sec.~\ref{sec:rate_and_distance}.

\subsection{Define cup product}\label{sec:cup_product}
As mentioned above, the main motivation to map the classical code to the CW complex is to define the spacetime path integral as a topological invariant.  The type of TQFT we are developing in this paper is a higher-form analog of the Dijkgraaf-Witten twisted gauge theory, which uses cup product in its path integral.   Meanwhile, the stabilizer code corresponding to this type of gauge theory can also be conveniently represented using cup product.  

It is well-known that cup product is well-defined for the CW complex \cite{Hatcher:2001ut}. We now introduce the cup product `$\cup$' as a bilinear map on the cochain groups:
\be
\cup :  C^p(\L_c) \times C^q(\L_c) \rightarrow C^{p+q}(\L_c).
\ee
Here, $C^p$ represents the $p^\text{th}$ cochain group.
This can be interpreted as the cup product between a $p$-cochain $\alpha^p \in C^{p}$ and a $q$-cochain $\beta^q \in C^{q}$ gives rise to a $(p+q)$-cochain  $\alpha^p \cup \beta^q \in C^{p+q}$.  Moreover, the cup product  induces a bilinear operation on the cohomology groups as well:
\be\label{eq:bilinear}
\cup :  H^p(\L_c) \times H^q(\L_c) \rightarrow H^{p+q}(\L_c).
\ee

We now consider how to explicitly evaluate the cup product.  In the special case of a simplicial complex $\L_{\Delta}$, we have the following explicit formula for the evaluation on a $(p+q)$-simplex $[v_0,v_1,\cdots,v_{p+q}]$ as \cite{Hatcher:2001ut} 
\begin{align}\label{eq:cup_def}
 & (\alpha^p \cup \beta^q)([v_0,v_1,\cdots, v_{p+q}]) \cr
=&\alpha^p([v_0,v_1,\cdots, v_{p}])\beta^q([v_p, v_{p+1},\cdots, v_{p+q}])~. 
\end{align}
Here, we can choose an arbitrary global ordering for vertices $v_i$ on the entire complex $\L_{\Delta}$.   Then on each ($p+q$)-simplex we have the ordering $v_0 < v_1 < v_2 \cdots < v_{p+q}$, which induces an orientation of the simplex and specifies how to pick the $p$-simplices and $q$-simplices in the above evaluation.

Now for a general CW complex, we can subdivide the CW complex $\L_c$ into a simplicial complex $\L_{\Delta}$ and then use Eq.~\eqref{eq:cup_def} to evaluate the cup product. One can hence define the qubit models on the refined (subdivided) complex $\L_{\Delta}$ by introducing a constant number of extra qubits.  Alternatively, one can still define the qubit models on the original CW complex $\L_c$, while using subdivision to simplicial complex  $\L_{\Delta}$ to derive the cup product evaluation on general cells.  The detailed discussion of this can be found in Ref.~\cite{zhu2025topological}.   We also note that in some cases it is also possible to directly derive the cup product evaluation on the CW complex, as will be disccused in Sec.~\ref{sec:evaluation}.

\section{Non-Abelian qLDPC codes as twisted gauge theories}\label{sec:non-Abelian_LDPC}

In this section, we construct the non-Abelian qLDPC codes along with the corresponding combinatorial TQFTs.  We also study the detailed properties of these codes such as their Clifford stabilizer and logical operator form, the underlying symmetries, and the code parameters including encoding rate and distance.

\subsection{Feynman path integral in spacetime}\label{sec:Feynman_path_integral}
We start with the following twisted $\Z_2^3$ gauge theory with action:
\begin{align}\label{eq:non-Abelian_action}
S=& \pi \int_{\tilde{\L}} \red{ a^p} \cup \red{d \tilde{a}^u } +  \pi \int_{\tilde{\L}} \blue{b^q} \cup \blue{ d\tilde{b}^v} + \pi \int_{\tilde{\L}} \green{c^s} \cup \green{d\tilde{c}^w} \cr
&+ \pi \int_{\tilde{\L}} \red{a^p} \cup \blue{b^q} \cup \green{c^s},
\end{align}
where the first three terms come from the BF theory, while the last term is the higher-form generalization of the Dijkgraaf-Witten twisted gauge theories. Here, we have $p+q+s=D$ and $p+u+1=q+v+1=s+w+1=D$, where $D=\mathsf{d}+1$ is the total spacetime dimension and $\mathsf{d}$ represents the spatial dimension. Here, $\tilde{\L}$ represents the closed space-time complex either built form the triangulation $\L_\Delta$ of the manifold constructed from codes or the CW complex $\L_c$.   

\begin{figure}[t]
    \centering
    \includegraphics[width=1\linewidth]{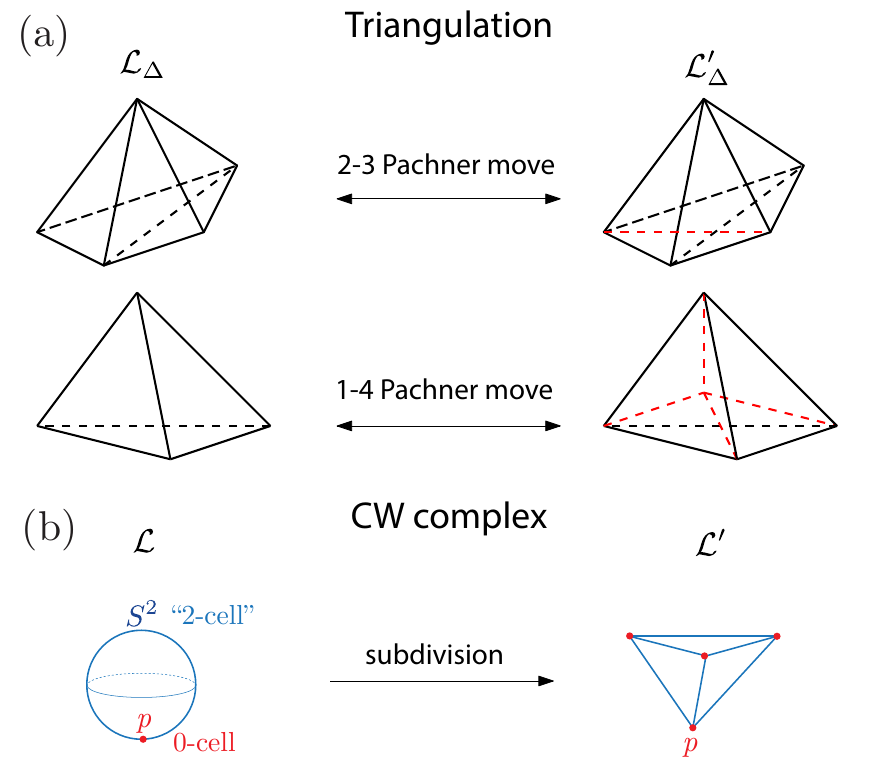}
    \caption{(a) Example of Pachner moves that change the triangulation. (b) Illustration on the subdivision of the cells in a CW complex. Here, an $S^2$ attached to a vertex (0-cell) is subdivided into the surface of a tetrahedron with four vertices. }
    \label{fig:pachner}
\end{figure}

We now consider the \textit{path integral} $\mathcal{Z}[\tilde{\L}]$, which is also called the \textit{partition function}.
For the code (ground) space $\C$ of the  twisted $\ZZ_2^3$ gauge theory, the path integral is related to the action by
\be
\mathcal{Z}[\tilde{\L}] = \frac{1}{\N} \sum_{\red{a^p}, \blue{b^q}, \green{c^s}, \red{\tilde{a}^u}, \blue{\tilde{b}^v}, \green{\tilde{c}^w} \in {C^{\empty}}^{\bullet}(\tilde{\L})} e^{iS[\red{a^p}, \blue{b^q}, \green{c^s}, \red{\tilde{a}^u}, \blue{\tilde{b}^v}, \green{\tilde{c}^w}]},
\ee
where $\N$ is a normalization constant. 
The first three  BF terms in $S$ [Eq.~\eqref{eq:non-Abelian_action}] gives a factor of $1$ contribution and imposes the cocycle condition:
\be
\red{da^p}=\blue{db^q}=\green{dc^s}=0, 
\ee
which physically means that $\red{a^p},  \blue{b^q}, \green{c^s}$ are flat gauge fields.   We hence reach the following form of the path integral:  
\be\label{eq:path_inegral_cup}
\mathcal{Z}[\tilde{\L}] = \sum_{\red{a^p}, \blue{b^q}, \green{c^s} \in {H^{\empty}}^{\bullet}(\tilde{\L})} (-1)^{\int_{\tilde{\L}} \red{a^p} \cup \blue{b^q} \cup \green{c^s}}, 
\ee
where for convenience we have set the normalization constant as $\N=1$ since it is irrelevant for the main results in this paper.
  One can re-express it using the Poincar\'e dual cycles (membranes) $\red{a^*_{D-p}} = \text{PD}(\red{a^p}), \blue{b^*_{D-q}} = \text{PD}(\blue{b^q}), \green{c^*_{D-s}}=\text{PD}(\green{c^s})$:
\be\label{eqn:pathintegralcycle}
\mathcal{Z}[\tilde{\L}] =\sum_{\red{a^*_{D-p}}, \blue{b^*_{D-q}}, \green{c^*_{D-s}} \in {H_\bullet}(\tilde{\L})} (-1)^{| \red{a^*_{D-p}} \cap \blue{b^*_{D-q}} \cap \green{c^*_{D-s}}|},
\ee
where we have used the geometric interpretation of triple cup product as the $\ZZ_2$ tripe intersection number  between the Poincar\'e dual cycles: $| \red{a^*_{D-p}} \cap \blue{b^*_{D-q}} \cap \green{c^*_{D-s}}|$ \cite{zhu2023non}.  Note that $\red{a^*_{D-p}}$,  $\blue{b^*_{D-q}}$ and $\green{c^*_{D-s}}$ are cycles (membranes) that the magnetic logical operators are supported.  

We have the following lemma:
\begin{lemma}\label{lemma:invariant} 
The path integral $\mathcal{Z}[\tilde{\L}]$ defined on the Poincar\'e CW complex $\tilde{\L}$ is a cohomology invariant which physically corresponds to the gauge invariance.  This means $\mathcal{Z}[\tilde{\L}]$ is invariant under the deformation by a coboundary or equivalently the gauge transformation on the gauge fields:  $\red{a^p} \rightarrow \red{a^p} + d\tilde\eta^{p-1}$, $\blue{b^q} \rightarrow \blue{b^q} + d \tilde \xi^{q-1}$, $\green{c^s} \rightarrow \green{c^s} + d\tilde\zeta^{s-1}$.
\end{lemma}
\begin{proof}
The cup product applies the following trilinear map on the cohomology group:
\be\label{eq:trilinear_map}
\cup : H^{p}(\tilde \L) \times H^{q}(\tilde \L) \times H^{q_s}(\tilde \L) \rightarrow H^{p+q+s}(\tilde \L).
\ee
We hence have 
\begin{align}
 & (\red{a^p} + d\tilde\eta^{p-1} )\cup (\blue{b^q} +d\tilde \xi^{q-1} )\cup (\green{c^s} + d\tilde\zeta^{s-1} ) \cr
=&  \red{a^p} \cup \blue{b^q} \cup \green{c^s} + d\omega^{D-1}.   
\end{align}
Therefore, the exponent in the path integral becomes:
\begin{align}
& \int_{\tilde \L} {(\red{a^p} + d\tilde\eta^{p-1} )\cup (\blue{b^q} +d\tilde \xi^{q-1} )\cup (\green{c^s} + d\tilde\zeta^{s-1} )}   \cr
=& \int_{\tilde \L}  \red{a^p} \cup \blue{b^q} \cup \green{c^s} +  \int_{\tilde \L} d\omega^{D-1} \cr
=& \int_{\tilde \L}  \red{a^p} \cup \blue{b^q} \cup \green{c^s},
\end{align}
where in the second equality we have used the Stoke's theorem: 
\be
\int_{\tilde \L} d\omega^{D-1} = \int_{ \partial 
\tilde{\L}}\omega^{D-1}=0.
\ee
In the above identity, we have use the property that the Poincar\'e CW complex has no boundary, i.e., $\partial 
\tilde{\L}=0$. 
\end{proof}

In the above proof, we can see the reason why we need to define the combinatorial TQFT on a Poincar\'e CW complex, since the path integral is still a topological invariant.  On the other hand, for a general CW complex $\tilde{\L'}$ which can have non-trivial boundary $\partial \tilde{\L'} \neq 0$, the path integral will no longer be a topological invariant.  

Moreover, an important property of TQFT defined on the triangulation of a manifold $\tilde\L_\Delta $ is that the path integral is invariant under re-triangulation  $\tilde\L_\Delta \rightarrow \tilde\L'_\Delta$ generated by the  Pachner moves illustrated in Fig.~\ref{fig:pachner}(a).  This property forbids one to directly define a TQFT on general chain complexes such as the product complexes obtained from expander graphs.  Nevertheless, there is a similar property which holds for the Poincar\'e CW complex:
\begin{lemma}\label{lemma:subdivision}
The path integral $\mathcal{Z}[\tilde{\L}]$ defined on the Poincar\'e CW complex $\tilde{\L}$ is invariant under subdivision. 
\end{lemma}
\begin{proof}
This simply follows from the fact that the cup product sum or equivalently the triple intersection is invariant under subdivision.
\end{proof}
We also illustrate the local subdivision of a CW complex from a sphere $S^2$ attached to a vertex $p$ (appeared previously in Fig.~\ref{fig:retraction_2-handle}) to a tetrahedron with 4 vertices in Fig.~\ref{fig:pachner}(b).

\subsection{Explicit construction of models with non-trivial triple intersections}\label{sec:triple_intersection_construction}

We start from a hypergraph-product code formed by the tensor product of two good 0-cocycle classical LDPC codes $\bar{\C}^*_c=\text{Ker}(d) =\text{Ker}(\Hs^T)$  corresponding to the product chain complex $ X_x \otimes X_y$, where $X_x=X_y$ represent the chain complex associated with the skeleton classical code $\bar{\C}_{\text{c}}$.  Using the mapping from the classical codes to the CW complexes: $X_x \rightarrow
\L_x$, $X_y \rightarrow \L_y$,  we obtain a product complex:
\be
\L = \L_x \otimes \L_y,  
\ee
where $\L_x$ and $\L_y$ are 8D CW complexes and $\L$ is a 16D CW complex.  The total spacetime complex $\tilde{\L} = \L \otimes I_t$ is 17D, where $I_t$ represents the time complex, i.e., a 1D line.  

We will consider the twisted gauge theory corresponding to the following path integral (with $p=8$, $q=6$ and $s=3$): 
\be
\mathcal{Z}[\tilde{\L}] = \sum_{\red{a^8}, \blue{b^6}, \green{c^3} \in {H^{\empty}}^{\bullet}(\tilde{\L})} (-1)^{\int_{\tilde{\L}} \red{a^8} \cup \blue{b^6} \cup \green{c^3}}.   
\ee
The reason for this particular choice is explained  as follows.
We first choose a cocycle basis for the gauge fields $\red{a^8}$, $\blue{b^6}$ and $\green{c^3}$ respectively, which are denoted by $\{\red{\tilde{\alpha}^8}\}$, $\{\blue{\tilde{\beta}^6}\}$ and $\{\green{\tilde{\gamma}^3}\}$. Here and throughout the paper, we use the tilde $\tilde{\cdot}$ to indicate \textit{spacetime} basis cycles and cocycles, while those without tilde represent \textit{space} the basis cycles and cocycles. 
We then consider the following decomposition of the basis cocycles into cocycles within each factor complex built from a classical code due to the K\"unneth formula: 
\begin{align}\label{eq:cocycle_decomposition}
\red{\tilde{\alpha}^{8}} =&  \bs^2 \otimes {\bs'^*}^6 \otimes  \ts^0 \equiv \red{{\alpha}^{8}} \otimes \ts^0   \cr
\blue{\tilde{\beta}^{6}} =& {\bs^*}^6  \otimes  \cs^0 \otimes     \ts^0 \equiv  \blue{{\beta}^{6}} \otimes \tau^0 \cr
\green{\tilde{\gamma}^{3}} =& \cs^0 \otimes  \bs'^2  \otimes {\ts^*}^1 \equiv \green{\gamma^2} \otimes {\ts^*}^1,  
\end{align}
where $\red{\alpha^8}$, $\blue{\beta^6}$ and $\green{\gamma^2}$ are the spatial components.

Note that $\bs^2$ and ${\bs^*}^6$ are a pair of dual cocycles (also due to Poincar\'e duality) that have non-trivial cup product and intersection, i.e.,  
\begin{equation}\label{eq:non-trivial_cup}
    \bs^2 \cup {\bs^*}^6 \neq 0 \in H^8(\cL; \ZZ_2),    \qquad    \int_{\cL_x} \bs^2 \cup {\bs^*}^6 = |\bs^*_6 \cap \bs_2|=1,
\end{equation}
where the second equation represents the sum of the evaluation of $\bs^2 \cup {\bs^*}^6 (\sigma_8)$ on all the 8-cells $\sigma_8$ of $\cL_x$ and corresponds to the intersection of the corresponding Poincar\'e  dual cycles $\bs^*_5$ and $\as_2$.   Now we can further obtain the triple cup product within each  classical code in the tensor product:
\begin{equation}
    \bs^2 \cup {\bs^*}^6 \cup \cs^0 \neq 0 \in H^8(\cL; \ZZ_2),    \qquad    \int_{\cL_x} \bs^2 \cup {\bs^*}^6 \cup \cs^0 =1,
\end{equation}
which gives rise to a non-trivial triple intersection.   Similarly, we have Poincar\'e duality on the time complex $I_t$: 
\begin{equation}
    \ts^0 \cup {\ts^*}^1 \neq 0 \in H^1(I_t; \ZZ_2),    \qquad    \int_{I_t} \ts^0 \cup {\ts^*}^1 = |\ts^*_1 \cap \ts_0|=1,
\end{equation}
where $\ts^*_1$ and $\ts_0$ are the Poincar\'e dual cycles of $\ts^0$ and $\ts^*_1$ respectively.  This can be extended to the non-trivial  triple intersection: 
\begin{equation}
    \ts^0 \cup \ts^0 \cup {\ts^*}^1   \neq 0 \in H^1(I_t; \ZZ_2),    \qquad    \int_{I_t} \ts^0 \cup \ts^0 \cup {\ts^*}^1  =1.
\end{equation}
Now due to the K\"unneth formula for cup product, we can re-express the triple cup product between three-cocycles as the tensor product of the triple cup product within each thickened classical code $\cL_x$ and $\cL_y$, as well as the temporal complex $I_t$:
\begin{align}\label{eq:triple-cup_decomposition}
    & \red{\tilde \alpha^8} \cup \blue{\tilde{\beta}^{6}} \cup \green{\tilde{\gamma}^3} \cr
    =& (\bs^2 \otimes {\bs'^*}^6 \otimes  \ts^0    
) \cup ( {\bs^*}^6 \otimes  \cs^0 \otimes     \ts^0 ) \cup
( \cs^0 \otimes  \bs'^2  \otimes {\ts^*}^1 )   \cr
=& (\bs^2 \cup {\bs^*}^6 \cup  \cs^0    
) \otimes ({\bs'^*}^6 \otimes  \cs^0 \otimes     \bs'^2 ) \otimes
( \ts^0 \cup  \ts^0  \cup {\ts^*}^1 ) \cr
\neq & 0 \in H^{17}(\tilde{\cL} ; \ZZ_2).
\end{align}
Note that the triple cup product $\red{\tilde{\alpha}^{8}} \cup \blue{\tilde{\beta}^{6}} \cup \green{\tilde{\gamma}^3}$ is non-trivial since each tensor component is non-trivial.  We hence have the non-trivial triple intersections:
\begin{align}\label{eq:triple_cup_sum}
   & \int_{\tilde{\cL}}  \red{\tilde{\alpha}^{8}} \cup \blue{\tilde{\beta}^{6}} \cup \green{\tilde{\gamma}^3}  \cr
=&  \int_{\cL_x}(\bs^2 \cup {\bs^*}^6 \cup  \cs^0    
) \cdot \int_{\cL_y} ({\bs'^*}^6 \otimes  \cs^0 \otimes     \bs'^2 )   \cr
&\cdot \int_{I_t} ( \ts^0 \cup  \ts^0  \cup {\ts^*}^1 ) \cr
=& 1 \cdot 1 \cdot 1 =1.
\end{align}

We note that the spatial component of these cocycles also have the non-trivial triple intersections:
\begin{align}
    \int_{{\cL}}  \red{{\alpha}^{8}} \cup \blue{{\beta}^{6}} \cup \green{{\gamma}^2} =1.
\end{align}

\subsection{Clifford stabilizer codes on Poincar\'e complexes}\label{sec:Clifford_stabilizer_codes}

Here, we present the explicit construction of the non-Abelian qLDPC codes, which belong to the family of Clifford stabilizer codes.  We then discuss the properties of this family of codes, including the stabilizer commutation relations, logical operators, and charge parity operators.

\subsubsection{Construction of the codes}\label{sec:code_construction}

We present a family of Clifford stabilizer codes with the code space represented by $\tilde{\C}$. The derivation of this model from gauging the higher-form symmetry-protected topological (SPT) phase is shown in details in  Appendix \ref{app:gauging_derivation}, which is a generalization of the higher-form twisted gauge theory (called \textit{cubic theory}) in Ref.~\cite{Hsin:2024nwc} to the context of CW complexes.   The Clifford stabilizer group $\tilde{\mathcal{S}}$ is a subgroup of the $n$-qubit Clifford group $Cl_n$, i.e., $\tilde{\mathcal{S}} < Cl_n $, and is generated as follows
\be
\tilde{\mathcal{S}} =\langle  \{ \tilde{A}^{\rd}_{\sigma_{p-1}},  \tilde{A}^{\bl}_{\sigma_{q-1}}, \tilde{A}^{\gr}_{\sigma_{s-1}},  B^{\rd}_{\sigma_{p+1}},  B^{\bl}_{\sigma_{q+1}}, B^{\gr}_{\sigma_{s+1}} \} \rangle, 
\ee
where $\tilde{A}$ represents the dressed X stabilizer generators corresponding to the Gauss's law term $H_\text{Gauss}$ and $B$ represents the $Z$-stabilizer generators corresponding to the flux term $H_\text{Flux}$.   Similar to the Pauli stabilizer code, we can define the code space $\tilde{\C}$ of the $n$-qubit  Clifford stabilizer code with Clifford stabilizer group $\tilde{\mathcal{S}} $ as:
\be\label{eq:code_space_Clifford}
\tilde{\C}(\tilde{\mathcal{S}}) = \{ \ket{\psi} \in (\CC^2)^{\otimes n}:  g \ket{\psi} = \ket{\psi} \ , \forall g \in \tilde{\mathcal{S}} \}. 
\ee
We now present the stabilizer generators as follows:    
\begin{align}\label{eq:stabilizer_summary}
\tilde{A}^{\rd}_{\sigma_{p-1}} =& A^{\rd}_{\sigma_{p-1}} \prod_{\sigma_q,\sigma_s: \int\tilde{\sigma}_{p-1}\cup \tilde{\sigma}_q\cup \tilde{\sigma}_s\neq 0}\text{CZ}^{\bl,\gr}_{\sigma_q, \sigma_s},  \cr
\tilde{A}^{\bl}_{\sigma_{q-1}} =& A^{\bl}_{\sigma_{q-1}}\prod_{\sigma_p,\sigma_s: \int \tilde{\sigma}_p\cup \tilde{\sigma}_{q-1}\cup  \tilde{\sigma}_{s} \neq 0}\text{CZ}^{\rd, \gr}_{\sigma_p, \sigma_s}, \cr
 \tilde{A}^{\gr}_{\sigma_{s-1}} =& A^{\gr}_{\sigma_{s-1}}\prod_{\sigma_p,\sigma_q: \int \tilde{\sigma}_p\cup \tilde{\sigma}_q\cup \tilde{\sigma}_{s-1} \neq 0}\text{CZ}^{\rd,\bl}_{\sigma_p, \sigma_q},  \cr
B^{\rd}_{\sigma_{p+1}}=& \prod_{\sigma_{p}\subset \partial \sigma_{p+1}}Z_{\sigma_p}^{\rd}, \quad  B^{\bl}_{\sigma_{q+1}}=\prod_{\sigma_{q}\subset \partial \sigma_{q+1}}Z_{\sigma_q}^{\bl}   \cr 
B^{\gr}_{\sigma_{s+1}} =& \prod_{\sigma_{s}\subset \partial \sigma_{s+1}}Z_{\sigma_{s}}^{\gr},
\end{align}
where $\tilde{\sigma}_{k}$ is an indicator $k$-cochain that takes value 1 on a single $k$-cell $\sigma_k $ and zero otherwise.

In the above expression,  $A^{\rd}_{\sigma_{p-1}}$, $A^{\bl}_{\sigma_{q-1}}$ and $A^{\gr}_{\sigma_{s-1}}$ are the  $X$-stabilizers of the corresponding Pauli stabilizer codes corresponding to the untwisted $\red{\Z_2^{(p)}}$$\times$$\blue{\Z_2^{(q)}}$$\times$$\green{\Z_2^{(s)}}$ gauge theories, which are three independent copies of higher-form $\ZZ_2$ qLDPC codes defined on the CW complex $\L$, i.e., $p$-form, $q$-form and $s$-form respectively, with their qubits defined on $p$-cells, $q$-cells and $s$-cells respectively.   These $X$-stabilizers have the following form:
\begin{align}\label{eq:X-stabilizers_untiwsted}
{A}^{\rd}_{\sigma_{p-1}} =& \prod_{\sigma_{p-1}\subset \partial \sigma_p} X^{\rd}_{\sigma_{p}},  \quad  {A}^{\bl}_{\sigma_{q-1}} = \prod_{\sigma_{q-1}\subset \partial \sigma_q} X^{\bl}_{\sigma_{q}}, \cr
{A}^{\gr}_{\sigma_{s-1}} =& \prod_{\sigma_{s-1}\subset \partial \sigma_s} X^{\gr}_{\sigma_{s}}. 
\end{align}
The corresponding untwisted  Pauli-stabilizer codes are associated with the standard Pauli-stabilizer group $\mathcal{S}$, which are subgroup of the $n$-qubit Pauli group $\mathcal{P}_n$, i.e.,  $\mathcal{S}< \mathcal{P}_n$. The Pauli-stabilizer group can be generated generated as:
\be\label{eq:stabilizer_untwisted}
\mathcal{S} =\langle  \{ {A}^{\rd}_{\sigma_{p-1}},  {A}^{\bl}_{\sigma_{q-1}}, {A}^{\gr}_{\sigma_{s-1}},  B^{\rd}_{\sigma_{p+1}},  B^{\bl}_{\sigma_{q+1}}, B^{\gr}_{\sigma_{s+1}} \} \rangle, 
\ee
where the $Z$-stabilizers $B^{\rd}_{\sigma_{p+1}},  B^{\bl}_{\sigma_{q+1}}, B^{\gr}_{\sigma_{s+1}}$ are completely the same as the twisted code in Eq.~\eqref{eq:stabilizer_summary}.  The corresponding untwisted codes space are defined as
\be
{\C}({\mathcal{S}}) = \{ \ket{\psi} \in (\CC^2)^{\otimes n}:  g \ket{\psi} = \ket{\psi} \ , \forall g \in {\mathcal{S}} \}. 
\ee

\subsubsection{Constraint of the code space from TQFT}
From the TQFT action $S$ in Eq.~\eqref{eq:non-Abelian_action}, we can derive the following equation of motions from the least action principle $\delta S=0$:
\begin{align}
\begin{split}
    & \red{da^p=0},\quad \blue{db^q=0}, \quad \green{dc^s=0} \cr
    & \red{d\tilde a^{\ds-p}}+\blue{b^q} \cup \green{c^s}=0, \\
    &\blue{d\tilde b^{\ds-q}} +\red{a^p} \cup \green{c^s}=0, \\
    &\green{d\tilde c^{\ds-s}}+\red{a^p} \cup \blue{b^q}=0.
    \end{split}
    \label{eq:EOM}
\end{align}
When summing over a subcomplex (cycle) $\tilde{\mathcal{V}}^*_{\ds-p+1} \equiv \tilde{\mathcal{V}}^*_{q+s}$ of the dual spacetime complex $\tilde{\L}^*$, one obtains
\begin{align}
    \int_{\tilde{\mathcal{V}}^*_{\ds-p+1}} \blue{b^q} \cup \green{c^s} :=& \int_{\tilde{\L}} \blue{b^q} \cup \green{c^s} \cup \tilde{\mathcal{V}}^p = \int_{\tilde{\mathcal{V}}^*_{\ds-p+1}} \red{d\tilde{a}^{\ds-p}} \cr
    =& \int_{\partial\tilde{\mathcal{V}}^*_{\ds-p+1}} \red{\tilde{a}^{\ds-p}}=0,
\end{align}
where $\tilde{\mathcal{V}}^p = \text{PD}(\tilde{\mathcal{V}}^*_{\ds-p+1})$ is the Poincar\'e dual cocycle of $\tilde{\mathcal{V}}^*_{\ds-p+1}$. 
We can also promote the gauge fields to operators  $\red{\hat{a}^p}, \blue{\hat{b}^q}$ and $\green{\hat{c}^s}$ defined in the Hilbert space $\H_{\L}$ on the space complex $\L$, which mathematically correspond to the operator-valued cochains (cocycles), i.e., they are the quantization of cochains. We hence obtain the following relations: 
\begin{align}\label{eq:constraint_operator}
\int_{\mathcal{V}^*_{\ds-p+1}} \blue{\hat{b}^q} \cup \green{\hat{c}^s} :=  \int_{\L} \blue{\hat{b}^q} \cup \green{\hat{c}^s} \cup \mathcal{V}^{p-1}  =  0,   \cr
\int_{\mathcal{V}^*_{\ds-q+1}} \red{\hat{a}^p} \cup \green{\hat{c}^s} :=  \int_{\L} \red{\hat{a}^p}  \cup \green{\hat{c}^s} \cup \mathcal{V}^{q-1}  =  0,   \cr 
\int_{\mathcal{V}^*_{\ds-s+1}} \red{\hat{a}^p} \cup \blue{\hat{b}^q} :=  \int_{\L} \red{\hat{a}^p}  \cup \blue{\hat{b}^q} \cup \mathcal{V}^{s-1}  =  0,   \cr 
\end{align}
where the later two can be simply obtained by permuting the colors.  Here, $\mathcal{V}^{k-1}$ is a $(k-1)$-cocycle in the space complex $\L$, with the relation to the corresponding spacetime cocycle being $\tilde{\mathcal{V}}^{k}= \mathcal{V}^{k-1} \otimes \tau^1$ and $\mathcal{V}^*_{\ds-k+1}= \text{PD}(\mathcal{V}^{k-1})$ being its Poincar\'e dual cocycle.   

We can also re-express one of the above relations in terms of the transversal CZ operators as follows: 
\begin{align}\label{eq:CZ_identity1}
&(-1)^{\int_{\mathcal{V}^*_{\ds-p+1}} \blue{\hat{b}^q} \cup \green{\hat{c}^s} } =  (-1)^{ \int_{\L} \blue{\hat{b}^q} \cup \green{\hat{c}^s} \cup \mathcal{V}^p } \cr
=&  \prod_{\sigma_q, \sigma_s} [ CZ^{\bl, \gr}_{\sigma_q, \sigma_s}]^{\int_{\L}  \tilde{\sigma}_q \cup   \tilde{\sigma}_s \cup \mathcal{V}^{p-1}} \equiv  \widetilde{\text{CZ}}^{\bl, \gr}_{\mathcal{V}^*_{\ds-p+1}} =1,
\end{align}
where we have derived the identity in the second line by expanding the gauge fields as $\blue{\hat{b}^q} $$ = $$ \sum_{\sigma_{q}}\hat{N}^q_{\bl}  \tilde{\sigma}_{q}$, $\green{\hat{c}^s} $$ = $$ \sum_{\sigma_{s}}\hat{N}^s_{\gr}  \tilde{\sigma}_{s}$,  along with the identity $(-1)^{ \hat{N}^q_{\bl} \hat{N}^{s}_{\gr}} $$=$$ \text{CZ}^{\bl,\gr}_{\sigma_q, \sigma_{s}}$. 
Similarly, we can re-express the other two relations by permuting the colors as:
\be\label{eq:CZ_identity2}
\widetilde{\text{CZ}}^{\rd, \gr}_{\mathcal{V}^*_{\ds-q+1}} = 1,  \qquad   \widetilde{\text{CZ}}^{\rd, \bl}_{\mathcal{V}^*_{\ds-s+1}} = 1.
\ee

\subsubsection{Stabilizer commutation relations and logical operators}

Now we investigate the detailed properties of the twisted qLDPC code in the form of a Clifford stabilizer code. 

As indicated by the definition of the Clifford stabilizer code Eq.~\eqref{eq:code_space_Clifford}, the code states have common +1-eigenvalues of all the stabilizer generators in $\tilde{\mathcal{S}}$, This shows that the parent Hamiltonian  
\begin{align}
\tilde{H}=& - \sum_{\sigma_{p-1} } \tilde{A}^{\rd}_{\sigma_{p-1}} - \sum_{\sigma_{q-1} }  \tilde{A}^{\bl}_{\sigma_{q-1}} - \sum_{\sigma_{s-1} }  \tilde{A}^{\gr}_{\sigma_{s-1}}  \cr
& - \sum_{\sigma_{p+1}} B^{\rd}_{\sigma_{p+1}} - \sum_{\sigma_{q+1}} B^{\bl}_{\sigma_{q+1}} - \sum_{\sigma_{s+1}} B^{\gr}_{\sigma_{s+1}}
\end{align}
is frustration free.   Nevertheless, this code is a non-commuting stabilizer code in the sense that the Clifford stabilizers do not commute in the full Hilbert space.   Instead,  as will be proved in the following,
the stabilizers commute within the $+1$ subspace of all the $Z$-stabilizer generators $\{ B^{\rd}_{\sigma_{p+1}},  B^{\bl}_{\sigma_{q+1}}, B^{\gr}_{\sigma_{s+1}} \}$:
\be
P_B  [g, g']    P_B  =1,  \  \forall g, g' \in \tilde{\mathcal{S}},
\ee
where $P_B$ is the projector onto the corresponding subspace.  Note that we have used the group commutator  $[U,V]$$:=$$UVU^{-1}V^{-1}$.

\begin{lemma}\label{lemma:stabilizer_commutation}
    The commutator of the stabilizers $\tilde A,B$ generate the $B$  stabilizers. Thus the stabilizers commute on the zero flux subspace $B=1$.
\end{lemma}
\begin{proof}
    The $B$ stabilizers commutes among themselves, and the $B$ stabilizers commute with the $\tilde A$ stabilizers. It remains to compute the commutation of the $\tilde A$ stabilizers themselves.
    We will use the following cup product identity for gauge field $a$ (i.e. it labels the Pauli $Z$ eigenvalue $Z=(-1)^a$) and Pauli $X$ on cells $s,s'$: for any integer value functions $f,g$,
    \begin{equation}\label{eqn:commutatoridentity}
        [X_s (-1)^{f(a)},\; X_{s'}(-1)^{g(a)}]=(-1)^{f(a+\tilde s')-f(a)}(-1)^{g(a+\tilde s)-g(a)}~,
    \end{equation}
    where $\tilde s$ is the indicator  cochain that takes value 1 on $s$ and 0 otherwise.
    This uses the property $[X_s, (-1)^{f(a)}]=(-1)^{f(a+\tilde s)-f(a)}$ \cite{Hsin:2024nwc}.
    
    For the case at hand, 
   \begin{widetext}
    \begin{align}
        &[\tilde A_{\sigma_{p-1}}^{\rd},\tilde A_{\sigma_{q-1}}^{\bl}]\cr 
        &=\left[
        \left(\prod_{\sigma_{p-1}\subset \partial \sigma_p} X_{\sigma_p}^r\right)(-1)^{\int \tilde{\sigma}_{p-1}\cup {\color{blue}\hat b^q}\cup {\green{\hat{c}^s}}},\;
        \left(\prod_{\sigma_{q-1}\subset \partial \sigma_q} X_{\sigma_q}^{\bl}\right)(-1)^{\int  {\color{red}\hat a^p}\cup \tilde{\sigma}_{q-1}\cup {\green{\hat{c}^s}}}
        \right]\cr 
        &=(-1)^{\int \tilde{\sigma}_{p-1}\cup d\tilde{\sigma}_{q-1}\cup {\green{\hat{c}^s}}+d\tilde{\sigma}_{p-1}\cup \tilde{\sigma}_{q-1}\cup \green{\hat{c}^s}}=(-1)^{\int d\left(\tilde{\sigma}_{p-1}\cup \tilde{\sigma}_{q-1}\right)\cup {\green{\hat{c}^s}}}
        =(-1)^{\int \left(\tilde{\sigma}_{p-1}\cup \tilde{\sigma}_{q-1}\right)\cup d{\green{\hat{c}^s}} }
        =B^{\gr}_{\widetilde{\tilde{\sigma}_{p-1}\cup\tilde{\sigma}_{q-1}}}~,
    \end{align}
    where $\widetilde{\tilde{\sigma}_{p-1}\cup\tilde{\sigma}_{q-1}}$ is the dual of $\tilde{\sigma}_{p-1}\cup\tilde{\sigma}_{q-1}$, i.e. we decompose the $(p+q-2)$ cochain $\tilde{\sigma}_{p-1}\cup\tilde{\sigma}_{q-1}$ as a linear combination of basic $(p+q-2)$ cochains $\tilde{\sigma}_{p-1}\cup \tilde{\sigma}_{q-1}=\sum \tilde s_{p+q-2}$ for $(p+q-2)$-cells $\{s\}$, then the dual is the chain given by replacing $\tilde s_{p+q-2}$ in the sum with $s_{p+q-2}$.
       \end{widetext}
    Similarly, 
    \begin{equation}
        [\tilde A_{\sigma_{p-1}}^{\rd},\tilde A_{\sigma_{s-1}}^{\gr}=B^{\bl}_{\widetilde{\tilde{\sigma}_{p-1}\cup \tilde{\sigma}_{s-1}}},\;\;
                [\tilde A_{\sigma_{q-1}}^{\gr},\tilde A_{\sigma_{s-1}}^{\bl}]=B^{\rd}_{\widetilde{\tilde{\sigma}_{q-1}\cup \tilde{\sigma}_{s-1}}}~.
    \end{equation}
Thus the stabilizers commute on the stabilized subspace $B^{\rd}=1,B^{\gr}=1,B^{\bl}=1$.

\end{proof}

We now discuss the logical operators of the twisted qLDPC codes.  The logical-$Z$ operators are analogous to the untwisted codes,  which corresponds to the electric Wilson operators in the gauge theories:
\begin{align}\label{eq:logical_Z_1}
\widetilde{Z}^{\rd}_{{\eta_p}} \equiv & W^{\rd}({\eta_p})=e^{i \pi\int_{{\eta_p}} \red{\hat{a}^p}} = (-1)^{\int_{{\eta_p}} \red{\hat{a}^p}} = \prod_{\sigma_p \in {\eta_p}} Z^{\rd}_{\sigma_p},  \cr
\widetilde{Z}^{\bl}_{{\eta_q}} \equiv & W^{\bl}({\eta_q})=e^{i \pi \int_{{\eta_q}} \blue{\hat{b}^q}} = (-1)^{\int_{{\eta_q}} \blue{\hat{b}^q}} = \prod_{\sigma_q \in {\eta_q}} Z^{\bl}_{\sigma_q},   \cr
\widetilde{Z}^{\gr}_{{\eta_s}} \equiv & W^{\gr}({\eta_s})=e^{i \pi \int_{{\eta_s}} \green{\hat{c}^s}} = (-1)^{\int_{{\eta_s}} \green{\hat{c}^s}} = \prod_{\sigma_s \in {\eta_s}} Z^{\gr}_{\sigma_s}.
\end{align}
Here, $\widetilde{Z}$ means it is an extended operator supported on a cycle, in particular the non-trivial $p$-cycle ${\eta_p}$, the $q$-cycle ${\eta_q}$ and the $s$-cycle ${\eta_s}$  for three types of electric charge $\red{e}_{\rd}$, $\blue{e}_{\bl}$, and $\green{e}_{\gr}$ respectively.  We do not use the notation $\lo{Z}$ here since these cycles $\eta_k$ may not be the basis cycles to label the logical qubits, but instead arbitrary non-trivial $k$-cycles.  

Now we label the logical qubits using the cocycle basis $\{\red{\alpha^p}\}$, $\{\blue{\beta^q}\}$ and $\{\green{\gamma^s}\}$ or equivalently their conjugate  cycle basis $\{\red{\alpha_p}\}$, $\{\blue{\beta_q}\}$ and $\{\green{\gamma_s}\}$ where the logical-$Z$ are supported.   For simplicity, we just use $\{\red{\alpha}\}$, $\{\blue{\beta}\}$ and $\{\green{\gamma}\}$ as the logical qubit labels.  We can hence define the logical-$Z$ operator acting a particular logical qubit as
\be\label{eq:logical_Z_2}
\lo{Z}^{\rd}(\red{\alpha}) \equiv  \widetilde{Z}^{\rd}_{\red{\alpha_p}}, \quad   \lo{Z}^{\bl}(\blue{\beta}) \equiv  \widetilde{Z}^{\bl}_{\blue{\beta_q}},  \quad  \lo{Z}^{\gr}(\green{\gamma}) \equiv  \widetilde{Z}^{\gr}_{\green{\gamma_s}}. 
\ee

In the twisted code,  there exists logical operators similar to the logical-$X$ operators, which can be called dressed $X$ operators.  The difference is that they do not simply apply logical-$X$ operation to the code space $\tilde{\C}$.    

First, there exist the follosing contractible magnetic operators in the twisted gauge theory, and can be defined as follows:  
\begin{align}\label{eqn:magoprCZ}
 {\widetilde{\mathcal{X}}}^{
 \rd}_{\Sigma^*_{\ds-p}} \equiv &  M^{
 \rd}(\Sigma^*_{\ds-p})  =    (-1)^{\int_{\Sigma^*_{\ds-p}} \red{\hat{\tilde a}^{\ds-p}}} \cdot (-1)^{\int_{{\cal V}^*_{\ds-p+1}}  \blue{\hat{b}^q} \cup  \green{\hat{c}^s} }\cr 
 =&  \prod_{\sigma^*_{\ds-p} \in \Sigma^*_{\ds-p}  }  X^{\rd}_{\sigma^*_{\ds-p}}  \prod_{\sigma_q, \sigma_s} [ \text{CZ}^{\bl, \gr}_{\sigma_q, \sigma_s}]^{\int_{\L}  \tilde{\sigma}_q \cup   \tilde{\sigma}_s \cup \mathcal{V}^{p-1}} \cr
  =&  \prod_{\sigma^*_{\ds-p} \in \Sigma^*_{\ds-p}  }  X^{\rd}_{\sigma^*_{\ds-p}}  \prod_{\sigma_q, \sigma_s: \ \int_{\L}\tilde{\sigma}_q \cup   \tilde{\sigma}_s \cup \mathcal{V}^{p-1} \neq 0 }  \text{CZ}^{\bl, \gr}_{\sigma_q, \sigma_s}~, \cr
{\widetilde{\mathcal{X}}}^{
 \bl}_{\Sigma^*_{\ds-q}} \equiv &  M^{
 \bl}(\Sigma^*_{\ds-q})  =    (-1)^{\int_{\Sigma^*_{\ds-q}} \blue{\hat{\tilde b}^{\ds-q}}} \cdot (-1)^{\int_{{\cal V}^*_{\ds-q+1}}  \red{\hat{a}^p} \cup  \green{\hat{c}^s} }\cr 
 =&  \prod_{\sigma^*_{\ds-q} \in \Sigma^*_{\ds-q}  }  X^{\bl}_{\sigma^*_{\ds-q}}  \prod_{\sigma_p, \sigma_s: \ \int_{\L}\tilde{\sigma}_p \cup   \tilde{\sigma}_s \cup \mathcal{V}^{q-1} \neq 0 }  \text{CZ}^{\rd, \gr}_{\sigma_p, \sigma_s}~, \cr  
 {\widetilde{\mathcal{X}}}^{
 \gr}_{\Sigma^*_{\ds-s}} \equiv &  M^{
 \gr}(\Sigma^*_{\ds-s})  =    (-1)^{\int_{\Sigma^*_{\ds-s}} \green{\hat{\tilde c}^{\ds-s}}} \cdot (-1)^{\int_{{\cal V}^*_{\ds-s+1}}  \red{\hat{a}^p} \cup  \blue{\hat{b}^q} }\cr 
=&  \prod_{\sigma^*_{\ds-s} \in \Sigma^*_{\ds-s}  }  X^{\gr}_{\sigma^*_{\ds-s}}  \prod_{\sigma_p, \sigma_q: \ \int_{\L}\tilde{\sigma}_p \cup   \tilde{\sigma}_q \cup \mathcal{V}^{s-1} \neq 0 }  \text{CZ}^{\rd, \bl}_{\sigma_p, \sigma_q}~. \cr  
\end{align}
In the above expressions, ${\cal V}^{p-1},{\cal V}^{q-1},{\cal V}^{s-1}$ are not closed, since their dual have nontrivial boundaries. Thus we cannot apply (\ref{eq:constraint_operator}) to annihilate these bulk contributions.
Note that we have associated the magnetic gauge fields $\red{\hat{\tilde{a}}^{\ds-p}},\blue{\hat{\tilde b}^{\ds-q}}, \green{\hat{\tilde c}^{\ds-s}}$  with the Pauli-$X$ operators on the dual cells $\sigma^*_{\ds - p}, \sigma^*_{\ds - q}, \sigma^*_{\ds - q} \in \L^*$, where $\L^*$ is the dual complex, e.g., $(-1)^{\red{\hat{\tilde{a}}^{\ds-p}} (\sigma^*_{\ds - p})}$$=$$ X^{\rd}_{\sigma}(\sigma^*_{\ds - p})$.  In the above equation, $\Sigma^*_{\ds-k} \in H_{\ds-k}(\L^*, \ZZ_2)$ is the $(\ds-k)$-cycle in the dual complex $\L^*$ ($k=p,q,s$), and in addition we have the relation $\partial \mathcal{V}^*_{\ds-k+1}=\Sigma^*_{\ds-k}$, where $\mathcal{V}^*_{\ds-k+1}$ is a relative cycle (open membrane with boundaries).       

In the code space $\tilde{\C}$, 
these magnetic operators do not depend on the particular choice of $\mathcal{V}^*_{\ds-k+1}$ as long as its boundary remains $\Sigma^*_{\ds-k}$  \cite{Barkeshli:2022edm}.  
To see this, we can choose two different relative cycles $\mathcal{V}^*_{\ds-p}$ and $\mathcal{V}'^*_{\ds-p}$, and their difference $\mathcal{V}^*_{\ds-p} - \mathcal{V}'^*_{\ds-p} = S^*_{\ds-p} $ is a closed cycle due to the cancellation of the common boundary $\Sigma^*_{\ds - p}$.  From the stabilizer conditions on the code space $\tilde {\cal C}$, we have $\int_{{S}^*_{\ds-p+1}}  \blue{\hat{b}^q} \cup  \green{\hat{c}^s}=0$, which shows $\int_{\mathcal{V}^*_{\ds-p+1}}  \blue{\hat{b}^q} \cup  \green{\hat{c}^s} = \int_{\mathcal{V}'^*_{\ds-p+1}}  \blue{\hat{b}^q} \cup  \green{\hat{c}^s}$.   Note that we have derived the CZ expressions by expanding the gauge fields as $\red{\hat{a}^p} $$ = $$ \sum_{\sigma_{p}}\hat{N}^p_{\rd}  \tilde{\sigma}_{p}$,   $\blue{\hat{b}^q} $$ = $$ \sum_{\sigma_{q}}\hat{N}^q_{\bl}  \tilde{\sigma}_{q}$, $\green{\hat{c}^s} $$ = $$ \sum_{\sigma_{s}}\hat{N}^s_{\gr}  \tilde{\sigma}_{s}$,  along with the identity $(-1)^{ \hat{N}^k_i \hat{N}^{k'}_j}= \text{CZ}^{i,j}_{\sigma_k, \sigma_{k'}}$, where $i,j \in \{\rd, \bl, \gr\}$ and $k, k' \in \{p, q, s\}$.

The above magnetic operators are supported on contractibel cycles which only give rise to trivial logical action and are hence not logical operators. We hence needs to consider magnetic operators supported on non-contractible cycles, which are the actual logical operators. We can infer the expression of these operators within the code space as follows. The magnetic operators create holonomy for the electric Pauli $Z$ operator that intersects the support of the magnetic operator at a point. On the other hand, the holonomies in the code space satisfy Eq.~(\ref{eq:constraint_operator}). Therefore the logical magnetic operators need to be decorated with projection to ensure Eq.~(\ref{eq:constraint_operator}):
\begin{widetext}
\begin{align}\label{eqn:magneticoprprojection0}
    \begin{split}
   \cX^{\rd}_{C^*_{\ds - p}} \equiv    M^{\rd}(C^*_{\ds - p}) &= P^{\rd}\left[\prod_{\sigma_p\in  C^*_{\ds - p}} X^{\rd}_{\sigma_p}\prod_{\sigma_q,\sigma_s\in C^*_{\ds - p}}(1+Z^{\bl}_{\sigma_q})(1+Z^{\gr}_{\sigma_s})\right] P^{\rd} \\
    \cX^{\bl}_{C^*_{\ds - q}} \equiv    M^{\bl}(C^*_{\ds - q}) &=P^{\bl}\left[\prod_{\sigma_q\in   
         C^*_{\ds - q}}X^{\bl}_{\sigma_q} \prod_{\sigma_p,\sigma_s\in C^*_{\ds - q}}(1+Z^{\rd}_{\sigma_p})(1+Z^{\gr}_{\sigma_s})\right]P^{\bl} \\
    \cX^{\gr}_{C^*_{\ds - q}} \equiv
        M^{\gr}( C^*_{\ds - q}) &=P^{\gr}\left[\prod_{\sigma_s\in  C^*_{\ds - q}}X^{\gr}_{\sigma_s}\prod_{\sigma_p,\sigma_q\in C^*_{\ds - q}}(1+Z^{\rd}_{\sigma_p})(1+Z^{\bl}_{\sigma_q})\right] P^{\gr}~,\\
    \end{split}
\end{align}
\end{widetext}
where $P^{\rd},P^{\bl},P^{\gr}$ are the stabilizer projectors to the +1-eigenspace for ${A}^{\bl}_{\sigma_{q-1}}, {A}^{\gr}_{\sigma_{s-1}}$, and ${A}^{\rd}_{\sigma_{p-1}}, {A}^{\gr}_{\sigma_{s-1}}$, and ${A}^{\rd}_{\sigma_{p-1}}, {A}^{\bl}_{\sigma_{q-1}}$, respectively. Here, $C^{*}_{\ds - p},C^{*}_{\ds - q},C^{*}_{\ds - q}$ are non-contractible cycles defined on the dual complex $\L^*$, and the product indexed by $\sigma_{q},\sigma_s$ in the first equation of \eqref{eqn:magneticoprprojection0} is over all cells of the $(\ds - p)$-cycle $C^*_{\ds - p}$, and similar for other equations.
The above expression is crucial for deriving the non-Abelian fusion rules as discussed in Ref.~\cite{Hsin2024_non-Abelian}.

Let us show that magnetic operators (\ref{eqn:magoprCZ}),(\ref{eqn:magneticoprprojection0}) commute with the stabilizers on the code space $\tilde {\cal C}$, where the flux is zero $B=1$:
\begin{lemma}\label{lemma:stabilizer_magnetic_commutation}
    The magnetic operators commute with the stabilizers on the code space $\tilde {\cal C}$.
\end{lemma}
\begin{proof}
For the magentic operators  in Eq.~\eqref{eqn:magoprCZ},
again we use the identity (\ref{eqn:commutatoridentity}). We have
\begin{align}
    & [ M^{\bl}(\Sigma_{\ds-q}^*),A_{\sigma_{p-1}}^{\rd}] \cr
    =& (-1)^{\int d\tilde{\sigma}_{p-1} \cup {\cal V}^{q-1}\cup  \green{\hat{c}^s}}(-1)^{\int \tilde{\sigma}_{p-1}\cup  d{\cal V}^{q-1}\cup \green{\hat{c}^s}}\cr
    =& (-1)^{\int \tilde{\sigma}_{p-1}\cup {\cal V}^{q-1}\cup  \green{d \hat{c}^s}}~,
\end{align}
which equals 1 in the subspace $B=1$ where $\green{d \hat{c}^s}=0$ mod 2.
Similarly, one can verify that other magnetic operators commute with the stabilizers in the code space. 

For the expression (\ref{eqn:magneticoprprojection0}) of the magnetic logical operators, the commutation with the stabilizers follows from the stabilizer projection. The operators $(\ref{eqn:magneticoprprojection0})$ commute with the B-type stabilizers as in the untwisted code. For the commutator with A type stabilizers, we note that $M^{\rd}$ commutes with ${A}^{\rd}$, and it commutes with ${A}^{\bl},{A}^{\gr}$ since the operator $M^{\rd}$ contains the projectors for the stabilizers ${A}^{\bl},{A}^{\gr}$:
\begin{align}
    &P^{\rd} {A}^{\bl}=P^{\rd} {A}^{\gr}=P^{\rd}=
    {A}^{\bl}P^{\rd} ={A}^{\gr}P^{\rd} 
    \cr
    &\Rightarrow  M^{\rd}{A}^{\bl}=M^{\rd}={A}^{\bl}M^{\rd},\; 
    M^{\rd}{A}^{\gr}=M^{\rd}={A}^{\gr}M^{\rd}~. \cr
\end{align}
Similarly, the other magnetic logical operators commute with the stabilizers.

\end{proof}

\subsubsection{Charge parity operators as higher-form and subcomplex symmetries}\label{sec:charge_parity_operator}

An important observable which will be used later for the gauging measurement in Sec.~\ref{sec:logical_operation} are the charge parity operators of $\rd$-, $\bl$-, $\gr$-types   defined as follows:
\begin{align}\label{eq:charge_parity}
{\mathsf{C}}^{\rd}_{\eta^{p-1}} \equiv {\mathsf{C}}^{\rd}_{\eta^*_{\ds-p+1}} =& \prod_{ \sigma_{p-1} \in \eta^{p-1}} \tilde{A}^{\rd}_{\sigma_{p-1}},   \cr
{\mathsf{C}}^{\bl}_{\eta^{q-1}} \equiv {\mathsf{C}}^{\bl}_{\eta^*_{\ds-q +1}} =& \prod_{ \sigma_{q - 1} \in \eta^{q - 1}} \tilde{A}^{\bl}_{\sigma_{q -1}},   \cr
{\mathsf{C}}^{\gr}_{\eta^{s-1}} \equiv {\mathsf{C}}^{\gr}_{\eta^*_{\ds-s +1}} =& \prod_{ \sigma_{s- 1} \in \eta^{s-1}} \tilde{A}^{\gr}_{\sigma_{s-1}}.        
\end{align}
These operators are higher-form [$(k-1)$-form] symmetries supported on the codimension-$(k-1)$ cocycles $\eta^{k-1} \in H^{k-1}(\L; \ZZ_2)$ or equivalently their  $[\ds-(k-1)]$-dimensional  Poincare dual cycles $\eta^*_{\ds -(k-1)} \equiv \eta^*_{D-k} \in H_{\ds-(k-1)}(\L^*; \ZZ_2)$   with $k=p,q,s$.  The higher-form symmetry operators are expressed as a product of the dressed $X$-stabilizers $\tilde{A}^i_{\sigma_{k-1}}$ defined on the cells $\sigma_{k-1} $. In Eq.~\eqref{eq:charge_parity}, $\sigma_{k-1} \in \eta^{k-1}$ stands for the condition $\eta^{k-1} (\sigma_{k-1}) =1$, which can be interpreted as: the cell $\sigma_{k-1}$ is supported on the cocycle  $\eta^{k-1}$.

\begin{figure*}[t]
 \centering
    \includegraphics[width=1\linewidth]{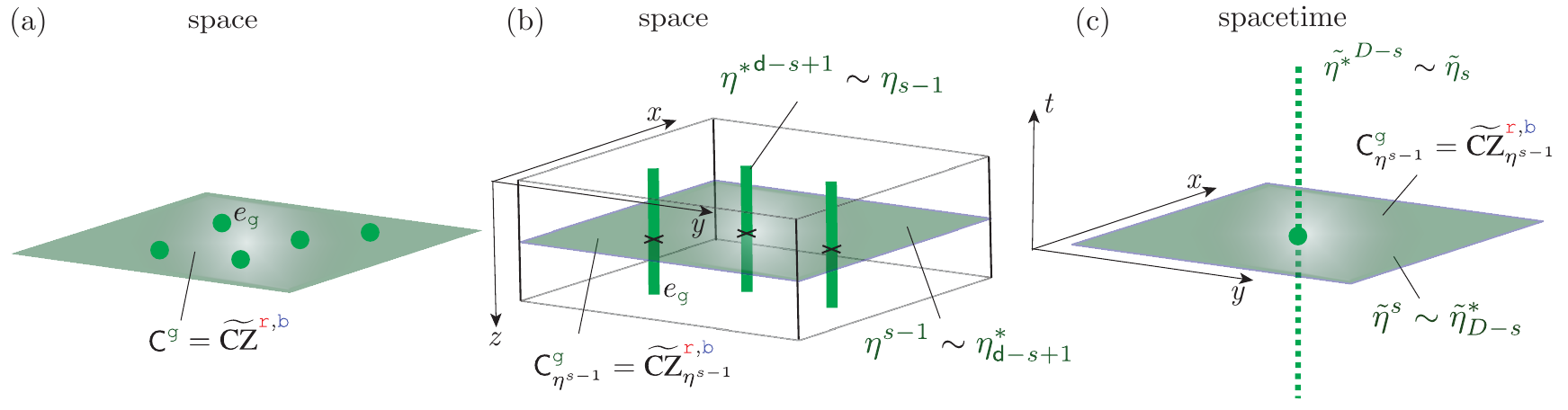}.
    \caption{(a) On a 2-manifold, there is only a single charge parity and transversal CZ operator $\mathsf{C}^{\gr} = \widetilde{\text{CZ}}^{\rd, \bl}$ supported on the entire manifold, which is a 0-form global symmetry in the untwisted code. The charge $e_{\gr}$ are point-like. (b) On a higher-dimensional $\ds$-manifold, the charge $e_{\gr}$ become an extended object support on an $(s-1)$-cycle $\green{\eta_{s-1}}$, which is detected by the intersected charge parity and transversal CZ operator supported on a $(\ds-s+1)$-cycle $\green{\eta^*_{\ds-s+1}}$.  The operator hence corresponds to a higher-form [($s-1$)-form] symmetry.  There are many classes of $(\ds-s+1)$-cycle $[\green{\eta^*_{\ds-s+1}}]$ which the higher-symmetry operators can be supported on.  The higher-symmetry operators can hence be measured independently, which is the key to addressable gauging measurement.    (c) In the spacetime picture,  a charge worldsheet supported on the spacetime $s$-cycle $\tilde \eta_s$ intersects with the spacetime $(D-s)$-cycle $\green{\tilde \eta^*_{D-s}}$ and is hence detected by the corresponding higher-symmetry operator.   }
    \label{fig:charge_intersection}
\end{figure*}

Note that the bare $X$-stabilizers in the untwisted qLDPC codes in Eq.~\eqref{eq:X-stabilizers_untiwsted} satisfy the following relations: 
\begin{align}
 \prod_{ \sigma_{p-1} \in \eta^{p-1}} {A}^{\rd}_{\sigma^{p-1}} =& 1,   \quad \prod_{ \sigma_{q - 1} \in \eta^{q - 1}} \tilde{A}^{\bl}_{\sigma_{q -1}} = 1,   \cr
\prod_{ \sigma_{s- 1} \in \eta^{s-1}} \tilde{A}^{\gr}_{\sigma_{s-1}} =& 1.
\end{align}
Therefore, all the Pauli-$X$ terms in the charge parity operators in Eq.~\eqref{eq:charge_parity} get canceled and what remained are just transversal CZ operators.  As an illustration, we consider the $\gr$-type charge operators: 
\begin{align}\label{eq:charge_operator_green}
{\mathsf{C}}^{\gr}_{\eta^{s-1}} =&  \prod_{ \sigma_{s-1} \in \eta^{s-1}} \ \  \prod_{\sigma_p,\sigma_q: \int \tilde{\sigma}_p\cup \tilde{\sigma}_q\cup \tilde{\sigma}_{s-1} \neq 0}\text{CZ}^{\rd,\bl}_{\sigma_p, \sigma_q} \cr
 \equiv & \prod_{ \sigma_{s-1} \in \eta^{s-1}}   \prod_{\sigma_p,\sigma_q} [\text{CZ}^{\rd,\bl}_{\sigma_p, \sigma_q}]^{ \int \tilde{\sigma}_p\cup \tilde{\sigma}_q\cup \tilde{\sigma}_{s-1}}    \cr
 =&    \prod_{\sigma_p,\sigma_q} [\text{CZ}^{\rd,\bl}_{\sigma_p, \sigma_q}]^{ \int \tilde{\sigma}_p\cup \tilde{\sigma}_q\cup \sum_{\sigma_{s-1} \in \eta^{s-1}} \tilde{\sigma}_{s-1}}    \cr
 =&  \prod_{\sigma_p,\sigma_q} [\text{CZ}^{\rd,\bl}_{\sigma_p, \sigma_q}]^{ \int \tilde{\sigma}_p\cup \tilde{\sigma}_q\cup  \eta^{s-1}}  \equiv  \widetilde{\text{CZ}}^{\rd,\bl}_{\eta^{s-1}} \equiv  \widetilde{\text{CZ}}^{\rd,\bl}_{\eta^*_{\ds-s +1}},   \cr
\end{align}
which is equal to a transversal CZ operator $\widetilde{\text{CZ}}^{\rd,\bl}_{\eta^{s-1}}$$\equiv  $$ \widetilde{\text{CZ}}^{\rd,\bl}_{\eta^*_{\ds-s +1}}$ supported on the cocycle $\eta^{s-1}$ or equivalently the dual cycle $\eta^*_{\ds-s +1}$.    By symmetry we can also derive the other two relations:
\begin{align}
{\mathsf{C}}^{\rd}_{\eta^{p-1}} \equiv \widetilde{\text{CZ}}^{\bl, \gr}_{\eta^{p-1}}  \equiv &  \widetilde{\text{CZ}}^{\rd,\bl}_{\eta^*_{\ds-p +1}}    \cr
{\mathsf{C}}^{\bl}_{\eta^{q-1}} \equiv \widetilde{\text{CZ}}^{\rd, \gr}_{\eta^{q-1}}  \equiv  & \widetilde{\text{CZ}}^{\rd,\gr}_{\eta^*_{\ds-q +1}}.    \cr
\end{align}

Recall that the transversal CZ operator can be expressed with the form:
\begin{align}
&\widetilde{\text{CZ}}^{\rd,\bl}_{\eta^*_{\ds-s +1}} =(-1)^{\int_{\eta^*_{\ds-s+1}} \red{\hat{a}^p} \cup  \blue{\hat{b}^q} }, \cr
\end{align}
which shows a generical form of higher-form ($k$-form) symmetry
in terms of cycle-cocycle pairing.   In this example, it is the pairing between the operator valued $(\ds-s+1)$-cocycle $\red{\hat{a}^p} \cup  \blue{\hat{b}^q}$  ($p+q=\ds-s+1$) and the $(\ds-s+1)$-cycle $\eta^*_{\ds-s+1}$. 

According to Eqs.~\eqref{eq:CZ_identity1} and \eqref{eq:CZ_identity2}, in the code (ground) space $\tilde{\C}$ the above transversal CZ operators all have eigenvalues $+1$.  Therefore, the charge parity operators should also have +1 eigenvalues:
\be
{\mathsf{C}}^{\rd}_{\eta^{p-1}} = 1,  \quad {\mathsf{C}}^{\bl}_{\eta^{q-1}} = 1,  \quad
{\mathsf{C}}^{\gr}_{\eta^{s-1}} = 1. 
\ee

\begin{figure}
    \centering \includegraphics[width=0.7\linewidth]{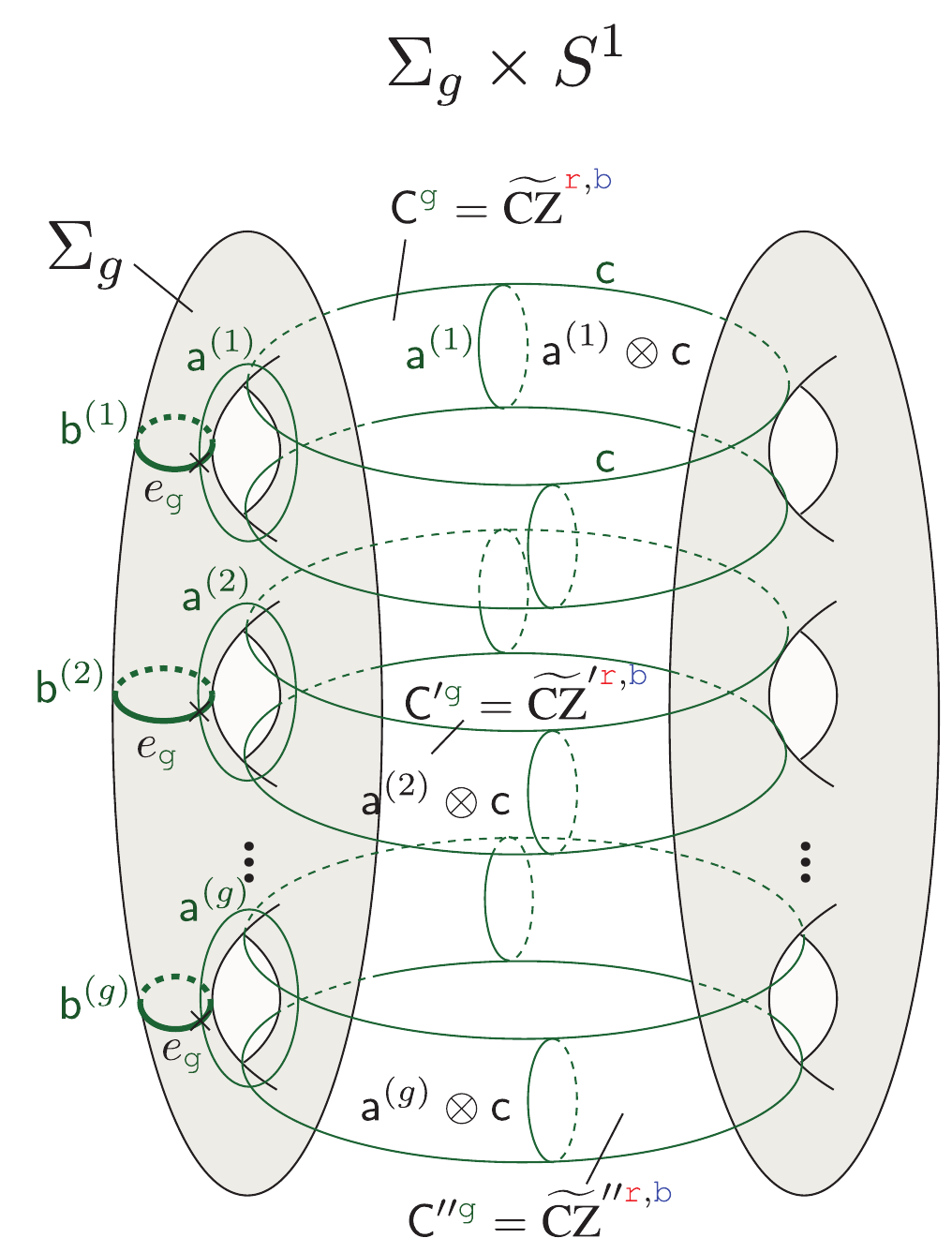}
    \caption{A toy example of a 3-manifold $\Sigma_g \times S^1$ ($\Sigma_g$ stands for a genus-$g$ surface) can host $O(g)$ different 1-form symmetry operators supported on distinct 2-cycle classes. Each operator is intersected with the charge supported on its dual 1-cycle.}
    \label{fig:three-manifold_examples}
\end{figure}

The violation of these constraints, i.e., $ {\mathsf{C}}^{i}_{\eta^{k-1}} = -1$ indicates that there exist \textit{charge excitations}  of type $i$ (\rd, \bl, \gr).  In the previous study in 2D topological codes \cite{Davydova:2025ylx}, where the transversal CZ operator is a 0-form global symmetry and hence there is only a single 2-cycle (0-cocycle) $\eta^*_2 \sim \eta^0$  corresponding to the entire 2-manifold where the operator can be supported, and $\mathsf{C}^{i}$ measures the total parity of all the electric charges of type $i$ in the entire code, as shown in Fig.~\ref{fig:charge_intersection}(a). This is due to the fact that there is only a unique top cycle ($\ds$-cycle), as well as a unique bottom cocycle ($0$-cocycle) in a $\ds$-dimensional manifold ($\ds=2$ here), since the Betti number $b_{\ds} = b_0=1$. In contrast,  the 2D hypergraph-product code considered here in the present case does not resemble a 2-manifold but a higher-dimensional manifold,  since it is mapped to a $\ds$-dimensional Poincar\'e CW complex or manifold (e.g., $\ds=16$).  The transversal CZ operator is a $(k-1)$-form symmetry  and can be supported on many possible  $(\ds-k+1)$-dimensional subcomplexes (sub-manifolds) corresponding to (co)cycle classes:  $[\eta^{k-1}]\sim[\eta^*_{d-k+1}]$, instead of the entire space complex $\L$, as illustrated in Fig.~\ref{fig:charge_intersection}(b). Note that this is possible since $\eta^*_{d-k+1}$ ($\eta^{k-1}$)is not a top cycle (bottom cocycle). More generally, such a higher-form symmetry is a special case of a \textit{subcomplex symmetry}.   

Note that the charge excitations are no longer point-like particles in the $\ds$-dimensional Poincar\'e CW complex (manifold) as in the case of the 2-manifold.  Instead, they are extended objects supported on the $(k-1)$-dimensional non-contractible cycles (strings/membranes) $\eta_k$ intersecting with the support of the transversal CZ operator.  As we can see, when setting $k=p=q=s=1$ and $\ds=2$ (2-manifold),  the charge reduces back to 0-dimensional point-like particle supported on a 0-cycle $\eta_k$, which is just a single vertex in the corresponding triangulation.   As illustrated in  Fig.~\ref{fig:charge_intersection}(b) for a toy example in a 3-manifold ($\ds=3$, $k=s=2$) where the charge are non-contractible strings going around the $z$-direction (with periodic boundary condition).   
We emphasize that since the $\ds$-dimensional Poincar\'e CW complex (manifold) is built from the 2D hypergraph-product code, most of the dimensions are effectively compactified to $O(1)$ size, as illustrated by the thin $z$-direction in Fig.~\ref{fig:charge_intersection}(b).  Therefore, although the local topological dimension of the charge excitation is $k-1$, at a large scale the charge excitation actually still looks point-like.  As can be seen in Fig.~\ref{fig:charge_intersection}(b), when the $z$-direction is compactified to $O(1)$ size, the string-like charge excitation effectively  becomes point-like at large scale.  Nevertheless, the higher topological dimension of the charge makes its intersection property quite non-trivial, since there exist many possible cycle (membrane) that the charge excitation can intersect with, which is impossible for a point charge in a 2-manifold.

Now let us think about a deep question: what is the essential difference between a 2-manifold and a 2D hypergraph-product code (a general 2-complex)?  
In the 2-manifold, there is only one unique 0-cycle class $[\eta_0]$ (a single vertex of the triangulation) as the support of point-like charge. The moving of the charge deforms the support to another vertex,  which is just another representative in the same 0-cycle class $[\eta_0]$ and that does not change the total charge parity $\mathsf{C}^i$.  The hypergraph-product code and the general 2-complex (3-term chain complex) does not have a unique 0-cycle
class, but in general a collection of them forming a cycle basis $\{ \eta_0\}$.  The charge can hence be supported on any basis cycle $\bar{\eta}_0$, and one hence needs to define the total charge parity for  every basis cycle class $[\bar{\eta}_0]$. The moving and deformation of the charge within the same cycle class $[\bar{\eta}_0]$ does not change the total charge parity for that class.  The 0-cycle class $[\bar{\eta}_0]$ in the hypergraph-product code is then mapped to the $(k-1)$-cycle class $[\eta_{k-1}]$ in the CW complex $\L$.    We will continue the  discussion of the subcomplex symmetry in the case of the general 2D chain complex in Sec.~\ref{sec:0-form}.

Now if we choose a set of cocycle basis $\{\eta^{k-1}\}$ and its Poincar\'e dual cycle basis $\{\eta^*_{d-k+1}\}$, then $ {\mathsf{C}}^{i}_{\eta^{k-1}}$ measures the total parity of charges supported on the dual non-contractible basis (co)cycle ${\eta^*}^{d-k+1} \sim \eta_{k-1}$ that intersects with the basis (co)cycle $\eta^{k-1} \sim \eta^*_{d-k+1}$, i.e.,
\be
\int_{\L}  \eta^{k-1} \cup {\eta^*}^{d-k+1}  = | \eta^*_{d-k+1} \cap \eta_{k-1} |=1.
\ee
Note that the more general relation is the following:
\be
\int_{\L}  \eta^{k-1} \cup {\eta'^*}^{d-k+1} =  | \eta^*_{d-k+1} \cap \eta'_{k-1} |= \delta_{\eta, \eta'},
\ee
which shows the unique intersection between only the pair of  dual (co)cycles.   Note that for each basis (co)cycle $\eta^{k-1} \sim \eta^*_{d-k+1}$, there is an independent total charge parity being measured. This is the essence why we can have addressable and parallelizable gauging measurement of individual transversal CZ operators $\widetilde{\text{CZ}}^{\rd,\bl}_{\eta^{s-1}}$,  rather than the product of all the operators, as will be studied in Sec.~\ref{sec:parallel_gauging}.  

We illustrate this interesting phenomenon in the toy example of a 3-manifold $\M^3=\Sigma_g \times S^1$ that is the product of a genus-$g$ surface $\Sigma_g$ and a circle $S^1$ as shown in Fig.~\ref{fig:three-manifold_examples}.  On the genus-$g$ surface, there is a standard 1-cycle basis $\{\as^{(1)}, \as^{(2)}, \cdots \as^{(g)},  \bs^{(1)}, \bs^{(2)}, \cdots \bs^{(g)} \}$ with the intersection condition $|\as^{(i)}\cap \bs^{(j)}| = \delta_{i,j}$.  The 1-cycle along the circle $S^1$ is denoted by $\cs$.   The K\"unneth formula gives rise to the basis 2-cycles $\as^{(i)}\times \cs$ ($i=1,2, \cdots g$).   We can see that there can be loop-like charges $e_{\gr}$ supported on different cycle classes $[\bs^{(1)}], [\bs^{(2)}], \cdots [\bs^{(g)}]$.   For the charges supported on each cycle class $[\bs^{(i)}]$, the corresponding total charge parity is detected by the charge parity operator $\mathsf{C}^{\gr}$ and transversal CZ operator supported on $\as^{(i)} \otimes \cs$.   This gives the opportunity to perform addressable (independent) measurement on each transversal CZ operator.

Next, we can also understand what the charge parity operator detects in the spacetime picture. We first choose a spacetime cocycle basis $\{\tilde{\eta}^{k}\} \equiv \{\eta^{k-1} \otimes {\tau^*}^1\}$, along with its Poincar\'e dual cycle basis $\{\tilde\eta^*_{D-k}\}$$\equiv$$\{ \eta^*_{D-k} \otimes \tau_0\}$, where $\eta^*_{D-k}  \equiv \eta^*_{\ds-k+1}$ is the space cycle.
There exist electric charge world-sheets represented by the  dual basis (co)cycle $\tilde{\eta}_k \sim \tilde{\eta^*}^{D-k}$ that have non-trivial $\ZZ_2$ intersection (intersect odd number of times) with the basis (co)cycles $\tilde\eta^{k} \sim \tilde\eta^*_{D-k}$ (or equivalently the Poincar\'e dual cocycles $\tilde\eta^{k}$). This can be expressed by the following intersection condition:
\be
\int_{\tilde \L}  \tilde{\eta}^k \cup \tilde{\eta^*}^{D-k}=|\tilde\eta^*_{D-k} \cap \tilde{\eta}_k |=1.
\ee
 Note that the closed spacetime cycles $\tilde\eta^*_{D-k}$ and  $\tilde{\eta}_k$ can intersect at a single point since their dimension add up to the total spacetime dimension~$D$. 
Therefore, $ {\mathsf{C}}^{i}_{\eta^{k-1}}$ detects the parity of the electric worldsheet intersecting with the spacetime (co)cycle $\tilde{\eta}^k \sim \tilde{\eta}^*_{D-k}$.

We now illustrate it for the $\gr$-type charge $e_{\gr}$ (i.e., 
$i=\gr, \ k=s$). The spacetime illustration in Fig.~\ref{fig:charge_intersection}(c) extends the spatial illustration in (b), and hence considers a (3+1)D spacetime by choosing $s=2$.  The figure hides the $z$-direction in (b) while showing the $t$-direction.  The \textit{electric charge} is string-like (with dimension $s-1=1$) as in (b) and its closed worldsheet represented by the s-cycle  $\green{\tilde{\eta}_s}$ is membrane-like ($s=2$), although in (c) it is visualized as a dashed line since the $z$-direction is hidden.  

The electric charge $e_{\gr}$ can be considered as the boundary of an open worldsheet, which can be represented by a relative cycle $\green{\tilde{\zeta}_s}$.  Namely, the charge is supported on $\green{\partial  \tilde{\zeta}_s}$.   

Although we were focusing on the path integral of the ground-state sector of the gauge theory in Sec.~\ref{sec:Feynman_path_integral},  the modified path integral can also describe the situation in the presence of \textit{excitations} including the electric charge excitations $e_{\rd}, e_{\bl}, e_{\gr}$,  and the magnetic flux excitations     $m_{\rd}, m_{\bl}, m_{\gr}$.  For example, in the presence of a charge excitation $e_{\gr}$ with the corresponding open worldsheet $\green{\tilde{\zeta}_s}$ which is a relative cycle terminated at the excitation, the path integral is modified as follows: 
\begin{align}\label{eq:path_integral_defect}
& \mathcal{Z}_{(e_{\gr}, \green{\tilde{\zeta}_s})}[\tilde{\L}] \cr
=& \sum_{\red{a^*_{D-p}}, \blue{b^*_{D-q}}, \green{c^*_{D-s}} \in {H_\bullet}(\tilde{\L})} (-1)^{| \red{a^*_{D-p}} \cap \blue{b^*_{D-q}} \cap \green{c^*_{D-s}}|}  \cdot (-1)^{| \green{\tilde{\zeta}_s} \cap \green{c^*_{D-s}} |}. \cr
\end{align}
One configuration in the sum with an extra ($-1$)-factor due to the excitation is  illustrated as below:
\be
    \centering
\includegraphics[width=0.8\linewidth]{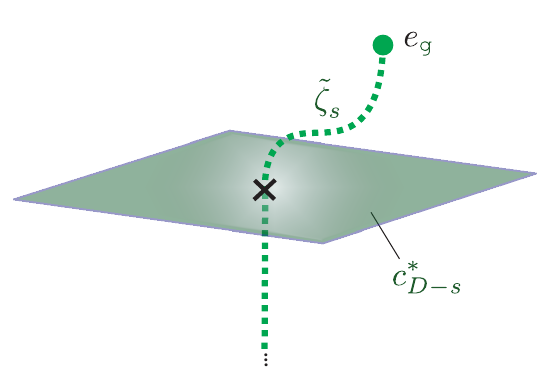}.
\ee
Extra ($-1$)-factor can be multiplied similarly for other world-sheets of charge $e_{\gr}$ intersecting with the $\gr$-type cycle (membrane) $c^*_{D-s}$.   Similarly, one can multiply the ($-1$)-factor from the world-sheet of charge $e_{\rd}$ and $e_{\bl}$ intersecting the $\rd$-type and $\bl$-type cycles (membranes) $\red{a^*_{D-p}}$ and $\blue{b^*_{D-q}}$ respectively.

\subsection{Encoding rate and distance for the explicit construction of the twisted hypergraph-product code}\label{sec:rate_and_distance}

\subsubsection{Encoding rate}

Since we have already constructed the specific CW (simplicial) complex $\L= \L_x \otimes \L_y$ from the product of the CW (simplicial) complexes $\L_x$ and $\L_y$ built from the chain complexes of the skeleton classical codes $X_x$ and $X_y$, we can now define the specific families of qLDPC codes on    $\L= \L_x \otimes \L_y$, both the untwisted code $\C$ and the twisted (gauged) code $\tilde{\C}$.   In particular, we call the former one \textit{thickened hypergraph-product code} since the skeleton classical code is ``thickened" to a higher-dimensional CW complex, and the later one the \textit{twisted (gauged) thickened hypergraph-product code}.  

We start from the two independent (decoupled) copies of untwisted codes $\C^{\rd}$ and $\C^{\bl}$ corresponding to a $\red{\ZZ_2^{(8)}}$$\times $$\blue{\ZZ_2^{(6)}}$ gauge theory, with the entire code space being: $\C = \C^{\rd} \otimes \C^{\bl}$.
Note that for our gauging measure  protocol in Sec.~\ref{sec:parallel_gauging}, we do not include the $\gr$ copy in the ungauged code space, instead the $\gr$ copy is merely used for the gauging measurement and does not store any logical information.  The qubits are defined on $8$-cells and $6$-cells 
respectively.   We have the following lemma:
\begin{lemma}\label{lemma:rate_untwisted}
The untwisted $[[n, k, d]]$ thickened hypergraph-product code  $\C=\C^{\rd} \otimes \C^{\bl}$  has constant rate, i.e.,  $k = \Theta(n)$. 
\end{lemma}
\begin{proof}
We have chosen the skeleton $[\bar{n}, \bar{k}, \bar{d}]$ 0-cocycle (transposed) classical code $\bar{\C}^*_{\text{c}}=\text{Ker}(d) =\text{Ker}(\Hs^T)$ to be a good classical LDPC code, i.e., with linear dimension (constant rate) $\bar{k} = \Theta(\bar{n})$ and linear distance $\bar{d}=\Theta(\bar n)$.  For our convenience, we also choose the parity check matrix $\Hs^T$ to be full-rank, such that the 1-cycle code $\bar{\C}_{\text{c}}=\text{Ker}(\partial) =\text{Ker}(\Hs)$ has dimension 0.

For the skeleton classical code  $\bar{\C}^*_{\text{c}}$,   we can choose a 1-cycle basis $\{\as_1\}$, a 1-cocycle basis $\{\as^1\}$, a 0-cycle basis $\{\bs_0\}$ and a 0-cocycle basis $\{\bs^0\}$. Note that the $\bs^0$ is the codeword of the 0-cocycle classical code $\bar{\C}^*_{\text{c}}=\text{Ker}(d) =\text{Ker}(\Hs^T)$.
The code dimension $k$ is encoded into  the 0th $\ZZ_2$ Betti number equaling to the dimension of the 0th $\ZZ_2$ homology group:
\begin{align}
\bar{k}=& \bar{b}_0(X; \ZZ_2) = \dim (H_0(X; \ZZ_2))  =  \dim (H^0(X; \ZZ_2)) \cr
=& \Theta(\bar n),
\end{align}
where $X$ is the associated $\ZZ_2$ chain complex of $\bar{\C}_{\text{c}}$.  Note that for $\ZZ_2$ (co)homology, there is an isomorphism between the $i^\text{th}$ homology and cohomology: $H_i(X; \ZZ_2) \cong H^i(X; \ZZ_2)$.   Since we have chosen $\Hs^T$ to be full-rank, there exists no non-trivial 1-cycle $\as_1$ or 1-cocycle $\as^1$. Just to be general enough, we still include them in the following discussion.

We have the following mapping between the (co)cycles in the associated chain complex $X_x$ of the skeleton classical code $\bar{\C}_{\text{c}}$ and those in the CW complex $\L_x$ built from $X_x$:
\begin{align}\label{eq:skeleton_mapping}
\bar{\bs}_0  \rightarrow \bs_2 \sim &{\bs^*}^6,  \qquad \bar{\bs}^0  \rightarrow \bs^2 \sim \bs^*_6. \cr
\bar{\as}_1  \rightarrow \as_3 \sim & {\as^*}^5,  \qquad \bar{\as}^1  \rightarrow \as^3 \sim \as^*_5,  
\end{align}
Here, ${\bs^*}^6 \equiv PD(\as_1) \in H^6(\L_x^*; \ZZ_2)$ is the Poincar\'e dual cocycle on the dual CW complex $\L_x^*$ of the 2-cycle $\bs_2 \in H_2(\L_x; \ZZ_2)$ on the original complex $\L_x$, and $\bs^*_6 \in H_6(\L_x^*; \ZZ_2)$ is the Poincar\'e dual cycle on the dual CW complex $\L_x^*$ of the 2-cocycle $\bs^2 \in H^2(\L_x; \ZZ_2)$ on the original complex $\L_x$. Similar duality relations hold in the other mappings above.  Note that the same mappings and isomorphism as above occurs between the (co)cycles in the chain complex $X_y$ also associated with $\bar{\C}_{\text{c}}$ and the CW complex $\L_y$, where we use $\as'_3$, $\bs_2'$ etc. to label the (co)cycles in $\L_x$ as we did previously in Eq.~\eqref{eq:cocycle_decomposition}.    Note that in our construction, we have chosen $\L_x = \L_y$ for simplicity. 

The above isomorphisms  also occur between the entire (co)homology groups of $X_x$ and $\L_x$: 
\begin{align}
H_0(X_x; \ZZ_2)  \cong H_2(\L_x; \ZZ_2) \cong
H^6(\L^*_x; \ZZ_2), \cr   
H^0(X_x; \ZZ_2)  \cong H^2(\L_x; \ZZ_2) \cong H_6(\L^*_x; \ZZ_2). \cr
H_1(X_x; \ZZ_2)  \cong H_3(\L_x; \ZZ_2) \cong H^5(\L^*_x; \ZZ_2), \cr
H^1(X_x; \ZZ_2)  \cong H^3(\L_x; \ZZ_2) \cong H_5(\L^*_x; \ZZ_2), 
\end{align}
This leads to the following identity of the Betti numbers:
\be
\bar{b}_0= b_2 = b_6, \quad \bar{b}_1= b_3 = b_5.
\ee

For the code defined on the {\rd} copy $\C^{\rd}$, which we call the primary code, the qubits are defined on 8-cells of the product complex $\L = \L_x \otimes \L_y$, and hence the code dimension is determined by the $8^\text{th}$ Betti number.   Due to the K\"unneth formula, we have
\begin{align}
& H_8(\L_x \otimes \L_y; \ZZ_2) = \bigoplus_{i+j=8} H_i(\L_x; \ZZ_2) \otimes H_j(\L_y; \ZZ_2)  \cr  
=&   [H_2(\L_x; \ZZ_2) \otimes H_6(\L_y; \ZZ_2)] \oplus [H_3(\L_x; \ZZ_2) \otimes H_5(\L_y; \ZZ_2)]  \cr
& \oplus \cdots
\end{align}
We can hence have the following lower bound for the code dimension $k^{\rd}$ for the $\rd$ copy through the Betti numbers:
\begin{align}
k^{\rd} = & \tilde{b}_8 =  b_2 \cdot b'_6 + b_3 \cdot b'_5 + \cdots  \cr
\ge & b_2 \cdot b'_6 = \bar{b}_0 \cdot \bar{b}'_0 = \Omega(\bar{n}) \cdot \Omega(\bar{n}) =\Omega(n), 
\end{align}
where we have used the fact that $n=\Theta(\bar{n}^2)$.   Since linear dimension is already optimal, we have the following asymptotic scaling:
\be
k^{\rd}= \Theta(n).
\ee

We then consider the code defined on the $
\bl$ copy $\C^{\bl}$, which we call the secondary code.  The qubits are defined on the 6-cells,  hence the code dimension is encoded into the $6^\text{th}$ Betti number $\tilde{b}_6$.  According to the K\"unneth formula, we have
\begin{align}
& H_6(\L_x \otimes \L_y; \ZZ_2) = \bigoplus_{i+j=6} H_i(\L_x; \ZZ_2) \otimes H_j(\L_y; \ZZ_2)  \cr  
=&  [H_6(\L_x; \ZZ_2) \otimes H_0(\L_y; \ZZ_2)] \oplus [H_0(\L_x; \ZZ_2) \otimes H_6(\L_y; \ZZ_2)] \cr
& \oplus \cdots
\end{align}
Therefore, we get the following lower bound for the code dimension and Betti number of the {\bl} copy:
\be
k^{\bl} = \tilde{b}_6 \ge  b_6 \cdot b'_0 =  \Omega(\bar{n}) \cdot 1 = \Omega(\sqrt{n}).
\ee
Note that compare to copy $\rd$ with $\Theta(n)$ logical qubits,  copy $\bl$ only 
only have $\Omega(\sqrt{n})$ logical qubits due to the fact that the zeroth Betti number of the CW complex is only 1, i.e., $b'_0=1$.   Nevertheless, the total number of logical qubits in the code space $\C= \C^{\rd} \otimes \C^{\bl} $ is the sum of these two, i.e.,
\be
k =  k^{\rd} + k^{\bl} = \Theta(n).
\ee

\end{proof}

We then consider the twisted (gauged) thickened hypergraph-product code with the Clifford stabilizer group specified by Eq.~\eqref{eq:stabilizer_summary}, and have the following lemma:
\begin{lemma}\label{lemma:rate_twisted}
The twisted thickened hypergraph-product code $\tilde{\C}$ with parameter $[[\tilde{n}=\Theta(n), \tilde{k}, \tilde{d}]]$  has constant rate, i.e.,  $\tilde{k} = \Theta(\tilde{n})=\Theta(n)$. 
\end{lemma}
\begin{proof}
The twisted (gauged) code $\tilde{\C}$ contains the $\rd$, $\bl$ and $\gr$ lattice copies which in this case are coupled together.
We can define the dimension $\tilde k$ of the twisted code via $\tilde k=\log_2(\text{GSD})$, where $\text{GSD}$ represents the ground-state degeneracy of the parent Hamiltonian equaling the dimension of the code space $\tilde \C$.

Similar to the untwisted code $\C$, the logical qubits in the twisted code $\tilde{\C}$ can also be labeled by the logical-$Z$ operators in Eq.~\eqref{eq:logical_Z_1} and \eqref{eq:logical_Z_2}, or equivalently the basis cycles where they are supported:  $\red{\alpha_8}$, $\blue{\beta_6}$. Note that for the $\gr$ copy, there is no nontrivial 3-cycle $\green{\xi_3} \neq 0$, therefore there is no logical qubits due to the $\green{green}$ physical qubits defined on 2-cells $\sigma_2$.

A large portion of logical gates in the untwisted codes $\C^{\rd}$ and $\C^{\bl}$ are inherited by the twisted code $\tilde{\C}$, while certain logical qubit states in the untwisted code that do not satisfy the constraints of the twisted code space $\tilde{\C}$ in Eq.~\eqref{eq:constraint_operator}, or equivalently Eqs.~\eqref{eq:CZ_identity1} and \eqref{eq:CZ_identity2}, are projected out of the twisted code space.   

The counting of the dimension of the twisted code space  $\tilde{\C}$ can hence be obtained by subtracting the number of independent constraints mentioned above from the dimension of the untwisted code space $\C$.   For our current construction,  the $\gr$ copy does not encode any logical qubit.  The constraints involving $\gr$ are hence all trivialized since $\widetilde{\text{CZ}}^{\rd, \gr}_{\blue{\beta^6}} =1$ and $ \widetilde{\text{CZ}}^{\bl, \gr}_{\red{\alpha^8}} =1$ are always satisfied.   Therefore, the only remaining non-trivial  constraints are $ \widetilde{\text{CZ}}^{\rd, \bl}_{\green{\gamma^2}} = 1$, and we only need to count the total number of such independent constraints.   According to Eq.~\eqref{eq:cocycle_decomposition},  we have the tensor decomposition $\green{\gamma^2}=\cs^0 \otimes  \bs'^2 $ or $\bs^2 \otimes c^0$. The size of the cocycle basis $\{\green{\gamma^2}\}$ is hence:
\be
|\{\green{\gamma^2}\}| = 1 \cdot \Theta(\bar{n}) + \Theta(\bar{n}) \cdot 1= \Theta(\bar{n}) =  \Theta(\sqrt{n}).
\ee
Therefore, we have $\Theta(\sqrt{n})$ independent constraints.  The code dimension of the twisted code is hence:
\be
\tilde{k} = k - \Theta(\sqrt{n}) = \Theta(n) = \Theta(\tilde{n}),
\ee
which again has constant rate.

\end{proof}

\subsubsection{Code distance}

We first investigate the code distance of the untwisted code $\C$.   We have the following lemma:

We note that there exists spurious 1-(co)cycles
and 7-(co)cycles in the CW complex $\L_x$ and $\L_y$ built from the classical codes, as mentioned in Sec.~\ref{sec:code_to_manifold}.  For example, a cycle of the form $\fs_1 \otimes \fs'_7$ is a spurious cycle with distance of $O(1)$ where $\fs_1$ and $\fs'_7$ are spurious cycles from $\L_x$ and $\L_y$ respectively.

Nevertheless, as has been established in Ref.~\cite{zhu2025topological, zhu2025transversal}, we can use a subsystem-code encoding which selects only a subset of cycles and their conjugate co-cycles to encode logical $Z$ and $X$ operators, while treating the rest of logical degree of freedom associated with logical operators of shorter cycles as \textit{gauge qubits}.

We hence have the following lemma:
\begin{lemma}\label{lemma:subsystem}
    For a quantum code defined on a  $\ds$-dimensional Poincar\'e CW complex $\L$, one can define a subsystem code by associating the logical-$Z$ operators with a subset of an $k^\text{th}$ homology basis $\{\eta_k\}$ and the conjugate logical-$X$  operators on the conjugate subset of $k^\text{th}$ cohomology basis $\{\eta'^{k}\}$ satisfying the intersection relation $|\eta_k \cap \eta'^k|\equiv \int_{\eta_k} \eta'^k= \delta_{\eta, \eta'}$.  The distance of the subsystem code is hence $d=\text{min}(\text{min}\{|\eta_k|\}, \text{min}\{|\eta'^{k}|\})$.
\end{lemma}
See Appendix \ref{app:subsystem} for the proof of this lemma and the explanation of the intuition behind. 

We then show the following lemma about the distance of the untwisted code $\C$:
\begin{lemma}\label{lemma:distance_untwisted}
The untwisted $[[n, k, d]]$ thickened hypergraph-product code  $\C=\C^{\rd} \otimes \C^{\bl}$  has subsystem-code distance scaling $d = \Omega(\sqrt{n})$. 
\end{lemma}
\begin{proof}
We first consider the logical-$X$ operators supported on the subsets of cocycle basis $\{\red{\alpha^8}\}$ and $\{\blue{\beta^6}\}$ with the tensor product decomposition $\red{\alpha^8} = \bs^2 \otimes {\bs'^*}^6$ and  $\blue{{\beta}^{6}} = {\bs^*}^6  \otimes  \cs^0$ from Eq.~\eqref{eq:cocycle_decomposition}. 
Note that there is no decomposition of the 8-cocycle of the form $\as^3 \otimes {\as'^*}^5$ since there exists no non-trivial cocycle class $\as^3 \neq 0$ or ${\as'^*}^5 \neq 0$ in the complex $\L$ or $\L'$.

We have the following lower bounds on the subset of cocycle basis $\{\red{\alpha^8}\}$ from the {\rd} copy:
\begin{align}
\min\{|\red{\alpha^8}|\} =& \min\{|\bs^2 \otimes {\bs'^*}^6|\}    \cr
\ge &  \max\big( \min\{|\bs^2|\}, \min\{|{\bs'^*}^6|\}\big)   \cr
=&  \max\big( \min\{|\bar{\bs}^0|\}, \min\{|{\bar{\bs}'_0}|\}\big)  \cr
=& \max\big( \Omega(\bar{n}), 1 \big) 
=\Omega(\bar{n}) = \Omega(\sqrt{n}), \cr
\end{align}
where the inequality in the second line has been derived in Ref.~\cite{Bravyi:2014bq}. In the third line, we have used the mapping of (co)cycles between the skeleton classical code and the CW complex $\L$ from Eq.~\eqref{eq:skeleton_mapping} and their equaling Hamming weight  $|\bs^2|=|\bar{\bs}^0|$ and $|{\bs'^*}^6|=|\bs'_2|=|\bar{\bs}'_0|$.  For the last line, we have used the fact that $\bar{\bs}^0$ is the codeword of the asymptotically good skeleton 0-cocycle classical  code and the distance of the 0-cocycle code is $\bar{d} =\min\{|\bar{\bs}^0|\} = \Theta(\bar{n})$, as well as the fact that $\bar{\bs}'_0$ is just occupies a single 0-chain in the chain complex $X_x$ of the skeleton classical code $\bar{\C}_{\text{c}}$.                                                                                               

Similarly, for the subset of cocycle basis $\{\blue{\beta^6}\}$ from the {\bl} copy,  we have
\begin{align}
\min\{|\blue{\beta^6}|\} =& \min\{|{\bs^*}^6  \otimes  \cs^0|\}    \cr
\ge &  \max\big( \min\{|{\bs^*}^6 |\}, \min\{|\cs^0|\}\big)   \cr
=&  \max\big( \min\{|\bar{\bs}_0|\}, \min\{|{\cs^0}|\}\big)  \cr
=& \max\big( \Omega(1), \Omega(\bar{n}) \big) 
=\Omega(\bar{n}) = \Omega(\sqrt{n}), \cr
\end{align}
where in the last line we have used the fact that $\cs^0$ is the 0-cocycle supported on all the 0-cells of the CW complex $\L_y$ and hence have size $\Omega(\bar{n})$. 

By symmetry, we also have the other subsets of cocycle basis $\{\red{\alpha'^8}\}$ and $\{\blue{\beta'^6}\}$  with the form $\red{\alpha'^8} =   {\as^*}^5 \otimes \bs'^2$ and  $\blue{{\beta'}^{6}} =  \cs^0 \otimes {\bs^*}^6$  by exchanging the tensor factor of $\red{\alpha^8}$ and $\blue{\beta^6}$ in $\L_x$ and $\L_y$. We can hence obtain the same lower bound for them:   $\min\{|\red{\alpha'^8}|\} = \Omega(\sqrt{n})$ and $\min\{|\blue{\beta'^6}|\} = \Omega(\sqrt{n})$. 

The subsets of cocycle basis above form the complete logical-$X$ basis of the thickened hypergraph-product code, and we hence have the following $X$-distance:
\be
d_X = \min\{|\red{\alpha^8}|, |\blue{\beta^6}|, |\red{\alpha'^8}|, |\blue{\beta'^6}|\}  = \Omega(\sqrt{n}).
\ee

We now consider the logical-$Z$ operators supported on the subsets of conjugate cycle basis $\{\red{\alpha_8}\}$ and $\{\blue{\beta_6}\}$ with the form $\red{\alpha_8} = \bs_2 \otimes \bs'^*_6$ and  $\blue{{\beta}_6} = \bs^*_6  \otimes  \cs_0$:
\begin{align}
\min\{|\red{\alpha_8}|\} =& \min\{|\bs_2 \otimes {\bs'^*_6}|\}    \cr
\ge &  \max\big( \min\{|\bs_2|\}, \min\{|{\bs'^*_6}|\}\big)   \cr
=&  \max\big( \min\{|\bar{\as}_1|\}, \min\{|{\bar{\as}'^1}|\}\big)  \cr
=& \max\big( \Omega(\bar{n}), 1 \big) 
=\Omega(\bar{n}) = \Omega(\sqrt{n}), \cr
\min\{|\blue{\beta_6}|\} =& \min\{|\bs^*_6  \otimes  \cs_0|\}    \cr
\ge &  \max\big( \min\{|\bs^*_6 |\}, \min\{|\cs_0|\}\big)   \cr
=&  \max\big( \min\{|\bar{\bs}_0|\}, \min\{|{\cs_0}|\}\big)  \cr
=& \max\big( \Omega(\bar{n}), 1 \big) 
=\Omega(\bar{n}) = \Omega(\sqrt{n}), \cr
\end{align}
where we have used the fact that $\cs_0$ is a single vertex (0-cell) in the CW complex $\L_y$.  By symmetry, we also have the other subsets of cycle basis  $\{\red{\alpha'_8}\}$ and $\{\blue{\beta'_6}\}$ with the form $\red{\alpha'_8} =  \bs'^*_6 \otimes \bs_2 $ and  $\blue{{\beta}'_6} = \cs_0  \otimes  \bs^*_6$, and the same lower bound holds: $\min\{|\red{\alpha'_8}|\} = \Omega(\sqrt{n})$ and $\min\{|\blue{\beta'_6}|\} = \Omega(\sqrt{n})$.   

The subsets of cycle basis above form the complete logical-$Z$ basis of $\tilde{\C}$, and we obtain the following $Z$-distance: 
\be
d_Z = \min\{|\red{\alpha_8}|, |\blue{\beta_6}|, |\red{\alpha'_8}|, |\blue{\beta'_6}|\}  = \Omega(\sqrt{n}),
\ee
as well as the overall distance:
\be
d=\min(d_X, d_Z). 
\ee
\end{proof}

Unlike the untwisted code $\C$, the twisted code $\tilde{\C}$ is not a CSS code, and hence the code distance is not simply $\min(d_X, d_Z)$.  Instead, we should use the more general definition of code distance, that is the shortest support on all possible logical operators.  Let us formalize this definition below in the specific context of the Clifford stabilizer codes:
\begin{definition}
For a Clifford stabilizer code $\tilde{\C}$ with the corresponding Clifford stabilizer group $\tilde{\mathcal{S}}$, we define logical operators $L \in A$ to be the operators  that commute with the entire stabilizer group $\tilde{\mathcal{S}}$ when projected to the code space:
\be
P_{\tilde{\C}}  [L, \tilde{\cS}] P_{\tilde{\C}}  =1,
\ee
and do not act as a logical identity:   $P_{\tilde{\C}} L P_{\tilde{\C}} \neq \bar{I}$.
Here, $[\bullet]$ is the group commutator, and $A$ is the set containing all possible logical operators.   The distance of the code $\tilde{\C}$ is defined as the minimal weight of all logical operators
\be
d = \min\{|\text{supp}(L)|:  L \in A\}.
\ee
\end{definition}

Similar to Lemma \ref{lemma:subsystem}, we can also define the subsystem-code distance when only choosing a subset of the logical operators $A' \subset A$.

We will show that the distance of the twisted code $\tilde {\cal C}$ scales in the same way as in the untwisted code ${\cal C}$.  Our proof will rely on some known generic properties of TQFT, and assume that should still hold for our code since it is a combinatorial TQFT.
In particular, the logical operators of the twisted code $\tilde{\C}$ correspond to topological operators $L$ (also called generalized symmetries) of the corresponding TQFT, which preserve the code space (TQFT Hilbert space) $\tilde{\C}$, i.e.,
\be
L : \tilde{\C} \rightarrow \tilde{\C}.
\ee
In this case, the corresponding TQFT is the $\red{\ZZ_2^{(p)}} \times \blue{\ZZ_2^{(q)}} \times \green{\ZZ_2^{(s)}}$ twisted gauge theory.  We now have the following statement from TQFT:
\begin{statement}\label{statement:TQFT}
 The topological operators of TQFTs are the basic electric and magentic operators that describe the superselection sectors, together with their condensation descendants \cite{Gaiotto:2019xmp,Kong_2020,Kong:2024ykr} which correspond to summing over the condensed operators on the support of the condensation (see e.g. \cite{Roumpedakis:2022aik,Choi:2022zal,Cordova:2024mqg}):
 the condensation of the electric and/or magnetic operators on a cycle ${\zeta}$ is the operator
\begin{equation}
    {\cal C}({\zeta})=\sum_{ \eta,\Sigma \in {\zeta}}\alpha_{q_e,q_m}(\eta,\Sigma) W^{q_e}(\eta) M^{q_m}(\Sigma)~,
\end{equation}
where the summation is over the cycles $\eta, \Sigma$ belonging to $\zeta$ that support the electric and magnetic operators $W^{q_e},M^{q_m}$ with $q_e,q_m=0,1$ determining the condensate. $\alpha_{q_e,q_m}(\eta,\Sigma)$ is a numerical factor that gives weight for different cycles $\eta,\Sigma$ in the sum.
We have omitted an overall normalization factor outside the above sum.
\end{statement}
In the above statement, the superselection sectors correspond to states that are not locally distinguishable. Different superselection sectors are related by non-contractible electric and magnetic logical operators.

The above statement is based on remote detectability principle of topological orders \cite{PhysRevX.8.021074,Johnson_Freyd_2022}, which says that well-defined topological orders that do not require a bulk, such as the TQFT in our case, are ``centerless'': the only operators that braid trivially with all other operators are condensation descendants. For example, a domain wall braids trivially with other operators since there are no nontrivial local operators, and thus the domain wall operators must be condensation descendants.

For example, the transversal CZ operator $(-1)^{\int_{{\zeta}} \blue{b^q} \cup \green{\hat{c}^s}}$ can be written as a condensation operators as: \cite{Roumpedakis:2022aik,Cordova:2024mqg}
\begin{align}
    &(-1)^{\int_{{\zeta}} \blue{\hat{b}^q} \cup \green{\hat{c}^s}}=\sum_{u,v} (-1)^{\int_{{\zeta}} \blue{\hat{b}^q} \cup v}(-1)^{\int_{{\zeta}} u \cup \green{\hat{c}^s}}(-1)^{\int_{{\zeta}} u\cup v}\cr
    &=\sum_{\eta=\text{PD}_{\zeta}(u),\eta'=\text{PD}_{\zeta}(v)} (-1)^{\int_{\eta'} \blue{\hat{b}^q}}(-1)^{\int_{\eta} \green{\hat{c}^s}}(-1)^{|\eta \cap \eta'|}~,
\end{align}
where $u,v$ are $\mathbb{Z}_2$ cocycles, $\text{PD}_{\zeta}(u)$ is the Poincar\'e dual of $u$ on ${\zeta}$, and $|\eta \cap \eta'|$ is the $\ZZ_2$ intersection number of $\eta,\eta'$ on ${\zeta}$. Thus the CZ operator $(-1)^{\int_{{\zeta}} \blue{\hat{b}^q} \cup \green{\hat{c}^s}}$ is a sum of the Wilson operators.

Based on the above statement, we can hence obtain the following Lemma:
\begin{lemma}\label{lemma:distance_twisted}
   The twisted thickened hypergraph-product code $\tilde{\cal C}$   with parameters $[[\tilde{n}=\Theta(n), \tilde{k}, \tilde{d}]]$ has  subsystem-code distance scaled as $\tilde{d}=\Omega(\sqrt{\tilde{n}})=\Omega(\sqrt{{n}})$,  assuming Statement \ref{statement:TQFT} holds for $\tilde{\cal C}$. 
\end{lemma}
\begin{proof}

As we have shown in Eqs.~\eqref{eq:logical_Z_1}, \eqref{eq:logical_Z_2} and \eqref{eqn:magoprCZ},  there are two basic types of logical operators (1) The $Z$-type logical operators corresponding to the basic  electric operators $\widetilde{Z}^{\rd}_{\red{\alpha_8}}\equiv W^{\rd}(\red{\alpha_8})$,  $\widetilde{Z}^{\bl}_{\blue{\beta_6}}\equiv W^{\bl}(\blue{\beta_6})$, and $\widetilde{Z}^{\gr}_{\green{\xi_3}}\equiv W^{\gr}(\green{\xi_3})$ etc.; (2) The dressed $X$-type logical operators corresponding to the basic magnetic operators ${\widetilde{\mathcal{X}}}^{
 \rd}_{\red{\alpha^8}} \equiv   M^{
 \rd}(\red{\alpha^8})$,  ${\widetilde{\mathcal{X}}}^{
 \bl}_{\blue{\beta^6}} \equiv   M^{
 \bl}(\blue{\beta^6})$ and ${\widetilde{\mathcal{X}}}^{
 \gr}_{\green{\xi^3}} \equiv   M^{
 \gr}(\green{\xi^3}
 )$ etc.   Note that the $Z$-type logical operators are the same as those on the 
untwisted code $\C$ and hence have exactly the same minimal weights.  On the other hand, the dressed $X$-type logical operators $\widetilde{\cX}^i_{\eta^k}$  has the form of the bare $X$-type logical operators $\widetilde{X}^i_{\eta^k}$ of the untwisted code $\C$ dressed by additional projectors invovling Pauli-$Z$'s as shown in Eq.~\eqref{eqn:magneticoprprojection0}.  Therefore, the dressed $X$-type logical operators have strictly larger support than the bare $X$-type logical operators, i.e.,
\be\label{eq:support_inequality}
|\text{supp}(\widetilde{\cX}^i_{\eta^k})|  \ge  |\text{supp}(\widetilde{X}^i_{\eta^k})| =|\eta^k|,
\ee
for $i=\rd,\bl,\gr$ and $k= \red{8}, \blue{6}, \green{3}$ respectively.

According to Statement \ref{statement:TQFT},  
since the condensation descendants correspond to summing over the condensed operators on the support of the condensation, which has higher dimensions compared to these condensed operators, the condensation descendants necessarily have higher weight compared to the basic electric and magnetic operators. Thus the minimal-weight logical operators arise from the basic electric and magnetic operators,  which have minimal weight lower bounded by the minimal weight of the logical-$Z$ and -$X$ operators in the three independent copies of untwisted codes $\C^{\rd} \otimes \C^{\bl} \otimes \C^{\gr}$.  Note that in our construction the green copy $\C^{\gr}$ does not encode any useful logical qubit, and is purely auxiliary for the gauging measurement protocol presented later in Sec.~\ref{sec:parallel_gauging}.

We will choose the same subsets of cycle basis $\{\red{\alpha_8}\}$, $\{\red{\alpha'_8}\}$, $\{\blue{\beta_6}\}$, $\{\blue{\beta'_6}\}$,  as well as their conjugate cocycle basis $\{\red{\alpha^8}\}$, $\{\red{\alpha'^8}\}$, $\{\blue{\beta^6}\}$, $\{\blue{\beta'^6}\}$,  to define the support of the electric ($Z$-type) and magnetic (dressed $X$-type) logical operators respectively as in the untwisted codes $\C^{\rd} \otimes \C^{\bl} \otimes \C^{\gr}$.
We hence have the subsystem-code distance of the twisted hypergraph-product code $\tilde{\C}$ as
\begin{align}
\tilde{d} =& \min\{ |\text{supp(}\widetilde{Z}^{\rd}_{\red{\alpha_8}})|, |\text{supp(}\widetilde{Z}^{\bl}_{\blue{\beta_6}})|, 
|\text{supp(}\widetilde{Z}^{\rd}_{\red{\alpha_8}})|,\cr
&|\text{supp(}\widetilde{Z}^{\bl}_{\blue{\beta_6}})|,  
|\text{supp(}\widetilde{\cX}^{\rd}_{\red{\alpha^8}})|, |\text{supp(}\widetilde{\cX}^{\bl}_{\blue{\beta^6}})|,  
|\text{supp(}\widetilde{\cX}^{\rd}_{\red{\alpha'^8}})|,  \cr
&|\text{supp(}\widetilde{\cX}^{\bl}_{\blue{\beta'^6}})| 
\}              \cr
\ge & \min\{|\red{\alpha_8}|, |\blue{\beta_6}|,  |\red{\alpha'_8}|, |\blue{\beta'_6}|,   
 |\red{\alpha^8}|, |\blue{\beta^6}|,   |\red{\alpha'^8}|, |\blue{\beta'^6}|  \}   \cr
=& \Omega(\sqrt{\tilde{n}}) = \Omega(\sqrt{n}),    
\end{align}
where the inequality comes from Eq.~\eqref{eq:support_inequality}.

\end{proof}

In the above discussion, we do not consider storing any logical information into the $\gr$  copy, either in the untwisted code $\C^{\gr}$ or the twisted code $\tilde{\C}$.  Nevertheless, due to the presence of spurious logical (co)cycles in the $\gr$ copy,  we still need to analyze its distance since it may affect the gauging measurement protocol introduced later in Sec.~\ref{sec:parallel_gauging}.  

First, we have the following lemma for the distance of the untwisted code $\C^{\gr}$: 
\begin{lemma}\label{lemma:green_distance_untwisted}
The untwisted thickened hypergraph-product code  $\C^{\gr}$ supported on the $\gr$ copy of the lattice has $X$-distance lower bound $d_X = \Omega(\sqrt{n})$ and $Z$-distance lower bound $d_Z=\Omega(1)$. 
\end{lemma}
\begin{proof}
The qubits are defined on the 3-cells of $\L$. There are two subsets of co-cycle basis $\{\green{\xi^3}\}$ and $\{\green{\xi'^3}\}$ as  the support of logical-$X$ operators, which have the form:
\be
\green{\xi^3} = \fs^1 \otimes \bs'^2,  \qquad
\green{\xi'^3} = \bs^2 \otimes \fs'^1,
\ee
where $\fs^1$ and $\fs'^1$ are the spurious basis 1-cocycles from $\L_x$ and $\L_y$ respectively,  which could have Hamming weight only $\Omega(1)$. Note that there is no basis 3-cocycle of the form $\cs^0 \otimes \as'^3$ or $\as^3 \otimes \cs^0$, since there is no non-trivial 3-cocycle $\as^3 \neq 0$ by construction. As explained before $\bs^2$ and $\bs'^2$ are basis cocycles which correspond to the codeword of the asymptotically good 0-cocycle classical code $\bar{\C}_{\text{c}}$: $\bar{\bs}^0$ and $\bar{\bs'}^0$.  The hamming weight hence scales as:
\be
|\bs^2| = |\bar{\bs}^0| = \Theta(\bar{n}),  \quad |\bs'^2| = |\bar{\bs'}^0| = \Theta(\bar{n}).
\ee
We can hence bound the Hamming weight of the basis 3-cocycles as
\begin{align}
\min\{|\green{\xi^3}|\} =& \min\{|\fs^1 \otimes \bs'^2|\} \cr
\ge &  \max\big( \min\{|\fs^1|\}, \min\{|{\bs'^2}|\}\big)   \cr
=& \max (\Omega(1),  \Theta(\bar{n}))
= \Omega(\bar{n}) = \Omega(\sqrt{n}), \cr
\end{align}
and similarly $\min\{|\green{\xi'^3}|\} =\Omega(\sqrt{n}) $.  We hence obtain the $X$-distance as
\be
d_X=\min\{|\green{\xi^3}|, |\green{\xi'^3}|\} = \Omega(\sqrt{n}). 
\ee

Now the logical-$Z$ are supported on the conjugate cycle basis $\green{\xi_3}$ and $\green{\xi'_3}$ with the form:
\be
\green{\xi_3} = \fs_1 \otimes \bs'_2,  \qquad
\green{\xi'_3} = \bs_2 \otimes \fs'_1.
\ee
Note that both $|\fs_1|$ and $|\bs'_2|$ can have only $O(1)$ size, we can only get the following lower bound on the $Z$-distance:
\be
d^{\gr}_Z=\min\{|\green{\xi_3}|, |\green{\xi'_3}|\} = \Omega(1). 
\ee

\end{proof}

Next, we have the following lemma for the twisted code: 
\begin{lemma}\label{lemma:green_distance_twisted}
 In the twisted code $\tilde{\C}$,  the minimal support of the magnetic logical operator defined on 3-cocycles is $\Omega(\sqrt{n})$, while the minimal support of the electric ($Z$-type) logical operator supported on 3-cycles is lower bounded by $\Omega(1)$.  
\end{lemma}
\begin{proof}
We have the following sets of magnetic logical operator (or dressed $X$-logical operators) $\{\widetilde{\cX}_{\green{\xi^3}}\}$ and $\{\widetilde{\cX}_{\green{\xi'^3}}\}$ defined on the subsets of 3-cocycle basis $\{\green{\xi'^3}\}$ and $\{\green{\xi'^3}\}$ respectively.  

According to Eq.~\eqref{eq:support_inequality}, we have the dressed $X$-logical operators in the twisted code $\tilde{\C}$ larger or equal to the bare $X$-logical operators in the untwisted code $\C^{\gr}$:
\begin{align}
|\text{supp}(\widetilde{\cX}^{\gr}_{\green{\xi^3}})|  \ge  |\text{supp}(\widetilde{X}^{\gr}_{\green{\xi^3}})| =|\green{\xi^3}|,    \cr
|\text{supp}(\widetilde{\cX}^{\gr}_{\green{\xi'^3}})|  \ge  |\text{supp}(\widetilde{X}^{\gr}_{\green{\xi'^3}})| =|\green{\xi'^3}|. 
\end{align}
We hence have the minimal support of the magnetic logical operator as:
\begin{align}
d^{\gr}_{\cX} :=&  \min\{ |\text{supp}(\widetilde{\cX}^{\gr}_{\green{\xi^3}})|, |\text{supp}(\widetilde{\cX}^{\gr}_{\green{\xi'^3}})|   \} \cr
\ge& \min\{|\green{\xi^3}|, |\green{\xi'^3}|\} \cr
=& \Omega(\sqrt{n}).
\end{align}

On the other hand,  the electric ($Z$-type) logical operator is identical to the electric ($Z$-type) logical operator is the untwisted code, so it has the same lower bound on the minimal support:
\begin{align}
d^{\gr}_Z=&\min\{ |\text{supp}(\widetilde{Z}^{\gr}_{\green{\xi_3}})|, |\text{supp}(\widetilde{Z}^{\gr}_{\green{\xi'_3}})|   \} \cr
=&\min\{|\green{\xi_3}|, |\green{\xi'_3}|\} = \Omega(1). 
\end{align}

\end{proof}

As mentioned above, since the twisted code $\tilde{\C}$ is not a CSS code, we do not analyze the error correction properties by considering $Z$ and $X$ errors as usual.  Instead, we can decompose local stochastic  errors into two types: the \textit{electric errors} and the \textit{magnetic errors}, which do not have commute with certain dressed $X$-stabilizers and $Z$-stabilizers respectively.   We can hence have the notion of the electric distance $d_Z$ and the magnetic distance $d_{\cX}$,  in analogy to the $Z$-distance and $X$-distance in the case of CSS code.  

 \subsection{Pullback to the skeleton hypergraph-product code on a general 2D chain complex}\label{sec:pullback_skeleton}

\subsubsection{Truncating the CW complex and pullback of the cup product}
 In the above subsections, we follow the procedure of first mapping each two-term chain complex $X_x$ and $X_y$ associated with the skeleton classical code $\bar{\C}_{\text{c}}$ (or equivalently its transpose $\bar{\C}^*_{\text{c}}$) to a higher-dimensional (8D) Poincar\'e CW complex $\L_x$ and $\L_y$, and then taking the tensor product to form the 16D Poincar\'e CW complex $\L_x \otimes \L_y$ which is used to define the `thickened' hypergraph-product code $\C$.  This is followed by gauging the thickened hypergraph-product code to a twisted code $\tilde{\C}$.  

The classical code defined on the 8D CW complex $\L_x$ or $\L_y$ is identical to the skeleton code $\bar{\C}_{\text{c}}$, while there still exit extra cell structure with dimension other than 2 (check) and 3 (dimension) in the CW complex $\L_x$ or $\L_y$ compared to the skeleton 1D (2-term) chain complex $X_x$ or $X_y$.  When taking the tensor product $\L_x \otimes \L_y$, the product of these extra cells, e.g., $\sigma_1 \otimes \sigma_7$, can have the same dimension as the qubits or checks, which just introduces extra redundancy to the quantum code although costing only a constant overhead.   A natural question is that whether we can get rid of these redundant structure and reduce the untwisted (twisted) thickened hypergraph-product code back to the untwisted (twisted) skeleton hypergraph-product described by a 2D (3-term) chain complex while preserving the cup-product structure and the path integral as a topological invariant.

As we can see from Eqs.~\eqref{eq:cocycle_decomposition} and \eqref{eq:triple-cup_decomposition}, within each CW complex $\L_x$ or $\L_y$, the topological invariant associated with the path integral only requires non-trivial cup product between the cochains or cocycles of the form $\bs^2 \cup {\bs^*}^6$   and $\bs^2 \cup {\bs^*}^6 \cup \cs^0$.  We will also need the cup product of the form $\as^3 \cup {\as^*}^5$   and $\as^3 \cup {\as^*}^5 \cup \cs^0$ when defining the Clifford stabilizers of the twisted code.   On the other hand, we will never need non-trivial cup product involving 1-chains, 4-chains, and 7-chains.    This suggests that we can truncate the CW complex properly to get rid of the irrelevant cells without affecting the definition of the path integral topological invariant and the Clifford stabilizer codes.

Recall that the CW complex with $\ZZ_2$ coefficients has the following form
\begin{widetext}
\begin{align}\label{eq:long_chain_CW}
& C_8 \rightarrow C_7 \rightarrow C_{6} \xrightarrow[]{{\partial}_6= {\mathsf{H}}^T}  C_5 \rightarrow 0  \rightarrow C_3 \xrightarrow[]{{\partial_3} = {\mathsf{H}}} C_{2} \rightarrow C_1 \rightarrow C_0, \cr
&\qquad \qquad \ \ \text{check} \qquad \quad \ \text{bit} \qquad \ \ \ \ \  \text{bit} \ \qquad \text{check} 
\end{align}
\end{widetext}
In this construction, there is no non-trivial boundary map between $C_2$ and $C_1$.   Therefore, one truncation we can choose is to preserve only the portion containing $C_3$ and $C_2$: 
\begin{align}
 0  \rightarrow & C_3 \xrightarrow[]{{\partial_3} = {\mathsf{H}}} C_{2} \rightarrow 0, \cr
&   \text{bit} \qquad \text{check}
\end{align}
which is nothing but the chain complex $X$ associated with the skeleton classical code $\bar{\C}_c$ and is re-written as:
\begin{align}
X:   \qquad     \bar{C}_1 & \xrightarrow[]{\partial_1=\mathsf{H}} \bar{C}_{0} \quad. \cr
            \text{bit}& \qquad \ \ \text{check}
\end{align}
Due to the Poincar\'e duality symmetry of $\L$, we can also truncate $\L$ by preserving only $C_6$ and $C_5$ on the left portion: 
\begin{align}
 0 & \rightarrow  C_6  \xrightarrow[]{{\partial}_6 = {\mathsf{H}}^T} C_5 \rightarrow 0, \cr
&  \quad \text{check} \qquad \ \ \ \text{bit}
\end{align}
which is re-written into:
\begin{align}
X^*:   \qquad     \bar{C}^*_1 & \xrightarrow[]{{\partial}_1 = {\mathsf{H}}^T}     \bar{C}^*_0 \ . \cr
           \ \text{check} & \qquad \quad \  \text{bit}
\end{align}
Note that $X^*$ is the chain complex associated with the transposed classical code $\bar{\C}^*_c=\text{Ker}(\mathsf{H^T})$ where the roles of bit and check get exchanged, and the parity check matrix $\partial_1$ becomes the transposed matrix $\mathsf{H}^T$. 

In addition, in order to include the cup product with $\cs^0$,  we also need a code involving 0-cells and 1-cells.  We hence have another choice of truncation of the CW complex:
\begin{align}
 0  \rightarrow & C_1 \rightarrow  C_0 \rightarrow 0,
\end{align}
which can be rewritten as:
\begin{align}
X_R:   \qquad     \bar{C}_1 & \rightarrow  \bar{C}_{0} \quad.
\end{align}
According to the structure of the 0-cells and 1-cells shown in Fig.~\ref{fig:CW_construction}(c), $X_R$ is the chain complex corresponding to a 0-cocycle code $\C_R = \text{Ker}(d)=\text{Ker}(\mathsf{H}_R)$ that is a repetition code with the codeword being the 0-cocycle $\cs^0$ and checks defined on 1-cells (green edges), where $\Hs_R$ is the corresponding parity check matrix.   

Similarly, we can truncate the CW complex at its leftmost part:
\begin{align}
 0  \rightarrow & C_8 \rightarrow  C_7 \rightarrow 0,
\end{align}
which can be rewritten as:
\begin{align}
X^*_R:   \qquad     \bar{C}^*_1 & \rightarrow  \bar{C}^*_{0} \quad.
\end{align}
Here, $X_R$ is the chain complex corresponding to a 1-cycle code $\C_R = \text{Ker}(\partial)=\text{Ker}(\mathsf{H}_R)$ that is a repetition code with the codeword being the 0-cocycle $\cs^0$ and checks defined on 1-cells (green edges).

We call $X^*$ and $X^*_R$ to be the dual complexes of $X$ and $X_R$ respectively, since their corresponding cells are actually the dual cells in the Poincar\'e CW complex $\L$.   

Now we can define the 2D (3-term) skeleton chain complexes for the {\rd}, {\bl} and {\gr} copies by properly truncating the CW complex $\L=\L_x \otimes \L_y$ as follows: 
\begin{align}
X^{\rd} =&  X \otimes X^*  \cr
X^{\bl} =&  X^* \otimes X_R  \cr
X^{\gr} =&  X_R \otimes X. \cr
\end{align}
Note that copy ${\rd}$ corresponds to the (untwisted) skeleton 2D hypergraph-product code, while copy $\bl$ and $\gr$ corresponds to the product of a good classical code and a repetition code.  Here, the truncated 2D chain complexes are all subcomplexes of the CW complex $\L$, i.e.,  $X^{\rd}, X^{\bl}, X^{\gr}  \subset \L$.

We can now consider the pullback of the cochains from the CW complex $\L_x=\L_y$ to the skeleton 1D chain complex $X$, $X^*$ and $X_R$, i.e.,
\begin{align}
 \bs^2 \in C^2(\L_x; \ZZ_2) &\rightarrow   \bar{\bs}^0 \in \bar C^0(X; \ZZ_2),  \cr
 {\bs^*}^6 \in  C^6 (\L_x; \ZZ_2) &\rightarrow \bar{\bs}^{*1} \in  \bar C^{1}(X^*; \ZZ_2),  \cr
  \as^3 \in C^3(\L_x; \ZZ_2) &\rightarrow   \bar{\as}^1 \in \bar C^0(X; \ZZ_2),  \cr
 {\as^*}^5 \in  C^6 (\L_x; \ZZ_2) &\rightarrow \bar{\as}^{*0} \in  \bar C^{0}(X^*; \ZZ_2),  \cr
 \cs^0 \in C^0(\L_x; \ZZ_2) &\rightarrow   \bar{\cs}^0 \in \bar C^0(X_R; \ZZ_2), \cr
  \fs^1 \in C^1(\L_x; \ZZ_2) &\rightarrow   \bar{\fs}^1 \in \bar C^1(X_R; \ZZ_2),
\end{align}
where $\fs^1$ is the spurious 1-cochain mentioned before. Now the cup product structure defined in the CW complex $\L_x = \L_y$ is also pulled back to the skeleton 1D chain complex $X$, $X^*$ and $X_R$.  For example, we have:
\be
\bs^2 \cup {\bs^*}^6 \cup \cs^0 \rightarrow \bar{\bs}^0 \cup \bar{\bs}^{*1} \cup \bar{\cs}^0,  \quad \as^3 \cup {\as^*}^5 
\rightarrow \bar{\as}^1 \
\cup \bar{\as}^{*0}. \ee
The evaluation rule of the cup product in the skeleton chain complexes is completely the same as the evaluation rule in the CW complex $\L_x= \L_y$, since the pullback operation of the non-trivial cup product is just a renaming of the cocycles and cells.  

\subsubsection{Path integral on the skeleton 2D general chain complex}
Equipped with the cup product structure, we can now define the path integral in the skeleton spacetime chain complexes $\tilde X^{\rd}= X^{\rd} \otimes I_t$, $\tilde X^{\bl}= X^{\bl} \otimes I_t$ and $\tilde X^{\gr}= X^{\gr} \otimes I_t$:  
\be\label{eq:path_integral_skeleton}
\mathcal{Z}[\{\tilde X^{\rd}, \tilde X^{\bl}, \tilde X^{\gr} \}] = \sum_{\red{\bar{a}^1}, \blue{\bar{b}^1}, \green{\bar{c}^1}} (-1)^{\int \red{\bar{a}^1} \cup \blue{\bar{b}^1} \cup \green{\bar{c}^1}},   
\ee
where $\red{\bar{a}^1} \in H^1(\tilde X^{\rd}; \ZZ_2)$, $\blue{\bar{b}^1} \in H^1(\tilde X^{\bl}; \ZZ_2)$ and $\green{\bar{c}^1} \in H^1(\tilde X^{\gr}; \ZZ_2)$ are the $\ZZ_2$ 1-cocycles corresponding to the gauge fields defined on the skeleton chain complexes.  The sum `$\int$' is over the triple of skeleton chain complexes $\{\tilde X^{\rd}, \tilde X^{\bl}, \tilde X^{\gr} \}$, or equivalently the spacetime CW complex $\tilde{\L}=\L \otimes I_t$ with these truncated skeleton complexes being its subcomplexes. 

Now we can verify that there is a non-trivial triple intersection structure in the spacetime path integral.  We denote the 1-cocyle bases for $\tilde X^{\rd}, \tilde X^{\bl}, \tilde X^{\gr}$ by $\{\red{\tilde{{\alpha}}^1}\}$, $\{\blue{\tilde{{\beta}}^1}\}$  and $\{\green{\tilde{{\gamma}}^1}\}$.  The basis cocycles can have a tensor decomposition as:
\begin{align}\label{eq:skeleton_cocycle_decomposition}
\red{\tilde{\alpha}^{1}} =&  \bar{\bs}^0 \otimes \bar{\bs}'^{*1} \otimes  \ts^0 \equiv \red{{\bar \alpha}^{1}} \otimes \ts^0   \cr
\blue{\tilde{\beta}^{1}} =& \bar\bs^{*1}  \otimes  \cs^0 \otimes     \ts^0 \equiv  \blue{{
\bar \beta}^{1}} \otimes \tau^0 \cr
\green{\tilde{\gamma}^{1}} =& \cs^0 \otimes  \bar\bs'^0  \otimes {\ts^*}^1 \equiv \green{\bar \gamma^0} \otimes {\ts^*}^1,  
\end{align}
where $\red{{\bar \alpha}^{1}}$, $\blue{{
\bar \beta}^{1}} $ and $\green{\bar \gamma^0}$ represent the space basis cocycles.  We hence have the following nontrivial triple intersection
\begin{align}\label{eq:triple_intersection_skeleton}
   & \int  \red{\tilde{\alpha}^{1}} \cup \blue{\tilde{\beta}^{1}} \cup \green{\tilde{\gamma}^1}  \cr
=&  \int(\bar{\bs}^0 \cup \bar{\bs}^{*1} \cup \bar{\cs}^0    
) \cdot \int (\bar{\bs}'^{*1} \cup \bar{\cs}^0  \cup \bar{\bs}'^0  )   \cr
&\cdot \int_{I_t} ( \ts^0 \cup  \ts^0  \cup {\ts^*}^1 ) \cr
=& 1 \cdot 1 \cdot 1 =1,
\end{align}
which is just the pullback from Eq.~\eqref{eq:triple_cup_sum}. 

\subsubsection{Explicit evaluation of the cup product on the skeleton chain complex and the CW complex}\label{sec:evaluation}
In this subsection, we also want to derive the explicit formula of the evaluation of cup product between cochains. 
Note that since the chain complex $X$ and $X^*$ only differ by an exchange of the bit and check, there exists an isomorphism (duality) between a  1-cell (0-cell) in $X$ and a 0-cell (1-cell) $X^*$, i.e., $\sigma_1 \sim \sigma^*_0$ and $\sigma_0 \sim \sigma^*_1$, which resembles the Poincar\'e duality structure in the Poincar\'e CW complex $\L_x=\L_y$.   Similarly, there is an isomorphism (duality) between the chains  and cochains in $X$ and $X^*$, i.e.,
\be
\bar{\bs}_0 \sim \bar{\bs}^{*1},   \quad \bar{\bs}^0 \sim \bar{\bs}^*_1, \quad \bar{\as}_1 \sim \bar{\as}^{*0}, \quad
\bar{\as}^1 \sim \bar{\as}^*_{0}.
\ee
Due to this isomorphism, we can evaluate the following cup product on a cell $\sigma_1^* \sim \sigma_0$:
\begin{align}
(\bar{\bs}^0 \cup \bar{\bs}^{*1})(\sigma^*_1) =  (\bar{\bs}^0, \bar{\bs}_0)(\sigma_0)   \equiv  \bar{\bs}^0 (\sigma_0) \bar{\bs}_0 (\sigma_0).
\end{align}
Here, $(\bar{\bs}^0, \bar{\bs}_0)$ denotes a chain-cochain pairing, which can be considered as the inner product between the vector $\bar{\bs}^0$ and the dual vector $\bar{\bs}_0$ since the chain group $C_i$ can be viewed as a $\Z_2$ vector space with the cochain  group $C^i$ being its dual vector space. Its evaluation on a cell $(\bar{\bs}^0, \bar{\bs}_0) (\sigma_0)$ is explicitly given above as the product of $\bar{\bs}^0 (\sigma_0)$ and $\bar{\bs}_0 (\sigma_0)$.  Similarly, we can have the following cup product evaluation on a 1-cell $\sigma_1$:
\begin{align}
\bar{\as}^1 \cup \bar{\as}^{*0}(\sigma_1) = (\bar{\as}^1, \bar{\as}_1)(\sigma_1)   \equiv  \bar{\as}^1 (\sigma_1) \bar{\as}_1 (\sigma_1).
\end{align}

We can then evaluate the following non-trivial cup-product sum when enforcing $\bar{\bs}^0 $ and $\bar{\bs}^{*1}$ to be a pair of dual basis cocycles ($d\bar{\bs}^0=d \bar{\bs}^{*1}=0$): 
\begin{align}
& \int \bar{\bs}^0 \cup \bar{\bs}^{*1} \equiv \sum_{\sigma^*_1 \in X^*}   (\bar{\bs}^0 \cup \bar{\bs}^{*1}) (\sigma^*_1) \cr
=& \sum_{\sigma_0 \in X} (\bar{\bs}^0, \bar{\bs}_0)(\sigma_0) \equiv \int_{X}(\bar{\bs}^0, \bar{\bs}_0) = 1, 
\end{align}
where $\int_{X}(\bar{\bs}^0, \bar{\bs}_0)$ can be interpreted as a cycle-cocycle pairing. In this case,  $\bar{\bs}^0$ and $\bar{\bs}_0$ are a conjugate pair of basis cocycle and cycle, which gives the non-trivial pairing   $\int_{X}(\bar{\bs}^0, \bar{\bs}_0)=1$ due to the following isomorphism between the $\ZZ_2$ homology and cohomology groups:
\be
H_i(X; \ZZ_2) \cong H^i(X; \ZZ_2), 
\ee
which can be called a \textit{cycle-cocycle duality}.  The more general pairing relation between arbitrary pairs of basis cocycle $\bar{\bs}^0$ and cycle $\bar{\bs}'_0$ can be written as:
\be
\int_{X}(\bar{\bs}^0, \bar{\bs}'_0)=\delta_{\bs, \bs'},
\ee
with $\bs=\bs'$ enforcing them to be the conjugate pairs.  Note that while the non-trivial cup product in the CW complex $\L$ comes from the Poincar\'e duality,  the cup product in the pair of truncated 1D general chain complexes $\{X, X^*\}$ comes from the cycle-cocycle duality, which is related. 

We now evaluate the following triple cup product of cochains on the cell $\sigma^*_1 \sim \sigma_0$:
\begin{align}
(\bar{\bs}^0 \cup \bar{\bs}^{*1} \cup \bar{\cs}^0)(\sigma^*_1) &=  (\bar{\bs}^0, \bar{\bs}_0, \bar{\cs}^0)(\sigma_0) \cr
&\equiv  \bar{\bs}^0 (\sigma_0) \bar{\bs}_0 (\sigma_0)\bar{\cs}^0(\sigma_0).
\end{align}
Where $(\bar{\bs}^0, \bar{\bs}_0, \bar{\cs}^0)$ represent a triple generalizing the chain-cochain pairing.
Note that, as can be seen from Fig.~\ref{fig:CW_construction}(c), a 0-cell (red dot) of the repetition code chain complex $X_R$ and also the CW complex $\L$ resides on a 2-cell in the CW complex $\L$, which is translated to the 0-cells $\sigma_0$ in the chain complex $X$.  Therefore, we consider the 0-cell in $X_R$ coinciding with the 0-cell in $X$, which are both denoted by $\sigma_0$.   

Similarly, we can also evaluate the following trip cup product of cochains:
\begin{align}
&(\bar{\as}^1 \cup \bar{\as}^{*0}  \cup \bar{\cs}^0)(\sigma_1) = (\bar{\as}^1, \bar{\as}_1, \bar{\cs}^0)(\sigma_1) \cr  \equiv &\,  \bar{\as}^1 (\sigma_1) \bar{\as}_1 (\sigma_1) \bar{\cs}^0(\sigma_0),
\end{align}
where $\sigma_0$ is the 0-cell on the repetition code chain complex $X_R$  that is connected to the 1-cell $\sigma_1$ on $X$. This can be seen from the CW complex $\L$ in Fig.~\ref{fig:CW_construction}(c), where $f(i)$ 0-cells $\sigma_0$ is connected to a 3-cell $\sigma_3$.  The 3-cell is translated to the 1-cell $\sigma_1$ in $X$.  Note that we can choose any $\sigma_0$ connected to $\sigma_1$, since in the code space (ground-state subspace) $\bar{c}^0$ is enforced to be a 0-cocycle corresponding to the repetition codeword that is 1 everywhere.

Next, we can evaluate the following triple cup product sum of cocycles: 
\begin{align}
 &\int \bar{\bs}^0 \cup \bar{\bs}^{*1} \cup \bar{\cs}^0 = \int_{X} (\bar{\bs}^0, \bar{\bs}_0, \bar{\cs}^0) \cr
 \equiv& \sum_{\sigma_0 \in X} \bar{\bs}^0 (\sigma_0) \bar{\bs}_0 (\sigma_0)\bar{\cs}^0(\sigma_0) = \sum_{\sigma_0 \in X} \bar{\bs}^0 (\sigma_0) \bar{\bs}_0 (\sigma_0) \cr
\equiv & \int_{X} (\bar{\bs}^0, \bar{\bs}_0) =1,
\end{align}
where we have used the property of the codeword of the repetition code $\bar{\cs}^0(\sigma_0)=1$ for any $\sigma_0 \in X$.  With this, we can now explicitly  evaluate the triple intersection within the path integral of the skeleton twisted 2D hypergraph product code from Eq.~\eqref{eq:triple_intersection_skeleton}: 
\begin{align}\label{eq:triple_cocycle_evaluation}
   & \int  \red{\tilde{\alpha}^{1}} \cup \blue{\tilde{\beta}^{1}} \cup \green{\tilde{\gamma}^1}  \cr
=& \int_{X} (\bar{\bs}^0, \bar{\bs}_0, \bar{\cs}^0)  \cdot \int_X (\bar{\bs}'_0, \bar{\cs}^0, \bar{\bs}'^0  )   \cdot \int_{I_t} ( \ts^0 \cup  \ts^0  \cup {\ts^*}^1 ) \cr
=& 1 \cdot 1 \cdot 1 =1.
\end{align}
Similarly, we can evalute the triple intersection between the spatial cocycles as follows:
\begin{align}\label{eq:triple_cocycle_evaluation_spatial}
   & \int  \red{\bar{\alpha}^{1}} \cup \blue{\bar{\beta}^{1}} \cup \green{\bar{\gamma}^0}  \cr
=& \int_{X} (\bar{\bs}^0, \bar{\bs}_0, \bar{\cs}^0)  \cdot \int_X (\bar{\bs}'_0, \bar{\cs}^0, \bar{\bs}'^0  )  \cr
=& 1 \cdot 1  =1.
\end{align}

We can see that, quite remarkably, for the specific construction of the path integral in Eq.~\eqref{eq:path_integral_skeleton} along with the triple intersection property \eqref{eq:triple_intersection_skeleton}, we can define all the cup product and topological invariants on the 2D general chain complex from a first principle using the cyclce-cocycle duality and the triple with the repetition codeword, without resorting to the pullback of the cup product from the CW complexes.

Interestingly, we have also effectively obtained the direct evaluation rule of the cup product in the constructed CW complexes $\L$, when mapping the cycle-cocycle duality back to $\L$, rather than relying on subdividing the CW compexes into simplicial complexes as mentioned in Sec.~\ref{sec:cup_product}.

\subsubsection{Twisted skeleton hypergraph-product codes and the code parameter scaling}

We now construct the twisted (Clifford stabilizer) codes $\tilde{\C}_\text{sk}$ defined on the skeleton product chain complexes $\{X^{\rd}, X^{\bl}, X^{\gr}\}$ by truncating the corresponding twisted codes $\tilde{\C}$ defined on the CW complex $\L$.  In particular, we choose a subset of Clifford stabilizer generators in Eq.~\eqref{eq:stabilizer_summary} and pull back to the following set of generators: 
\begin{align}\label{eq:stabilizer_skeleton}
\tilde{A}^{\rd}_{\sigma_{0}} =& A^{\rd}_{\sigma_{0}} \prod_{\sigma_1,\sigma'_1: \int\tilde{\sigma}_{0}\cup \tilde{\sigma}_1\cup \tilde{\sigma}'_1\neq 0}\text{CZ}^{\bl,\gr}_{\sigma_1, \sigma'_1},  \cr
\tilde{A}^{\bl}_{\sigma_{0}} =& A^{\bl}_{\sigma_{0}}\prod_{\sigma_1,\sigma'_1: \int \tilde{\sigma}_1\cup \tilde{\sigma}_{0}\cup  \tilde{\sigma}'_{1} \neq 0}\text{CZ}^{\rd, \gr}_{\sigma_1, \sigma'_1}, \cr
 \tilde{A}^{\gr}_{\sigma_{0}} =& A^{\gr}_{\sigma_{0}}\prod_{\sigma_1,\sigma'_1: \int \tilde{\sigma}_1\cup \tilde{\sigma}'_1\cup \tilde{\sigma}_{0} \neq 0}\text{CZ}^{\rd,\bl}_{\sigma_1, \sigma'_1},  \cr
B^{\rd}_{\sigma_{2}}=& \prod_{\sigma_{1}\subset \partial \sigma_{2}}Z_{\sigma_1}^{\rd}, \quad  B^{\bl}_{\sigma_{2}}=\prod_{\sigma_{1}\subset \partial \sigma_{2}}Z_{\sigma_1}^{\bl}   \cr 
B^{\gr}_{\sigma_{2}} =& \prod_{\sigma_{1}\subset \partial \sigma_{2}}Z_{\sigma_{1}}^{\gr}.
\end{align}
We note that the dressed CZ operators above is determined by the non-trivial triple cup product sum of the indicator cochains, such as  $\int\tilde{\sigma}_{0}\cup \tilde{\sigma}_1\cup \tilde{\sigma}'_1=1$, which is the pullback of the corresponding triple cup product sum in the CW complex $\L$.   The untwisted Pauli stabilizer code on the skeleton chain complex $\C_\text{sk}$ is defined simply by removing the dressed CZ operators in the above stabilizer generators.  

Since the above Clifford stabilizer code $\tilde \C_\text{sk}$ on the skeleton chain complex  is a truncated version of the twisted thickened code $\tilde{\C}$ on the CW complex $\L$, the stabilizer generators on the skeleton are just the pullback of a subset of  the stabilizer generators of the thickened code.  Therefore, all the stabilizers in the skeleton code will still commute when projected to the 0-flux subspace, obeying Lemma~\ref{lemma:stabilizer_commutation}.  

We also have the following lemma for the encoding rate and code distance:
\begin{lemma}\label{lemma:scaling_skeleton}
   The untwisted and twisted 2D hypergraph-product codes $\C_{sk}$ and $\tilde{\cal C}_{sk}$  defined on the skeleton chain complexes both with parameters $[[\Theta(n), \Theta(k), \Theta(d)]]$ has constant rate, i.e., $k=\Theta(n)$, and subsystem-code distance scaling  $d=\Omega(\sqrt{n})$,  assuming Statement \ref{statement:TQFT} holds for $\tilde{\cal C}_{sk}$. 
\end{lemma}
\begin{proof}
Note that all the (co)cycles in the untwisted and twisted skeleton codes $\C_\text{sk}$ and $\tilde\C_\text{sk}$ are the pullback of the untwisted and twisted thickened codes $\C$ and $\tilde{\C}$ defined on the CW complex $\L$.  

Since a constant fraction of the non-trivial (co)cycles in the CW complex $\L$ is preserved on the truncated 2D skeleton chain complexes, we still obtain a constant rate $k=\Theta(n)$ for the skeleton codes due to Lemma \ref{lemma:rate_untwisted} and Lemma \ref{lemma:rate_twisted}.  

Since the (co)cycle and logical operators in the skeleton codes obtained from the pullback has the same weight of support as the (co)cycles and logical operators in the thickened codes, we preserve the subsystem-code distance scaling $d=\Omega(\sqrt{n})$ due to Lemma \ref{lemma:distance_untwisted} and Lemma \ref{lemma:distance_twisted}, assuming Statement \ref{statement:TQFT} holds for the twisted skeleton code $\tilde{\cal C}_{sk}$.    

We also note that some spurious (co)cycles in the thickened codes such as the 1-cocycle $\fs^1$ are also pulled back to the skeleton codes, therefore we still need to consider the subsystem-code distance which ignores the short spurious (co)cycles. 
\end{proof}

\subsubsection{The 0-form subcomplex symmetries}\label{sec:0-form}

As has been shown in Sec.~\ref{sec:charge_parity_operator}, the charge parity operator or equivalently the transversal CZ operator is a higher-form [($k$-1)-form] symmetry in the $\ds$-dimensional CW complex $\L$ (or a triangulated $\ds$-manifold) acting on a codimension-($k$-1) or equivalently a $(\ds -k+1)$-dimensional subcomplex (cycle). This provides a mechanism for addressable transversal CZ gate in the untwisted code $\C$.  

In this subsection, we show that when pulled back to the skeleton 2D chain complexes, the corresponding operator becomes a novel \textit{0-form subcomplex symmetry}, which in contrast is supported on a codimension-0, or equivalently top dimensional ($\ds$-dimensional) subcomplex. The subcomplex can also be understood as a $\ds$-cycle $\bar \eta^*_d $ or equivalently the Poincar\'e dual $0$-cocycle $\bar\eta^0=\text{PD}(\bar\eta^*_{\ds})$.   This hence provides a more general mechanism of addressable transversal CZ gates in the untwisted code: the \textit{subcomplex symmetries}, with higher-form symmetries being special cases.

We start with the charge parity operator for the $\gr$ type or equivalently the transversal CZ operator between the $\rd$- and $\bl$-copies defined on the CW complex $\L$ according to Eqs.~\eqref{eq:CZ_identity1} and \eqref{eq:charge_operator_green}:
\begin{align}
& \mathsf{C}^{\gr}_{\green{\gamma^2}} \equiv \mathsf{C}^{\gr}_{\green{\gamma^*_{14}}} \equiv \widetilde{\text{CZ}}^{\rd, \bl}_{\green{\gamma^2}} \equiv \widetilde{\text{CZ}}^{\rd, \bl}_{\green{\gamma^*_{14}}} \cr
=&  (-1)^{\int_{\green{\gamma^*_{14}}} \red{\hat{a}^8} \cup \blue{\hat{b}^6}  } \equiv (-1)^{\int_{\L} \red{\hat{a}^8} \cup \blue{\hat{b}^6} \cup \green{\gamma^2} } .
\end{align}
The operator is supported on the basis cocycles $\green{\gamma^2}$$=$$\cs^0 $$\otimes$$\bs'^2 $, or equivalently its Poincar\'e dual basis cycles $\green{\gamma^*_{14}} \equiv PD(\green{\gamma^2})$.

Now we pull back the basis (co)cycles to the skeleton 2D chain complex:
\begin{align}
\green{\gamma^2}=\cs^0 \otimes \bs'^2 &\rightarrow  \green{\bar\gamma^0} = \bar\cs^0 \otimes \bar\bs'^0 \in H^0(X^{\gr};\ZZ_2),   \cr 
\green{\gamma^*_{14}} = \cs^*_8  \otimes \bs'^*_6 &\rightarrow  \green{\bar\gamma^*_2} = \bar \cs^*_1  \otimes \bar\bs'^*_1 \in H_2(X^{*\gr};\ZZ_2),  
\end{align}
where the cohomology is defined on the skeleton chain complex $X^{\gr} =  X_R \otimes X$, and the homology on its dual complex $X^{*\gr} =  X^*_R \otimes X^*$.  
The above charge parity and transversal CZ operators are also pulled back to the following form:
\begin{align}
& \mathsf{C}^{\gr}_{\green{\bar \gamma^0}} \equiv \mathsf{C}^{\gr}_{\green{\bar \gamma^*_{2}}} \equiv \widetilde{\text{CZ}}^{\rd, \bl}_{\green{\bar \gamma^0}} \equiv \widetilde{\text{CZ}}^{\rd, \bl}_{\green{\bar\gamma^*_{2}}} \cr
=&  (-1)^{\int_{\green{\bar\gamma^*_{2}}} \red{\hat{\bar a}^1} \cup \blue{\hat{\bar b}^1}  } \equiv (-1)^{\int \red{\hat{\bar a}^1} \cup \blue{\hat{\bar b}^1} \cup \green{\bar \gamma^0} },
\end{align}
where the gauge fields $\red{\hat{\bar a}^1}$ and  $ \blue{\hat{\bar b}^1}$ are operator-valued 1-cochains.   The above triple cup product can evaluated in a similar way to Eq.~\eqref{eq:triple_cocycle_evaluation} by carrying out cup product on the tensor components. 
Note that similar to the case of higher-form symmetry, the 0-form subcomplex symmetry also has the generaic form of a cycle-cocycle pairing.   In this case, it is the pairing between the 2-cycle $\green{\bar\gamma^*_{2}}$ and the operator-valued 2-cocycle $\red{\hat{\bar a}^1} \cup \blue{\hat{\bar b}^1}$. 

\begin{figure}
    \centering
\includegraphics[width=0.7\linewidth]{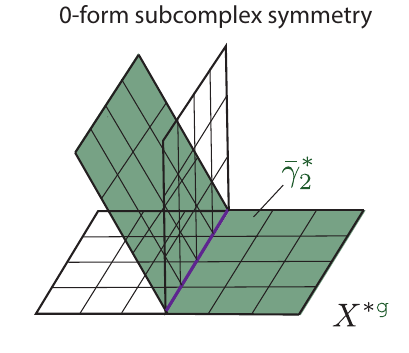}
    \caption{Illustration of a 0-form subcomplex symmetry on a toy model of square complex $X^{*\gr}$. The symmetry operator is  supported on the codimension-0 subcomplex (2-cycle) $\green{\bar{\gamma}^*_2} \subset X^{*\gr}$.}
  \label{fig:subcomplex_illustration}
\end{figure}

As we can see,  the above operators are supported on the top-dimensional cycle (2-cycle)  $\green{\bar\gamma^*_{2}}$ or equivalently a codimension-0 cocycle (0-cocycle) $\green{\bar{\gamma}^0}$.  Here, the 2-cycle is a codimension-0 (top dimensional) subcomplex of the general 2D  chain complex, i.e.,  $\green{\bar\gamma^*_{2}} \subset X^{*\gr} $, as illustrated in Fig.~\ref{fig:subcomplex_illustration}.  We hence call these operators the \textit{0-form subcomplex} symmetries of the untwisted skeleton hypergraph-product codes. Later in Sec.~\ref{sec:logical_operation}, we use such subcomplex symmetries to implement addressable and parallel gauging measurement of logical Clifford operators.

Remarkably, such type of symmetries have never shown up before in TQFT or gauge theory defined on a manifold since there is only one unique top-dimensional cycle ($\ds$-cycle) on a $\ds$-manifold $\M^d$, which is just the manifold $\M^d$ itself.  
This suggests that there exist much richer structures of generalized symmetries in general chain complexes beyond manifolds.

\section{Logical non-Clifford operation via gauging measurement and spacetime path integral}\label{sec:logical_operation}

\subsection{Addressable and parallel gauging measurement of logical CZ}
\label{sec:parallel_gauging}

In this section, we study the gauging measurement protocol which implements the non-Clifford logical operation and the corresponding spacetime path-integral description.

\subsubsection{Gauging measurement of higher-form symmetries}
The central idea of the protocol can be understood as a gauging measurement of the transversal CZ operators in two independent copies ($\rd$ and $\bl$) of untwisted qLDPC codes: $\C= \C^{\rd} \otimes \C^{\bl}$.   These transversal CZ operators are supported on the  basis cocycles $\green{\gamma^2}$, or equivalently its Poincar\'e dual basis cycles $\green{\gamma^*_{14}} \equiv PD(\green{\gamma^2})$ in the space complex $\L$:
\be\label{eq:measure_higher-form}
\widetilde{\text{CZ}}^{\rd, \bl}_{\green{\gamma^2}} \equiv \widetilde{\text{CZ}}^{\rd, \bl}_{\green{\gamma^*_{14}}} =  (-1)^{\int_{\green{\gamma^*_{14}}} \red{\hat{a}^8} \cup \blue{\hat{b}^6}  } \equiv (-1)^{\int_{\L} \red{\hat{a}^8} \cup \blue{\hat{b}^6} \cup \green{\gamma^2} }.  
\ee
We note that $\widetilde{\text{CZ}}^{\rd, \bl}_{\green{\gamma^2}}  \equiv \widetilde{\text{CZ}}^{\rd, \bl}_{\green{\gamma^*_{14}}}$ is hence a higher-form (3-form) symmetry operator supported on a codimension-3 cocycle $\green{\gamma^2}$ or equivalently the dual 13D cycle $\green{\gamma^*_{14}}$. 
By expressing the operator-valued cochain with the cocycle basis $\{\red{\alpha^8}\}$ and $\{\blue{\beta^6}\}$:  $\red{\hat{a}^8}$$=$$ \sum_{\red{\alpha}}\red{\hat{n}_\alpha \alpha^8}$ and $\blue{\hat{b}^6}$$=$$\sum_{\blue{\beta}}\blue{\hat{m}_\beta \beta^6}$, we can write the transversal CZ operator as:
\begin{align}\label{eq:logical_CZ}
\widetilde{\text{CZ}}^{\rd, \bl}_{\green{\gamma^*_{14}}} =&  \prod_{\red{\alpha}, \blue{\beta}} (-1)^{\int_{\green{\gamma^*_{14}}} \red{\hat{n}_\alpha \alpha^8} \cup \blue{\hat{m}_\beta \beta^6} } = \prod_{\red{\alpha}, \blue{\beta}} [(-1)^{\red{\hat{n}_{\alpha}} \cdot \blue{\hat{m}_\beta}}]^{\int_{\green{\gamma^*_{14}}} \red{ \alpha^8} \cup \blue{ \beta^6} }  \cr
=& \prod_{\red{\alpha}, \blue{\beta}} \lo{\text{CZ}}(\red{\alpha}, \blue{\beta})^{\int_{\green{\gamma^*_{14}}} \red{ \alpha^8} \cup \blue{ \beta^6} }  
\equiv \prod_{\red{\alpha}, \blue{\beta}} \lo{\text{CZ}}(\red{\alpha}, \blue{\beta})^{\int_{\L} \red{\alpha^8} \cup \blue{\beta^6} \cup \green{\gamma^2} }.
\end{align}
We can see that the transversal operator corresponds to the product of logical CZ operators acting on logical qubits with their logical-$X$ supported on ${\alpha^8}$ and ${\beta^6}$ if and only if the exponent gives rise to a non-trivial triple intersection   
\be
\int_{\L} \red{\alpha^8} \cup \blue{\beta^6} \cup \green{\gamma^2} = |\red{\alpha^*_8} \cap \blue{\beta^*_{10}} \cap \green{\gamma^*_{14}}|= 1,
\ee
where $\red{\alpha^*_8} $,  $ \blue{\beta^*_{10}}$ and $ \green{\gamma^*_{14}}$ are the Poincar\'e dual cycles defined on the dual space complex $\L^*$. 

Note that as pointed out in Sec.~\ref{sec:charge_parity_operator}, the above transversal CZ operator in the untwisted code is equivalent to the charge parity operator in the twisted (gauged) code which can be factorized into the product of local dressed $X$-stabilizers, i.e., 
\be\label{eq:CZ_factorizing}
\widetilde{\text{CZ}}^{\rd, \bl}_{\green{\gamma^*_{14}}} = \mathsf{C}^{\gr}_{\green{\gamma^2}} \equiv \mathsf{C}^{\gr}_{\green{\gamma^*_{14}}} = \prod_{ \sigma_2 \in \green{\gamma^2}} \tilde{A}^{\gr}_{\sigma_2}. 
\ee
The main gauging idea can hence be stated as below:

We start with two copies (\red{red} and \blue{blue}) of independent untwisted codes $\C= \C^{\rd} \otimes \C^{\bl}$ corresponding to a $\red{\ZZ_2} \times \blue{\ZZ_2}$ gauge theory.   We now gauge the collection of higher-form symmetries corresponding to the transversal CZ operators $\widetilde{\text{CZ}}^{\rd, \bl}_{\green{\gamma^*_{14}}}$ for all basis (co)cycles $ \green{\gamma^2} \sim \green{\gamma^*_{14}}$.    The essence of gauging is to turn a global symmetry (a non-local operator) into many local gauge  symmetries ($k$-local operators).
 In our case, we are gauging the generalized global symmetries $\widetilde{\text{CZ}}^{\rd, \bl}_{\green{\gamma^*_{14}}}$,  which are promoted to local gauge symmetries to all the dressed $X$-stabilizers $\tilde{A}^{\gr}_{\sigma_2}$ in the product in Eq.~\eqref{eq:CZ_factorizing}.   During the gauging process, we can fault-tolerantly measure  $\tilde{A}^{\gr}_{\sigma_2}$'s and also correct them properly.  By proper combination of the classical measurement data,  i.e., multiplying them along all the basis (co)cycles $\green{\gamma^2} \sim \green{\gamma^*_{14}}$,  one can in turn determine all the eigenvalues of the higher-form symmetry operators $\widetilde{\text{CZ}}^{\rd, \bl}_{\green{\gamma^*_{14}}}$. This gauging measurement is hence addressable and highly parallelizable.   

One can describe gauging the CZ symmetry in untwisted code using path integral as follows. We begin with path integral over $\mathbb{Z}_2$ gauge fields $a,b$ with weight 1 for the two copies of the untwisted code, where the CZ symmetry is $(-1)^{\int a\cup b}$. To gauge the symmetry, we introduce another $\mathbb{Z}_2$ gauge field $c$, and couple it to the charge as $(-1)^{\int (a\cup b)\cup c}=(-1)^{\int a\cup b\cup c}$. This gives the same path integral as the twisted code. In other words, gauging the CZ symmetry in the untwisted code produces the twisted code.

We note that in the twisted code that has the path integral with weight $(-1)^{\int a\cup b\cup c}$ inserting a Wilson operator $(-1)^{\rho\int_\gamma c}=(-1)^{\rho\int \gamma^* \cup c}$ for $\rho=0,1$ changes the weight to be $(-1)^{\int (a\cup b+\rho\gamma^*)\cup c}$. Thus $c$ is a Lagrangian multiplier that enforces $[a\cup b]=\rho [\gamma^*]$ in cohomology. In other words, the value of the operator $(-1)^{\int a\cup b}$ on closed support that intersects with $\gamma$ at a point is fixed to be $(-1)^{\rho}$.

 The gauging process corresponds to a non-unitary process at $t_i$ that maps the untwisted code $\C= \C^{\rd} \otimes \C^{\bl}$ to the twisted (gauged) code $\tilde{\C}$ which corresponds to a twisted $\red{\ZZ_2} \times \blue{\ZZ_2} \times \green{\ZZ_2}$ gauge theory, which is then followed by an \textit{ungauging} process at $t_f$ that maps $\tilde \C$ back to $\C$.   The entire gauging-ungauging process is illustrated in Fig.~\ref{fig:protocol}.   One can interpret the gauging and ungauging processes as the spacetime domain walls $W$ and $W'$ separating different topological phases.  As explained in Sec.~\ref{sec:charge_parity_operator}, the measurement outcome of  $\widetilde{\text{CZ}}^{\rd, \bl}_{\green{\gamma^*_{14}}} = \mathsf{C}^{\gr}_{\green{\gamma^*_{14}}}$ in the twisted code $\tilde \C$  from  $t_i$ to $t_f$ tells the total parity of charge that intersects with the cycle $\green{\gamma^*_{14}}$.    Similarly, in the spacetime picture $\widetilde{\text{CZ}}^{\rd, \bl}_{\green{\gamma^*_{14}}} = \mathsf{C}^{\gr}_{\green{\gamma^*_{14}}}$ detects the total parity of charge world-sheets along the relative cycle $\green{\tilde{\gamma}_3}$ terminated at the upper and lower spacetime domain-walls $W$ and $W'$  that intersect with the spacetime cycle $\green{\tilde{\gamma}^*_{14}} \equiv \green{{\gamma}^*_{14}} \otimes \tau_0$.  Note that there is representative of the spacetime cycle $\green{\tilde{\gamma}^*_{14}}$ which has exactly the same support as the space cycle $\green{{\gamma}^*_{14}}$ since $\tau_0$ is just a point (i.e., a vertex in the CW complex $\L$).   The spacetime picture will be understood more deeply using the spacetime path integral formalism in Sec.~\ref{sec:spacetime_path_integral}.

\begin{figure}
    \centering
    \includegraphics[width=1\linewidth]{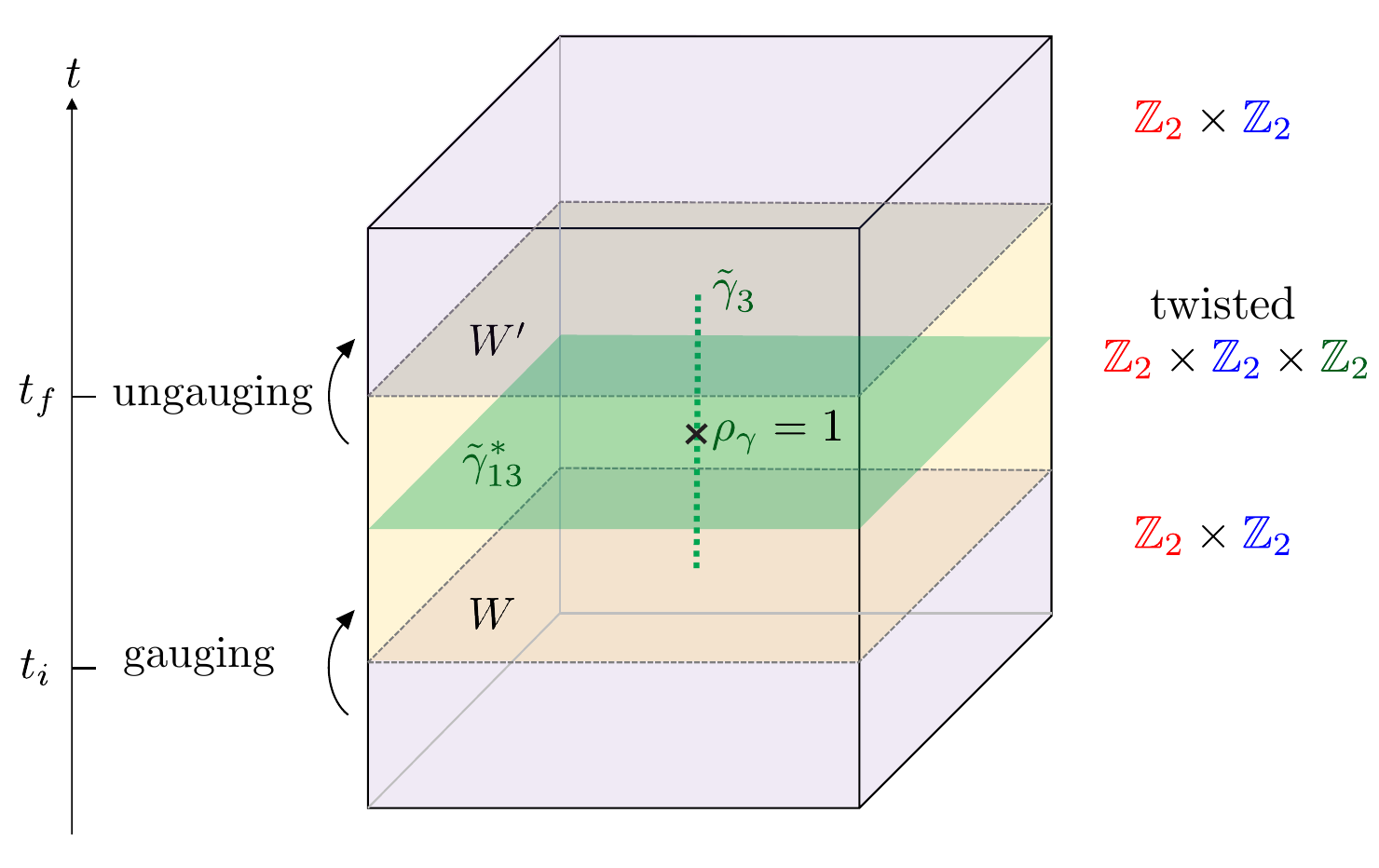}
    \caption{Schematic illustration of the gauging measurement protocol in the spacetime picture, with the arrow of time going upward. The spacetime domain walls $W$ and $W'$ corresponds to the gauging and ungauging processes, which separates the untwisted code $\C=\C^{\rd} \otimes \C^{\bl}$ on the bottom and top corresponding to $\red{\ZZ_2} \otimes  \blue{\ZZ_2}$ gauge theory and the twisted code $\tilde{\C}$ in the middle corresponding to the twisted $\red{\ZZ_2} \otimes  \blue{\ZZ_2} \otimes \green{\ZZ_2}$ gauge theory.  The protocol effectively measures the charge parity operator or equivalently the addressable transversal CZ gate $\mathsf{C}^{\gr}_{\green{\gamma^*_{14}}} =\widetilde{\text{CZ}}^{\rd, \bl}_{\green{\gamma^*_{14}}}$ supported on each spatial 14-cycle $\green{\gamma^*_{14}}$ or equivalently the spacetime 14-cycle $\green{\tilde\gamma^*_{14}} = \green{\gamma^*_{14}} 
    \otimes \tau_0$ shown in the figure.  This operator detects the total charge parity $\green{\rho_\gamma}$ supported on the cycle $\green{\gamma^*_{14}}$ (in the illustration $\green{\rho_\gamma}=1$) at any time between $t_i$ and $t_f$, or equivalently the total parity of the charge worldsheet (spacetime relative cycle) $\green{\tilde \gamma_3}$ that intersects with the spacetime cycle  $\green{\tilde\gamma^*_{14}}$. The protocol stays in the twisted code $\C$ for $O(d)$ rounds of error correction in order to collect reliable syndromes even in the presence of measurement errors.  }
    \label{fig:protocol}
\end{figure}

\subsubsection{The gauging measurement  protocol}\label{sec:protocol}

Now we list the concrete gauging measurement protocol as follows:
\begin{enumerate}
    \item   
Start with three copies of qubit lattices, denoted by $\rd$, $\bl$, and $\gr$. Prepare the {\rd} and {\bl} copies as two decoupled $\Z_2$ qLDPC codes: $\C= \C^{\rd} \otimes \C^{\bl}$.  Initialize the logical qubits at  $\overline{\ket{+}}$ or   $\overline{\ket{0}}$ states depending on the required  connectivity of the logical CZ measurements, denoted by $\lo{\ket{\psi}}_i$.  Meanwhile, initialize the qubits in the $\gr$ copy at $\ket{0}^{\otimes n}$ state.  The stabilizer group of the initial state is 
\be\label{eq:stabilizer_group_1}
\cS_1= \langle   A^{\rd}_{\sigma_{7}},  B^{\rd}_{\sigma_{9}}, A^{\bl}_{\sigma_{5}},  B^{\bl}_{\sigma_{7}}, Z^{\gr}_{\sigma_3},  B^{\gr}_{\sigma_{4}}, \lo{Z}^{\gr}_{\green{\xi_3}}     \rangle
\ee
which is composed of the bare $X$-stabilizers and $Z$-stabilizers in the {\rd} and {\bl} copies. For the green copy, the $\ket{0}^{\otimes n}$ state is stabilized by all the single-qubit Pauli-$Z$ on each qubit $Z^{\gr}_{\sigma_3}$ as well as the $X$-stabilizers $B^{\gr}_{\sigma_{4}}$ from a virtual untwisted qLDPC code with qubits defined on the 4-cells $\sigma_4$ of the space complex $\L$.  Here,  the logical-$Z$ operators $\lo{Z}^{\gr}_{\green{\xi_3}}$  supported on the basis 3-cycle $\green{\xi_3}$ from the $\gr$ copy due to the presence of spurious 1-cycle $\fs_1$ in the CW complex $\L_x$ and $\L_y$ have eigenvalues $+1$. Recall that the 3-cycle $\green{\xi_3}$ has the decomposition form $\fs_1 \otimes \as'_2$ or $ \as_2  \otimes \fs'_1$.   Due to the $\ket{0}^{\otimes n}$ state in the $\gr$ copy, we can also rewrite the above stabilizer group in terms of dressed $X$-stabilizers from the twisted code $\tilde{\C}$ which we will prepare later:
\be
\cS_1= \langle   \tilde{A}^{\rd}_{\sigma_{7}},  B^{\rd}_{\sigma_{9}}, \tilde{A}^{\bl}_{\sigma_{5}},  B^{\bl}_{\sigma_{7}}, Z^{\gr}_{\sigma_3},  B^{\gr}_{\sigma_{4}},  \lo{Z}^{\gr}_{\green{\xi_3}}   \rangle.
\ee

\item  Gauging (at time $t_i$): make a transition to the twisted $\Z_2^3$ non-Abelian qLDPC code by measuring all the green dressed $X$-stabilizers ${\tilde{A}}^{\gr}_{\sigma_2}$, resulting in a collection of  eigenvalues $(-1)^{\mu_{\sigma_2}} =\pm 1$. The product of the dressed $X$-stabilizers $\prod_{\sigma_q} {\tilde{A}}^{\gr}_{\sigma_2}$ along each  basis (co)cycle $\green{\gamma^2} \sim \green{\gamma^*_{14}} $ gives rise to the transversal operator $\widetilde{\text{CZ}}^{\rd, \bl}_{\green{\gamma^*_{14}}}$ since all the Pauli $X$'s get canceled.   We hence get a collection of  projective measurement results 
\be
\{\widetilde{\text{CZ}}^{\rd, \bl}_{\green{\gamma^*_{14}}}= (-1)^{\green{\rho_{\gamma}}}  \equiv \prod_{\sigma_2 \in \green{\gamma^2}}(-1)^{\mu_{\sigma_2}} = \pm 1 \}
\ee
for the entire cycle basis $\{\green{\gamma^*_{14}}\}$, where $\green{\rho_{\gamma}}=0,1$ storing the total charge parity intersecting the basis cycle $\green{\gamma^*_{14}}$, according to the relation $\widetilde{\text{CZ}}^{\rd, \bl}_{\green{\gamma^*_{14}}} = \mathsf{C}^{\gr}_{\green{\gamma^*_{14}}}$.
For each $\green{\gamma^2} \sim \green{\gamma^*_{14}}$,  this gives rise to the logical measurements of the product of logical operators $\prod_{\red{\alpha}, \blue{\beta}} \lo{\text{CZ}}(\red{\alpha}, \blue{\beta})= (-1)^{\green{\rho_{\gamma}}} = \pm 1$  satisfying the condition that the triple intersection is non-trivial, i.e., $\int_{\L} \red{\alpha^8} \cup \blue{\beta^6} \cup \green{\gamma^2}=1$.  
The stabilizer group now becomes
\begin{align}
\qquad \cS_2= \langle   \tilde{A}^{\rd}_{\sigma_{7}},  B^{\rd}_{\sigma_{9}}, \tilde{A}^{\bl}_{\sigma_{5}},  B^{\bl}_{\sigma_{7}}, (-1)^{\mu_{\sigma_2}}\tilde{A}^{\gr}_{\sigma_{2}},  B^{\gr}_{\sigma_{4}}, \lo{Z}^{\gr}_{\green{\xi_3}}     \rangle,  \cr
\end{align}
which involves the random dressed $X$-stabilizer eigenvalues in the $\gr$ copy.

\item
Ungauging (at time $t_f$): to get back to the two original copies of $\ZZ_2$ qLDPC codes $\C= \C^{\rd} \otimes \C^{\bl}$,  we measure out all the qubits of the {\gr} copy in the $Z$ basis.  This results in the following stabilizer group
\be
\cS_3= \langle   \tilde{A}^{\rd}_{\sigma_{7}},  B^{\rd}_{\sigma_{9}}, \tilde{A}^{\bl}_{\sigma_{5}},  B^{\bl}_{\sigma_{7}}, (-1)^{\mu_{\sigma_3}} Z^{\gr}_{\sigma_3},  B^{\gr}_{\sigma_{4}},  (-1)^{\green{\rho_{\gamma}}}, \lo{Z}^{\gr}_{\green{\xi_3}}, \mathsf{C}^{\gr}_{\green{\gamma^2}}     \rangle.
\ee
Here, we have introduced newly measured operators to the stabilizer group while dropping the stabilizers that do not commute with them, i.e., $ \tilde{A}^{\gr}_{\sigma_{2}}$ in this case which does not commute with the overlapping single-qubit measurement of $Z^{\gr}_{\sigma_3}$, while the product of $\tilde{A}^{\gr}_{\sigma_{2}}$ equaling $\mathsf{C}^{\gr}_{\green{\gamma^2}}$ still commutes with all $Z^{\gr}_{\sigma_3}$'s.  The redundant operators $B^{\gr}_{\sigma_4}$  in $\cS_3$ impose constraint on the single-qubit measurement outcome $\mu_{\sigma^3}$:  those qubits with measured value $\mu_{\sigma_3}=1$ ($Z^{\gr}_{\sigma_3}=-1$) must form a cocycle  $\green{\eta^3}$, since $B^{\gr}_{\sigma_4}=+1$ for any $\sigma_4$ which is equivalent to the cocycle condition $(d \green{\eta^3})(\sigma_4) =0$ for any $\sigma_4$;  now due to the fact that  $\lo{Z}^{\gr}_{\green{\xi_3}}=1$ for any (spurious) basis 3-cycle $\xi_3$, there does not exist any non-contractible cycle with $\mu_{\sigma_3}=1$ values on its support. Therefore, we have $\green{\eta^3}$ being a contractible cocycle, i.e., a coboundary of certain open 2-cochain (membrane) $\mathcal{V}^2$:  $\green{\eta^3}= d \mathcal{V}^2$.  We can hence flip all the $\mu_{\sigma_3}  (Z^{\gr}_{\sigma_3})$  to 0 (+1) by applying the correction consisting of a  product of the dressed $X$-stabilizers supported on arbitrary $\mathcal{V}^2$ satisfying $\green{\eta^3}=d \mathcal{V}^2$:
\be
\mathcal{R} = \prod_{\sigma_2 \in  \mathcal{V}^2} \mathcal{A}^{\gr}_{\sigma_2}. 
\ee
This correction leads to the following stabilizer group
\be
\qquad \ \cS_4= \langle   A^{\rd}_{\sigma_{7}},  B^{\rd}_{\sigma_{9}}, {A}^{\bl}_{\sigma_{4}},  B^{\bl}_{\sigma_{7}}, Z^{\gr}_{\sigma_3},  B^{\gr}_{\sigma_{4}}, (-1)^{\green{\rho_{\gamma}}} , \lo{Z}^{\gr}_{\green{\xi_3}}, \mathsf{C}^{\gr}_{\green{\gamma^2}}     \rangle.
\ee
Note that we have replaced the dressed $X$-stabilizers $\tilde{A}^{\rd}_{\sigma_{7}}$ and $  \tilde{A}^{\bl}_{\sigma_{5}}$ with the bare $X$-stabilizers $ A^{\rd}_{\sigma_{7}}$ and $ {A}^{\bl}_{\sigma_{4}}$,  since $\text{CZ}^{\rd, \gr}$ and  $\text{CZ}^{\bl, \gr}$ get trivialized due to the fact that all the \green{green} qubits are in the $\ket{0}^{\otimes n}$ state.  Therefore, the bare and dressed $X$-stabilizers all have the same $+1$ eigenvalues, and the $\gr$ copy gets completely  decoupled with the $\rd$ and $\bl$ copies coming back to the original untwisted code space $\C= \C^{\rd} \otimes \C^{\bl}$.  The only difference with the initial stabilizer group $\cS_1$ in Eq.~\eqref{eq:stabilizer_group_1} is the extra constraint of $\mathsf{C}^{\gr}_{\green{\gamma^2}}=\widetilde{\text{CZ}}^{\rd, \bl}_{\green{\gamma^*_{14}}} = (-1)^{\green{\rho_\gamma}}$, which corresponds to a projection onto the $(-1)^{\green{\rho_\gamma}} = \pm1$ eigenspace of the operator $\widetilde{\text{CZ}}^{\rd, \bl}_{\green{\gamma^*_{14}}}$ or  $\mathsf{C}^{\gr}_{\green{\gamma^2}}$.

\item
 This outputs the following logical state: 
\begin{align}\label{eq:output_logical_state}
\lo{\ket{\psi}}_f =& \prod_{\green{\gamma}} \hat{P}_{1-2\green{\rho_\gamma}}\big[\widetilde{\text{CZ}}^{\rd, \bl}_{\green{\gamma^*_{14}}}\big] \lo{\ket{\psi}}_i   \cr
=& \prod_{\green{\gamma}} \hat{P}_{1-2\green{\rho_\gamma}}\big[\prod_{\red{\alpha}, \blue{\beta}} \lo{\text{CZ}}(\red{\alpha}, \blue{\beta})^{\int_{\L} \red{\alpha^8} \cup \blue{\beta^6} \cup \green{\gamma^2} } \big]\lo{\ket{\psi}}_i, \cr 
\end{align}
where $1-2\green{\rho_\gamma} \equiv (-1)^{\green{\rho_\gamma}}=\pm 1$ for $\green{\rho_\gamma} \in \{0,1\}$, and  $\hat{P}_{\pm 1}\big[\widetilde{\text{CZ}}^{\rd, \bl}_{\green{\gamma^*_{14}}}]$ represents the projector to the $\pm 1$-eigenspace of $\widetilde{\text{CZ}}^{\rd, \bl}_{\green{\gamma^*_{14}}}$. 

\end{enumerate}

\subsubsection{Instantiation: preparing disjoint magic states via the magic state fountain}

\begin{figure*}[t]
    \centering
\includegraphics[width=0.7\linewidth]{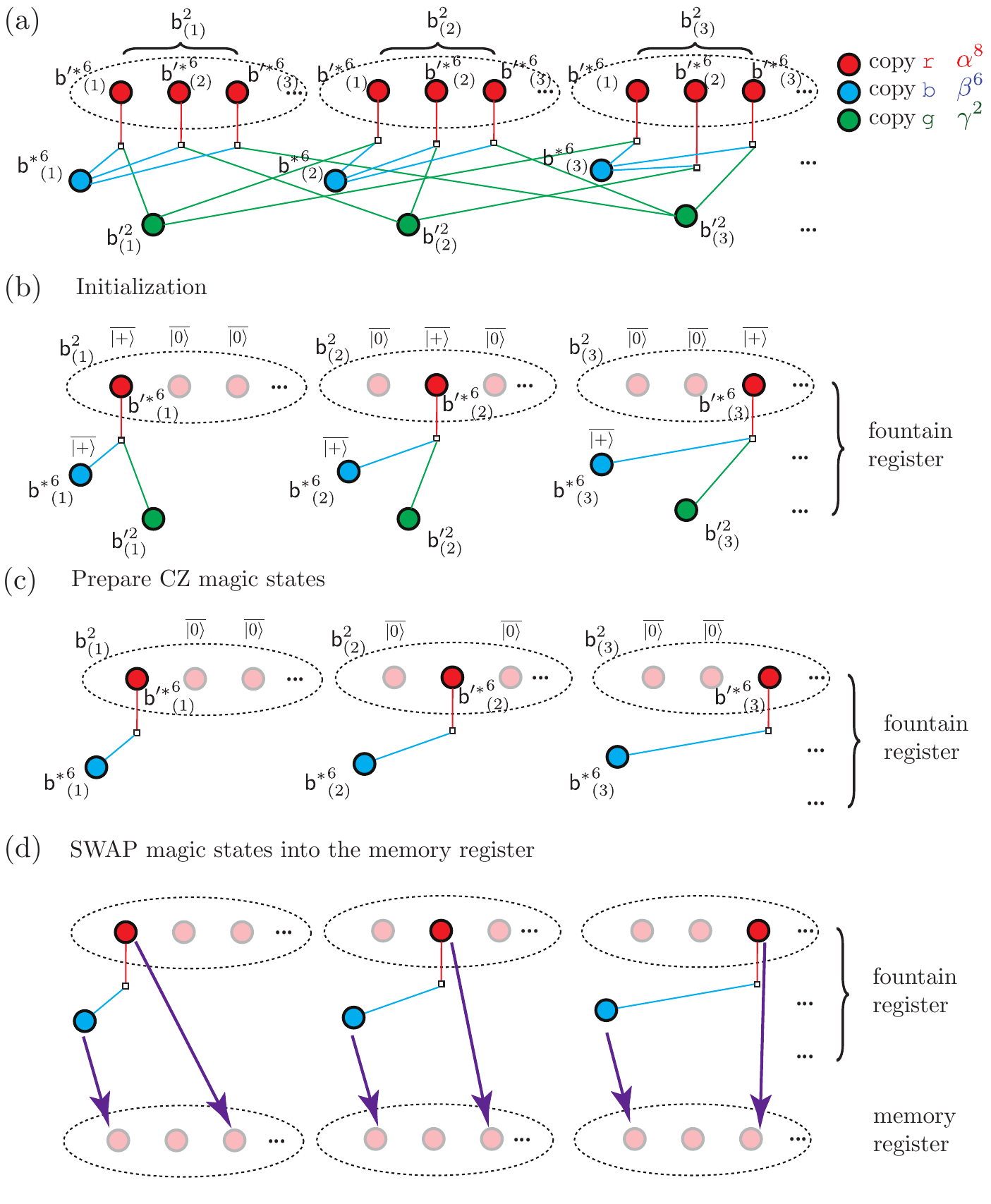}
 \caption{Illustration of the magic state fountain protocol via gauging measurement.
(a) A hypergraph representing the triple intersection structure between the cocycles $\int_{\L} \red{\alpha^8} \cup \blue{\beta^6} \cup \green{\gamma^2}$ in the three copies ({\rd}, {\bl}, {\gr}). Each vertex represents a cocycle, while each hyper-edge (with three legs) represents a triple intersection. Each $\rd$ and $\bl$ vertices also corresponds to a logical qubit, while the $\gr$ vertex represents the logical CZ operator that acts on the intersected $\rd$ and $\bl$ logical qubits. In copy~$\rd$,  each logical qubit (red vertex)  carry two cocycle labels, \(\bs^{2}_{(i)}\) and \({b'^*}^6_{(j)}\), which are hence grouped into $\Theta(\sqrt{n})$ clusters (dashed ellipses) with the same $\bs^{2}_{(i)}$ label. Each cluster contains $\Theta(\sqrt{n})$ logical qubits with different \({b'^*}^6_{(j)}\) labels. On the other hand, each logical qubit in copy~$\bl$ and CZ operator in copy $\gr$ carry a single cocycle label \({b'^*}^6_{(i)}\) and $\bs'^{2}_{(i)}$ respectively.
(b) A controlled fraction of logical qubits in copy~$\rd$ is initialized in the $\lo{\ket{0}}$ state, which effectively disables all hyper-edges  incident on those qubits. The remaining logical qubits in copy $\rd$, one in each cluster, are initialized in the $\lo{\ket{+}}$ state along with all the logical qubits in copy $\bl$.
(c) After performing the gauging measurement protocol with $O(d)$ rounds of error correction, one prepares $\Theta(\sqrt{n})$ disjoint CZ magic states. 
(d) The magic states are then swapped from the fountain register into the memory register using logical Clifford gates and subsequently converted into parallelizable logical non-Clifford gates via gate teleportation.}
    \label{fig:interaction-hypergraph}
\end{figure*}

As we can see, in the general case if we initialize all logical qubits in the $\lo{\ket{+}}$ states, we will prepare a highly entangled logical magic state coupling many logical qubits together.  Nevertheless, for more parallelizable computation,  it is desirable to make the logical magic states supported on disjoint logical qubits.  It has been proposed in Ref.~\cite{zhu2023non} and elaborated in Ref.~\cite{zhu2025topological} that qLDPC codes can be viewed as a resource for fault-tolerantly preparing a large number of logical magic states in parallel,  with the corresponding scheme dubbed `\textit{magic state fountain}'. 

With the specific construction in Sec.\ref{sec:triple_intersection_construction}, we can obtain the following theorem:
\begin{theorem}\label{theorem:protocol}
There exists a gauging measurement protocol applied on the 2D thickened hypergraph-product code that can produce $\Theta(\sqrt{n})$ disjoint logical CZ magic state with distance $\Omega(\sqrt{n})$ (assuming Statement \ref{statement:TQFT} holds for the twisted code).  
\end{theorem}
\begin{proof}
According to Eq.~\eqref{eq:cocycle_decomposition} in Sec.\ref{sec:triple_intersection_construction},  we obtain a triple intersection structure of the cocycles $\red{\alpha^8}$, $\blue{\beta^6}$ and $\green{\gamma^2}$ as shown by the hypergraph in Fig.~\ref{fig:interaction-hypergraph}(a).  Each vertex in the hypergraph represents a cocycle class, while each hyper-edge connecting to three vertices represents a non-trivial triple intersection $\int_{\L} \red{\alpha^8} \cup \blue{\beta^6} \cup \green{\gamma^2}=1 $.  In particular, each red vertex corresponds to cocycle class $\red{{\alpha}^{8}} =  \bs^2 \otimes {\bs'^*}^6$, where the tensor decomposition gives rise to two cocycle labels $\bs^2_{(i)}$ and ${\bs'^*}^6_{(j)}$ within each CW complex $\L_x$ and $\L_y$ as tunable knobs.  The red vertices in the hypergraph can hence be grouped into clusters (dashed ellipse), with each cluster labeled by the same $\bs^2_{(i)}$ and vertices in that cluster labeled by $\bs^2_{(j)}$. On the other hand, for the blue and green vertices corresponding to cocycle class $\blue{{\beta}^{6}} = {\bs^*}^6  \otimes  \cs^0$ and $\green{{\gamma}^{2}} = \cs^0 \otimes  \bs'^2$, there is only one varying label ${\bs^*}^6$ and $\bs'^2$ respectively since the CW complex $\L_x$ or $\L_y$ has only a unique 0-cocycle class $\cs^0$. 

  Now we set only one of the red logical qubit (vertex) in each cluster (dashed ellipse) at $\lo{\ket{+}}$ state, and the others in the $\lo{\ket{0}}$ states. All the blue logical qubits (vertices) are set also in the $\lo{\ket{+}}$ states. This gives rise to a subset of $\Theta(\sqrt{n})$ logical qubits initialized at $\lo{\ket{+}}$ states, while the rest are in the $\lo{\ket{0}}$ states. With this choice, for a subset of basis cocycles $\{\green{\gamma'^2}\} \subset \{\green{\gamma^2}\}$, there is a unique pair $(\alpha'^8, \beta'^6)$ that has non-trivial intersection of with $\green{\gamma'^2}$, i.e., $\int_{\L} \red{\alpha'^8} \cup \blue{\beta'^6} \cup \green{\gamma'^2} = 1$. 
We hence get the following output state: 
\begin{align}
\lo{\ket{\psi}}_f =& \prod_{(\red{\alpha'}, \blue{\beta'})} \hat{P}_{\pm 1}\big[ \lo{\text{CZ}}(\red{\alpha'}, \blue{\beta'}) \big] \lo{\ket{++}}_{\red{\alpha'}, \blue{\beta'}}
\end{align}
where all the logical operator $\lo{\text{CZ}}(\red{\alpha'}, \blue{\beta'})$ are supported on disjoint pairs of logical qubits labeled by $(\red{\alpha'}, \blue{\beta'})$.  If the logical CZ measurement outcome  on a logical qubit pair $(\red{\alpha'}, \blue{\beta'})$ is $(-1)^{\green{\rho_\gamma}}=+1$,  we get the output state on these two logical qubits as the logical CZ magic state:
\begin{align}
& \hat{P}_{+ 1}\big[ \lo{\text{CZ}}(\red{\alpha'}, \blue{\beta'}) \big] \lo{\ket{++}}_{\red{\alpha'}, \blue{\beta'}} = \lo{\ket{\text{CZ}}}_{\red{\alpha'}, \blue{\beta'}} \cr
\equiv & \frac{1}{\sqrt{3}}(\lo{\ket{00}}_{\red{\alpha'}, \blue{\beta'}} + \lo{\ket{01}}_{\red{\alpha'}, \blue{\beta'}} + \lo{\ket{10}}_{\red{\alpha'}, \blue{\beta'}}). 
\end{align}
On the other hand, if the logical CZ measurement outcome is $(-1)^{\green{\rho_\gamma}}=-1$, we end up with the logical state:
\be
\hat{P}_{-1}\big[ \lo{\text{CZ}}(\red{\alpha'}, \blue{\beta'}) \big] \lo{\ket{++}}_{\red{\alpha'}, \blue{\beta'}} = \lo{\ket{11}}_{\red{\alpha'}, \blue{\beta'}}.
\ee

In each execution of gauging measurement, we get on average  $\Theta(\sqrt{n})$ $+1$-measurement outcomes, which leads to the preparation of $\Theta(\sqrt{n})$ disjoint CZ magic states. 

\end{proof}

\subsection{Spacetime path integral as logical action}\label{sec:spacetime_path_integral}

\subsubsection{General formalism }

We can generically represent the spacetime path integral as
\be
\mathcal{Z}= \sum_{\vec{\mathbf{c}}} \prod_\sigma w_\sigma(\vec{\mathbf{c}}), 
\ee
with $\vec{\mathbf{c}} \equiv \{ \red{a^p}, \blue{b^q}, \green{c^s} \}$ represent all the cochain (gauge field) configurations.
The value $\cZ$ of the path integral can be considered as the outcome of the tensor contraction in the corresponding spacetime tensor-network equivalent to the state sum over all the cochain (gauge field) configurations.   

More importantly, $\cZ$ has a meaning when evaluated on a spacetime complex $\tilde{\L}$ with an input or output \textit{state boundary} (denoted by $\B_I$ and $\B_F$), placed in the beginning or end of the protocol. The corresponding spacetime tensor-network should be considered as a matrix (tensor) product operator (MPO) $\cT$, which  can be graphically represented as: 
\be
\cT \ = \raisebox{-1.5cm}{\includegraphics[width=0.3\linewidth]{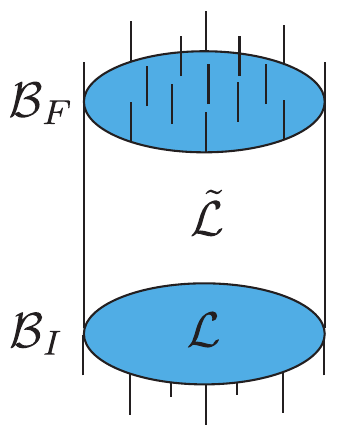}}  \quad ,
\ee
where the time direction goes upward. The above  operator $\T$ is defined on a ``spacetime cylinder" $\tilde{\L}=\L \times I_t$, where the input and output state boundaries corresponding to the space complex $\L$.
Note that there are no summation of the cochain variables at the state boundary, i.e., some tensor legs (shown above) at the boundary have not been contracted.
For each cochain (gauge field) configuration on the state boundary $\vec{\mathbf{c}}_\B \equiv \{a^p|_\B$,  $b^q|_\B,  c^s|_\B \}$, there exists a separate amplitude $\cZ(\vec{\mathbf{c}}_\B)$.   We can hence define an input (output) \textit{boundary state} $\ket{\psi_\B}$ via the evaluation of the path integral with an input (output) state boundary $\B$, with the wavefunction amplitude being $\bket{\vec{\mathbf{c}}_\B}{\psi_\B} = \cZ(\vec{\mathbf{c}}_\B)$.  We can diagrammatically represent the input (ket) and output (bra) boundary states as:
\be
\ket{\psi_{\B_I}}= \raisebox{-0.9cm}{\includegraphics[width=0.3\linewidth]{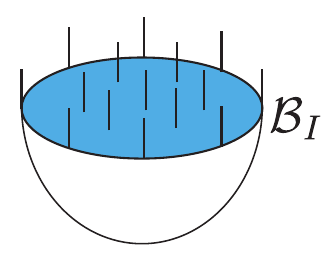}}, \quad 
\bra{\psi_{\B_F}}= \raisebox{-0.9cm}{\includegraphics[width=0.3\linewidth]{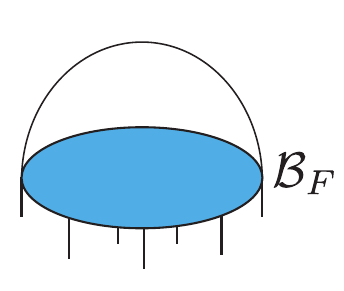}}. 
\ee
Note that here we have generalized the TQFT notion of associating the ket and bra with  open spacetime manifolds $\M$ that have state boundaries.  Here, in the context of qLDPC codes, we associate the ket and bra with spacetime Poincar\'e complexes $\tilde{\L}$ that have state boundaries.

The basis cocycle class of the input and output states has one-to-one correspondence with the logical state in the $Z$ basis.
The $i^\text{th}$ input (initial) and output (final) basis cocycle class are labeled as $I_i$ and $F_i$ respectively. An arbitrary cohomology class for  the input and output states can hence be represented as
\be
[a_\text{in}]= \sum_i u'_i I_i \quad \text{and} \quad [a_\text{out}]= \sum_i u''_i F_i,
\ee
where $u'_i$ and $u''_i$ are the $\Z_2$ coefficients. We can hence label the basis of the input and output logical states as ${\ket{\vec{u}'}}$ and $\ket{\vec{u}''}$ respectively, where $\vec{u}' \equiv (u'_1, u'_2, \cdots )$ and    $\vec{u}'' \equiv (u''_1, u''_2, \cdots )$ are $\Z_2$ vectors.   We can hence use the path integral to define the operator $\T$, whose matrix elements  $\boket{\vec{u}''}{\T}{\vec{u}'}$ will determine the logical action applied between the input and output states.  The matrix element can be diagrammatically represented as:
\be
\boket{\vec{u}''}{\T}{\vec{u}'} \ = \raisebox{-2cm}{\includegraphics[width=0.4\linewidth]{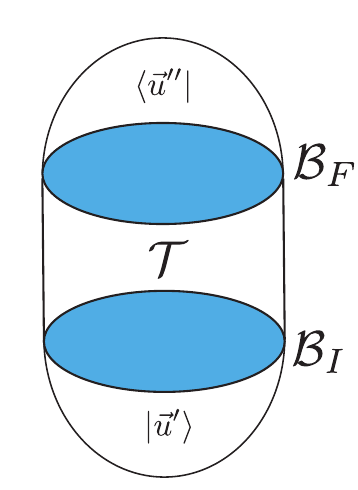}}  \quad,
\ee
where the spacetime complex corresponds to the operator $\mathcal{T}$ is sandwiched by the spacetime complexes corresponding to the  ket and bra. We can see that all the tensors at the state boundaries $\B_I$ and $\B_F$ have been contracted, and we hence end up with a scalar.

In order to determine the matrix, we choose a cocycle basis $\{\alpha_j\}$ depending on the bulk space-time topology.  An arbitrary bulk cohomology class can hence be expressed as $[a_\text{bulk}]=\sum_j l_j \alpha_j$, with $l_j \in \ZZ_2$.   We can hence label the cohomology class via the $\Z_2$ vector $\vec{l} = (l_1, l_2, \cdots)$. We then define the input restriction matrix $\mathcal{I}$,  where the restriction of $\alpha_j$ to the $i^\text{th}$ input cocycle class $I_i$ is captured by the $\ZZ_2$-valued  matrix elements $\mathcal{I}_{i,j}$.  Similarly, one can define the output restriction matrix $\mathcal{F}_{i,j}$, where $\alpha_j$ is restricted to the $i^\text{th}$ output cocycle class $F_i$.   We can hence evaluate the matrix element via the sum of the path integral: 
\be
\boket{\vec{u}''}{\T}{\vec{u}'} =  \sum_{\vec{l}: \  \mathcal{I}\cdot \vec{l}= \vec{u}', \  \mathcal{F}\cdot \vec{l}= \vec{u}''  } \cZ\big(\ \vec{l} \ \big). 
\ee

In general, the operator $\mathcal{T}$ is not a unitary, but acts as a quantum channel which depends on the measurement outcomes.  In the example which we consider, there will be non-trivial relative basis cocycles $C_i$ (charge worldsheet) connecting the input and output state boundaries and an arbitrary class can be expended with the basis as $[a_C]= \sum_i \rho_i C_i $.  The coefficients form a $\ZZ_2$ vector $
\vec{\mathbf{\rho}} =(\rho_1, \rho_2, \cdots)$.  The path-integral value,  denoted by $\cZ^{\vec{\mathbf{\rho}}}(\vec{l})$, hence depends on the measurement outcome $\vec{\mathbf{\rho}}$ of the non-trivial cocycle class $[a_C]$.  Therefore, the logical action corresponds to a projective measurement defined by a collection of operators ${\T}^{\vec{\mathbf{\rho}}}$, with the corresponding matrix elements being
\be
\boket{\vec{u}''}{{\T}^{\vec{\mathbf{\rho}}}}{\vec{u}'} =  \sum_{\vec{l}: \  \mathcal{I}\cdot \vec{l}= \vec{u}', \  \mathcal{F}\cdot \vec{l}= \vec{u}''  } \cZ^{\vec{\mathbf{\rho}}}\big(\ \vec{l} \ \big). 
\ee

\subsubsection{Imaginary-time vs real-time evolution: relation to quantum error correction}

Here, we briefly discuss the physical interpretation of the path integral.  More systematic study can be found in e.g.~the textbook \cite{Peskin:1995ev}.  The spacetime path integral we study in this paper can also be interpreted  as an imaginary-time evolution of a local quantum Hamiltonian in the literature \cite{Tsui:2019ykk}, and we hence also call it an \textit{imaginary-time path integral}.\footnote{
In the previous expression of topological actions we use Lorentzian signature; in the Euclidean signature with imaginary time, they have additional $i$ since it is odd under parity: $S_E= -iS$, and the path integral weight in the Euclidean signature $e^{-S_E}=e^{iS}$ is the same in both convention.
} For example, as shown in Ref.~\cite{Tsui:2019ykk}, one can express the matrix element introduced above as a imaginary-time evolution for time $T$ as follow:
\be
\boket{\vec{u}''}{\T}{\vec{u}'} = \boket{\vec{u}''}{e^{-TH_{\infty}}}{\vec{u}'},
\ee
where $H_{\infty}$ is a Hamiltonian with eigenvalue 0 and $\infty$, i.e., have an infinite gap.  Therefore, $\T$ corresponds to a projector with eigenvalues 0 and 1,  which projects to the ground state subspace of $H_{\infty}$. One can view the imaginary-time path integral as a non-unitary circuit composed of local unitary gates and projectors of the stabilizer generators to the +1 sectors with the form $P_{\sigma}=\frac{1}{2}(1+S_{\sigma})$, where $S_{\sigma}$ is a stabilizer generator defined on a cell $\sigma$.  The imaginary time evolution hence corresponds to a post-selected time evolution such that all the stabilizer measurements have +1 outcome. 

Now the actual quantum error correction (QEC) process is a real-time evolution 
corresponding to a quantum channel where we replace the projector in the imaginary-time path integral (tensor network) with stabilizer (syndrome) measurement, which can have $\pm 1$ outcome instead of just +1.  The $-1$ outcomes of the stabilizers correspond to the presence of \textit{topological defects} (here defined as the worldsheet of the excitations), and can also be captured by the \textit{defect-decorated imaginary-time path integral} such as Eq.~\eqref{eq:path_integral_defect}.

We now consider how to perform the error correction.   If the syndrome outcomes in the history are all $+1$, then we have successfully implemented the defect-free imaginary-time path integral.  In other cases, there exist some $-1$ syndromes in the history, which means we have implemented certain defect-decorated imaginary-time path integral.   The syndrome histories are then passed to the decoder which outputs a recovery operation $\R$ to close the defects.   If the decoder succeeds, the defects are closed in a homologically trivial way. This gives rise to a defect-decorated path integral which equals to the defect-free path integral, due to the topological invariance of the defects. On the other hand, if the defects are closed in a homological non-trivial way, a logical error occurs.        

From the above discussion, we can see that although the defect-free imaginary-time path integral corresponds to a post-selected process, the evaluation of it can still give the correct logical action which is the same as the logical action of the actual error-corrected process.  Therefore, we will derive the logical action via the imaginary-time path integral later in Sec.~\ref{sec:derivation_logical_action}.

More interestingly, as elaborated in Refs.~\cite{Bauer:2023awl,Bauer:2024qpc,Bauer:2024alh,Davydova:2025ylx}, the path integral such as Eq.~\eqref{eq:path_inegral_cup} can be used to construct the non-unitary quantum circuit of a fault-tolerant quantum protocol, including the unitary gates and measurements. This includes not only the special case of the standard Pauli stabilizer codes, but also more general spacetime codes such as the Floquet codes \cite{Hastings:2021ptn,Bauer:2024alh}. In both cases the corresponding quantum circuit is a Clifford circuit along with Pauli measurement.  
One can hence obtain a more unified framework to view a wide range of fault-tolerant protocols as a spacetime code, including the time evolution of a fault-tolerant memory as well as performing logical actions in the context of both stabilizer and Floquet codes.  In this unified picture, the spacetime path integral can serve as the major language,  which can be used to  both describe the error correction process and derive the logical action.  In the context of this paper, we have further generalized the spacetime code framework to Clifford-stabilizer codes, which now corresponds to a non-Clifford quantum circuit.   Although we have described our protocol in Sec.~\ref{sec:protocol}  in the context of a Clifford stabilizer LDPC code, the path integral in Eq.~\eqref{eq:path_inegral_cup} can be used to straightforwardly extend the current protocol also to the variation using a corresponding Floquet spacetime LDPC codes. 

Note that in this paper we have not derived the stabilizer measurement circuit directly from the path integral, but instead uses a gauging procedure to derive it.   However, such derivation from path integral has be carried out in Ref.~\cite{Davydova:2025ylx} in the context of $D_4$ non-Abelian topological code.   Such a procedure can be straightforwardly generalized to here since the path integral has a similar form.

\subsubsection{Derivation of the logical action from cellular path integral}\label{sec:derivation_logical_action}

Here we derive the logical action mentioned above using the cellular path integral in spacetime.   

The input of the protocol is two independent copies of qLDPC codes (\rd \ and \bl) as a $\red{\ZZ_2} \times \blue{\ZZ_2}$ gauge theory.   The cocycle basis for the input and output boundary states are denoted by $\{\red{I}_{\red{\alpha'}}\}$, $\{\blue{I}_{\blue{\beta'}}\}$ and $\{\red{F}_{\red{\alpha''}}\}$, $\{\blue{F}_{\blue{\beta''}}\}$ respectively.   After the code switching when going across the spacetime boundary, we enter the non-Abelian qLDPC codes corresponding to the twisted $\red{\ZZ_2} \times \blue{\ZZ_2} \times  \green{\ZZ_2}$ gauge theory.  As we have shown before, the relevant cocycle basis in the twisted gauge theory is $\{\red{\tilde{\alpha}^8}\}$,  $\{\blue{\tilde{\beta}^6}\}$ and $\{\green{\tilde{\gamma}^3}\}$.    We hence have the following input restriction matrix elements:  
\be
\mathcal{I}_{\red{\alpha', \alpha}} = \delta_{\red{\alpha', \alpha}}, \quad      \mathcal{I}_{\blue{\beta', \beta}} = \delta_{\blue{\beta', \beta}}, \quad \text{all} \ 0 \  \text{otherwise}.   
\ee
The Kronecker delta originates from the fact that the bulk cocycle class $[\red{\tilde{\alpha}^8}]$ and $[\blue{\tilde{\beta}^6}]$ are only connected to the input cocycle classes $[\red{I}_{\red{\alpha}}]$ and $[\blue{F}_{\blue{\beta}}]$ respectively, i.e., their corresponding representatives are identified at the input spacetime boundary $\mathcal{B}_I$:   $\red{\tilde{\alpha}^8}|_{\mathcal{B}_I} = \red{I}_{\red{\alpha}}$, $\blue{\tilde{\beta}^6}|_{\mathcal{B}_I}= \blue{I}_{\blue{\beta}}$.  
Similarly, we have the output restriction matrix elements: 
\be
\mathcal{F}_{\red{\alpha'', \alpha}} = \delta_{\red{\alpha'', \alpha}}, \quad      \mathcal{F}_{\blue{\beta'', \beta}} = \delta_{\blue{\beta'', \beta}}, \quad \text{all} \ 0 \  \text{otherwise},   
\ee
where we have the identification of the output spacetime boundary $\mathcal{B}_F$:
$\red{\tilde{\alpha}^8}|_{\mathcal{B}_F} = \red{F}_{\red{\alpha}}$, $\blue{\tilde{\beta}^6}|_{\mathcal{B}_F}= \blue{F}_{\blue{\beta}}$.  

There also exists relative basis cocycles $\green{{\tilde{\gamma^*}}^{14}}$  of the  bulk charge world-sheet of the $\gr$-type terminating at the input and output state boundaries, which are the dual of the basis cocycles  $\green{{\gamma}^{4}}$ satisfying the property that $\int_{\tilde{\L}} \green{\tilde{\gamma}^3} \cup \green{{\tilde{\gamma^*}}^{14}}$$=$$1$.  An arbitrary relative cocycle $\green{\zeta^{14}}$ corresponding to the open charge worldsheet ending in the input and output state boundaries  can be expanded with this basis as $\green{\zeta^{14}}$$=$$\sum_{\green{\gamma}}\green{\rho_\gamma} $$\cdot$$\green{{\tilde{\gamma^*}}^{14}}$, where we have also truncated $\green{{\tilde{\gamma^*}}^{14}}$ to relative cycle terminated at the state boundaries.  We denote the collection of all the measurement results as $\green{\vec{\rho}}=\{\green{\rho_{\gamma}}\}$.    The path integral weight  conditioned on the measurement results for a given bulk homology class parameterized by the winding number $\red{\vec{n}} \equiv \{\red{n_\alpha}\}$, $\blue{\vec{m}} \equiv \{\blue{m_\beta}\}$ and $\green{\vec{l}} \equiv \{\green{l_\gamma}\}$  can hence be expressed as
\be\label{eq:Z_component}
\mathcal{Z}^{\green{\vec{\rho}}}(\red{\vec{n}}, \blue{\vec{m}},  \green{\vec{l}}) = \prod_{\green{\gamma}} (-1)^{\green{\rho_\gamma} \cdot \green{l_\gamma} } \cdot \prod_{\red{\alpha}, \blue{\beta}}   \big[(-1)^{\red{n_\alpha} \cdot \blue{m_\beta} \cdot \green{l_\gamma}} \big]^{\int_{\tilde{\L}} \red{\tilde{\alpha}^8} \cup \blue{\tilde{\beta}^6} \cup \green{\tilde{\gamma}^3} }  
\ee
The matrix elements of the logical action can be expressed as the some of these weights over all the bulk cohomology classes satisfying the incidence conditions at both the input and output state boundaries:
\be
\boket{\red{\vec{u''}}, \blue{\vec{v''}}}{{\T}^{\green{\vec{\rho}}}}{\red{\vec{u'}}, \blue{\vec{v'}}} =
\sum_{ \substack{\red{\vec{n}}, \blue{\vec{m}},  \green{\vec{l}}: \   \mathcal{I} \cdot (\red{\vec{n}}, \blue{\vec{m}})= (\red{\vec{u'}}, \blue{\vec{v'}  }), \\  \qquad  \ \  \mathcal{F} \cdot (\red{\vec{n}},  \blue{\vec{m}}) = (\red{\vec{u''}}, \blue{\vec{v''}  })}}  \mathcal{Z}^{\green{\vec{\rho}}}(\red{\vec{n}}, \blue{\vec{m}},  \green{\vec{l}} )
\ee

We can see that if $(\red{\vec{u'}}, \blue{\vec{v'}  }) \neq (\red{\vec{u''}}, \blue{\vec{v''}  }) $, there is no $\red{\vec{n}}$ and $\blue{\vec{m}}$ that satisfies the incidence conditions at the input and output state boundaries, and hence only have zero contribution. Therefore the off-diagonal matrix elements of ${\T}^{\green{\vec{\rho}}}$ are all zeros. Meanwhile, for the diagonal matrix elements, i.e., $(\red{\vec{n}},  \blue{\vec{m}})=(\red{\vec{u'}}, \blue{\vec{v'}  }) = (\red{\vec{u''}}, \blue{\vec{v''}})$, we have fixed the the particular cohomology classes on the {\rd} and {\bl} copies, i.e., $(\red{\vec{n}},  \blue{\vec{m}})$.  In this case, we are only summing over the cohomology class in the {\gr} copy captured by $\green{\vec{l}} \equiv \{\green{l_\gamma} \in \ZZ_2 \}$. For each component $\green{l_\gamma}$, we are summing over two values   $\green{l_\gamma}= 0, 1$.   We can hence express the diagonal matrix elements as 
\begin{align}
& \boket{\red{\vec{m}}, \blue{\vec{n}}}{{\T}^{\green{\vec{\rho}}}}{\red{\vec{m}}, \blue{\vec{n}}} \cr
=& \sum_{\green{\vec{l}}=\{\green{l_\gamma}\}} \prod_{\green{\gamma}} (-1)^{\green{\rho_\gamma} \cdot \green{l_\gamma} } \cdot \prod_{\red{\alpha}, \blue{\beta}}   \big[(-1)^{\red{n_\alpha} \cdot \blue{m_\beta} \cdot \green{l_\gamma}} \big]^{\int_{\tilde{\L}} \red{\tilde{\alpha}^8} \cup \blue{\tilde{\beta}^6} \cup \green{\tilde{\gamma}^3} }  \cr
=&  \prod_{\green{\gamma}} \sum_{\green{l_\gamma}} (-1)^{\green{\rho_\gamma} \cdot \green{l_\gamma} } \cdot \prod_{\red{\alpha}, \blue{\beta}}   \big[(-1)^{\red{n_\alpha} \cdot \blue{m_\beta} \cdot \green{l_\gamma}} \big]^{\int_{\tilde{\L}} \red{\tilde{\alpha}^8} \cup \blue{\tilde{\beta}^6} \cup \green{\tilde{\gamma}^3} } \cr
\equiv &\prod_{\green{\gamma}} \boket{\red{\vec{m}}, \blue{\vec{n}}}{\tilde{\T}^{\green{\rho_\gamma}}_{\green{\gamma}}}{\red{\vec{m}}, \blue{\vec{n}}}, 
\end{align}
where we have factorized the diagonal operator ${\T}^{\green{\vec{\rho}}}$ into a product of diagonal operators:  ${\T}^{\green{\vec{\rho}}}= \prod_{\green{\gamma}} \tilde{\T}^{\green{\rho_\gamma}}_{\green{\gamma}}$. We can hence analyze the matrix elements of each factor operator, i.e., $\boket{\red{\vec{m}}, \blue{\vec{n}}}{\tilde{\T}^{\green{\rho_\gamma}}_{\green{\gamma}}}{\red{\vec{m}}, \blue{\vec{n}}}$.  

We discuss the situations of the two different measurement outcome $\green{\rho_\gamma}$:
\begin{enumerate}
\item 
For $\green{\rho_\gamma}=0$,   we have 
\begin{align}
&\boket{\red{\vec{m}}, \blue{\vec{n}}}{\tilde{\T}^{\green{0}}_{\green{\gamma}}}{\red{\vec{m}}, \blue{\vec{n}}}  \cr
=& \sum_{\green{l_\gamma}=0,1} \prod_{\red{\alpha}, \blue{\beta}}   \big[(-1)^{\red{n_\alpha} \cdot \blue{m_\beta} \cdot \green{l_\gamma}} \big]^{\int_{\tilde{\L}} \red{\tilde{\alpha}^8} \cup \blue{\tilde{\beta}^6} \cup \green{\tilde{\gamma}^3} }    \cr
=& 1 + \prod_{\red{\alpha}, \blue{\beta}}   \big[(-1)^{\red{n_\alpha} \cdot \blue{m_\beta} } \big]^{\int_{\tilde{\L}} \red{\tilde{\alpha}^8} \cup \blue{\tilde{\beta}^6} \cup \green{\tilde{\gamma}^3} }  \cr
=&  \boket{\red{\vec{m}}, \blue{\vec{n}}}{ \Big(1 +  \prod_{\red{\alpha}, \blue{\beta}} \lo{\text{CZ}}(\red{\alpha}, \blue{\beta})^{\int_{\tilde{\L}} \red{\tilde{\alpha}^8} \cup \blue{\tilde{\beta}^6} \cup \green{\tilde{\gamma}^3} } \Big) }{\red{\vec{m}}, \blue{\vec{n}}}  \cr
\propto& \boket{\red{\vec{m}}, \blue{\vec{n}}}{\hat{P}_{+1}\big[\prod_{\red{\alpha}, \blue{\beta}} \lo{\text{CZ}}(\red{\alpha}, \blue{\beta})^{\int_{\tilde{\L}} \red{\tilde{\alpha}^8} \cup \blue{\tilde{\beta}^6} \cup \green{\tilde{\gamma}^3} } \big] }{\red{\vec{m}}, \blue{\vec{n}}}.  \cr
\end{align}
Note that in the last line we have used `$\propto$' to  ignore the difference of a $\frac{1}{2}$ normalization constant, and one should keep in mind that our definition of $\tilde{\T}^{\green{0}}_{\green{\gamma}}$ has not been properly normalized, which is not really important here.  Here, $\hat{P}_{+1}\big[\prod_{\red{\alpha}, \blue{\beta}} \lo{\text{CZ}}(\red{\alpha}, \blue{\beta})^{\int_{\tilde{\L}} \red{\tilde{\alpha}^8} \cup \blue{\tilde{\beta}^6} \cup \green{\tilde{\gamma}^3} } \big]$ represents the projector to the $(+1)$-eigenstate of a product of logical CZ  acting on logical qubits with their logical-$X$ support satisfying the triple intersection condition $\int_{\tilde{\L}} \red{\tilde{\alpha}^8} \cup \blue{\tilde{\beta}^6} \cup \green{\tilde{\gamma}^3}=1$.  Note that the expression is slightly different from the projector in Eq.~\eqref{eq:output_logical_state}, since here we use the spacetime cocycles on the spacetime complex $\tilde{\L}=\L \otimes I_t$ such as $\green{\tilde{\gamma}^3}$ instead of the cocycles on the space complex $\tilde{\L}$ such as $\green{\gamma^2}$.   
This projector can also be rewritten into the form $  {\hat{P}_{+ 1}\big[\widetilde{\text{CZ}}^{\rd, \bl}_{\green{\gamma^*_{14}}}\big] }$, which is the projector
onto the $(+1)$-eigenstate of the transversal CZ operator between copy $\rd$ and $\bl$ supported on the cycle $\green{\gamma^*_{14}}$ of either the space complex $\L$ or the spacetime complex $\tilde{\L}$.

\item 
For $\green{\rho_\gamma}=1$,   we have 
\begin{align}
&\boket{\red{\vec{m}}, \blue{\vec{n}}}{\tilde{\T}^{\green{1}}_{\green{\gamma}}}{\red{\vec{m}}, \blue{\vec{n}}}  \cr
=& \sum_{\green{l_\gamma}=0,1} (-1)^{\green{l_\gamma}} \cdot \prod_{\red{\alpha}, \blue{\beta}}   \big[(-1)^{\red{n_\alpha} \cdot \blue{m_\beta} \cdot \green{l_\gamma}} \big]^{\int_{\tilde{\L}} \red{\tilde{\alpha}^8} \cup \blue{\tilde{\beta}^6} \cup \green{\tilde{\gamma}^3} }    \cr
=& 1 - \prod_{\red{\alpha}, \blue{\beta}}   \big[(-1)^{\red{n_\alpha} \cdot \blue{m_\beta} } \big]^{\int_{\tilde{\L}} \red{\tilde{\alpha}^8} \cup \blue{\tilde{\beta}^6} \cup \green{\tilde{\gamma}^3} }  \cr
=&  \boket{\red{\vec{m}}, \blue{\vec{n}}}{ \Big(1 - \prod_{\red{\alpha}, \blue{\beta}} \lo{\text{CZ}}(\red{\alpha}, \blue{\beta})^{\int_{\tilde{\L}} \red{\tilde{\alpha}^8} \cup \blue{\tilde{\beta}^6} \cup \green{\tilde{\gamma}^3} } \Big) }{\red{\vec{m}}, \blue{\vec{n}}}  \cr
\propto& \boket{\red{\vec{m}}, \blue{\vec{n}}}{\hat{P}_{-1}\big[\prod_{\red{\alpha}, \blue{\beta}} \lo{\text{CZ}}(\red{\alpha}, \blue{\beta})^{\int_{\tilde{\L}} \red{\tilde{\alpha}^8} \cup \blue{\tilde{\beta}^6} \cup \green{\tilde{\gamma}^3} } \big] }{\red{\vec{m}}, \blue{\vec{n}}}.  \cr
\end{align}
Here, we have $\hat{P}_{-1}[\bullet]$ represents the project of the product of logical CZ satisfying the triple intersection.  Again, it can be expressed as the projector onto the ($-1$)-eigenstate of the transversal CZ operator on $\green{\gamma^*_{14}}$, i.e., $  {\hat{P}_{-1}\big[\widetilde{\text{CZ}}^{\rd, \bl}_{\green{\gamma^*_{14}}}\big] }$.
\end{enumerate}

In sum, we see the logical action ${\T}^{\green{\vec{\rho}}}$ is a diagonal operator which can be expressed as a product of projectors onto the logical CZ operators satisfying the triple intersection condition:
\begin{align}\label{eq:logical_action_form}
{\T}^{\green{\vec{\rho}}}  =& \prod_{\green{\gamma}} {\hat{P}_{1-2\green{\rho_\gamma}}\big[\prod_{\red{\alpha}, \blue{\beta}} \lo{\text{CZ}}(\red{\alpha}, \blue{\beta})^{\int_{\tilde{\L}} \red{\tilde{\alpha}^8} \cup \blue{\tilde{\beta}^6} \cup \green{\tilde{\gamma}^3} } \big] }, \cr
\equiv&  \prod_{\green{\gamma}}   \hat{P}_{1-2\green{\rho_\gamma}}\big[\widetilde{\text{CZ}}^{\rd, \bl}_{\green{\gamma^*_{14}}}\big]
\end{align}
where $1-2\green{\rho_\gamma}=\pm 1$ for $\green{\rho_{\gamma}} \in \{0, 1\}$. It is also equivalent to the product of projectors onto the $\pm 1$ eigenstate of the transversal operator $\widetilde{\text{CZ}}^{\rd, \bl}_{\green{\gamma^*_{14}}}$ as shown above in the second line.

\subsection{Pullback to the skeleton 2D chain complex and gauging measurement of the 0-form subcomplex symmetries}\label{sec:pullback_protocol}

In this subsection, we translate the above gauging measurement protocol for the CW complex to the skeleton 2D chain complex.  

Instead of measuring the transversal CZ gate as a higher-form symmetry operator in Eq.~\eqref{eq:measure_higher-form},  we will perform gauging measurement of the following 0-form subcomplex symmetry operator: 
\begin{align}
\widetilde{\text{CZ}}^{\rd, \bl}_{\green{\bar \gamma^0}} \equiv \widetilde{\text{CZ}}^{\rd, \bl}_{\green{\bar\gamma^*_{2}}} 
= (-1)^{\int_{\green{\bar\gamma^*_{2}}} \red{\hat{\bar a}^1} \cup \blue{\hat{\bar b}^1}  } \equiv (-1)^{\int \red{\hat{\bar a}^1} \cup \blue{\hat{\bar b}^1} \cup \green{\bar \gamma^0} },
\end{align}
as has been introduced in Sec.~\ref{sec:0-form}. By expressing the operator-valued cochain with the cocycle basis $\{\red{\bar\alpha^1}\}$ and $\{\blue{\bar \beta^1}\}$:  $\red{\hat{\bar a}^1}$$=$$ \sum_{\red{\alpha}}\red{\hat{n}_\alpha \bar\alpha^1}$ and $\blue{\hat{\bar b}^1}$$=$$\sum_{\blue{\beta}}\blue{\hat{m}_\beta \bar\beta^1}$, we can re-write the transversal CZ operator as:
\begin{align}\label{eq:logical_CZ}
\widetilde{\text{CZ}}^{\rd, \bl}_{\green{\gamma^*_{2}}} =&  \prod_{\red{\alpha}, \blue{\beta}} (-1)^{\int_{\green{\gamma^*_{2}}} \red{\hat{n}_\alpha \bar \alpha^1} \cup \blue{\hat{m}_\beta \bar \beta^1} } = \prod_{\red{\alpha}, \blue{\beta}} [(-1)^{\red{\hat{n}_{\alpha}} \cdot \blue{\hat{m}_\beta}}]^{\int_{\green{\gamma^*_{2}}} \red{ \bar \alpha^1} \cup \blue{ \bar \beta^1} }  \cr
=& \prod_{\red{\alpha}, \blue{\beta}} \lo{\text{CZ}}(\red{\alpha}, \blue{\beta})^{\int_{\green{\gamma^*_{2}}} \red{ \bar \alpha^1} \cup \blue{\bar  \beta^1} }  
\equiv \prod_{\red{\alpha}, \blue{\beta}} \lo{\text{CZ}}(\red{\alpha}, \blue{\beta})^{\int \red{\bar \alpha^1} \cup \blue{\bar \beta^1} \cup \green{\bar \gamma^0} }.
\end{align}
Note that the triple cup product (triple intersection)  in the exponent can be evaluated non-trivially as $\int \red{\bar \alpha^1}$$\cup $$\blue{\bar \beta^1}$$\cup$$\green{\bar \gamma^0}$$=$$1$ for certain triples according to Eq.~\eqref{eq:triple_cocycle_evaluation_spatial}.
The gauging measurement protocol in Sec.~\ref{sec:protocol} can be easily adapted to the case of the skeleton chain complex by simplying pulling back all the cocycles and cup product. 
The output state of this protocol hence has the following general form:
\begin{align}\label{eq:output_logical_state_skeleton}
\lo{\ket{\psi}}_f =& \prod_{\green{\gamma}} \hat{P}_{1-2\green{\rho_\gamma}}\big[\widetilde{\text{CZ}}^{\rd, \bl}_{\green{\gamma^*_{2}}}\big] \lo{\ket{\psi}}_i   \cr
=& \prod_{\green{\gamma}} \hat{P}_{1-2\green{\rho_\gamma}}\big[\prod_{\red{\alpha}, \blue{\beta}} \lo{\text{CZ}}(\red{\alpha}, \blue{\beta})^{\int  \red{\bar \alpha^1} \cup \blue{\bar \beta^1} \cup \green{\bar \gamma^0} } \big]\lo{\ket{\psi}}_i, \cr 
\end{align}
which is nothing but a translation of Eq.~\eqref{eq:output_logical_state}.

For the derivation from the spacetime path integral, we have the following path integral weight
\be\label{eq:Z_component_skeleton}
\mathcal{Z}^{\green{\vec{\rho}}}(\red{\vec{n}}, \blue{\vec{m}},  \green{\vec{l}}) = \prod_{\green{\gamma}} (-1)^{\green{\rho_\gamma} \cdot \green{l_\gamma} } \cdot \prod_{\red{\alpha}, \blue{\beta}}   \big[(-1)^{\red{n_\alpha} \cdot \blue{m_\beta} \cdot \green{l_\gamma}} \big]^{\int \red{\tilde{ \alpha}^1} \cup \blue{\tilde{\beta}^1} \cup \green{\tilde{\gamma}^1} },  
\ee
which is a translation of Eq.~\eqref{eq:Z_component}. One can hence derive the final logical action as follows:
\begin{align}\label{eq:logical_action_form_skeleton}
{\T}^{\green{\vec{\rho}}}  =& \prod_{\green{\gamma}} {\hat{P}_{1-2\green{\rho_\gamma}}\big[\prod_{\red{\alpha}, \blue{\beta}} \lo{\text{CZ}}(\red{\alpha}, \blue{\beta})^{\int \red{\tilde{\alpha}^1} \cup \blue{\tilde{\beta}^1} \cup \green{\tilde{\gamma}^1} } \big] }, \cr
\equiv&  \prod_{\green{\gamma}}   \hat{P}_{1-2\green{\rho_\gamma}}\big[\widetilde{\text{CZ}}^{\rd, \bl}_{\green{\bar{\gamma}^*_{2}}}\big],
\end{align}
which is translated from Eq.~\eqref{eq:logical_action_form}.

\section{Non-Abelian fusion and braiding statistics}
\label{sec:nonabelianfusionbraiding}
\subsection{Non-Abelian fusion rules}

As discussed in \cite{Hsin:2024pdi}, we can compute the fusion rule for the magnetic operators using the expressions (\ref{eqn:magneticoprprojection0}). The magnetic operators obey non-Abelian fusion due to the projection:
\begin{align}
    M^{\rd}(\Sigma_{\ds-p})\times M^{\rd}(\Sigma_{\ds-p}) = \sum_{\substack{\eta_q\in H_q(\Sigma_{\ds-p},\Z_2) \\\eta_s\in H_s(\Sigma_{\ds-p},\Z_2)}} W^{\bl}(\eta_q) W^{\gr}(\eta_s)~,
    \label{eq:fusing M red}
\end{align}
\begin{align}
    M^{\bl}(\Sigma_{\ds-q})\times M^{\bl}(\Sigma_{\ds-q}) = \sum_{\substack{\eta_p\in H_p(\Sigma_{\ds-q},\Z_2) \\\eta_s\in H_s(\Sigma_{\ds-q},\Z_2)}} W^{\rd}(\eta_p) W^{\gr}(\eta_s)~,
\end{align}
\begin{align}
    M^{\gr}(\Sigma_{\ds-r})\times M^{\gr}(\Sigma_{\ds-r}) = \sum_{\substack{\eta_p\in H_p(\Sigma_{\ds-r},\Z_2) \\\eta_q\in H_q(\Sigma_{\ds-r},\Z_2)}} W^{\rd}(\eta_p) W^{\bl}(\eta_q)~.
\end{align}
For instance, the fusion rule \eqref{eq:fusing M red} implies that the magnetic operators $M^{\rd}$ fuse into a sum of electric operators supported at cycles on $\Sigma_{\ds-p}$. This generalizes the fusion rule of magnetic flux excitations in non-Abelian $\mathbb{D}_4$ topological order in (2+1)D.

\subsection{Borromean ring braiding}
Let us show that the theory has Borromean ring braiding, which is characteristic of non-Abelian topological order. The Borromean ring braiding of extended excitations is discussed in \cite{Hsin:2024pdi}. Here, we will present several approaches to Borromean braiding statistics.

\subsubsection{Path integral}
In the continuum limit, the Borromean type of braiding process can be captured as follows. Consider spherical space\footnote{
The braiding of topological operators are local processes that can already be detected with ambient trivial topology.
} with operators supported on the contractible cycles $\tilde \Sigma_1,\tilde \Sigma_2,\cdots,\tilde \Sigma_N$ of $n_1,n_2,\cdots n_N$ dimensions in spacetime dimension $D$.
The contractible cycles in spherical space are bounded by chains $\tilde \Sigma_i=\partial \tilde V_i$. The Poincar\'e dual of the chain $\tilde V_i$ is a $(D-n_i-1)$-cochain $\text{PD}(\tilde V_i)$. For the set of contractible cycles with dimensions $n_i$ satisfying
\begin{equation}
    \sum (D-n_i-1)=D\;\Rightarrow \; \sum n_i =(N-1)D-N~,
\end{equation}
the Borromean ring braiding can be captured by 
\begin{equation}
    \int \cup_i  \text{PD}(\tilde V_i)~.
\end{equation}

In our case at hand, we want to compute the Borromean ring braiding for the twisted higher-form $\Z_2^3$ gauge theory with action:
\begin{align}
 S=   \pi\int\red{\tilde a^u}\cup  \red{d a^p} +\blue{\tilde b^v}\cup  \blue{d b^q} + \green{\tilde c^w}\cup \green{d c^s} +
    \red{a^p}\cup\blue{b^q}\cup \green{c^s}
\end{align}
where $\red{a^p},\blue{b^p},\green{c^s}$ are $\Z_2$ $p,q,s$-forms (cocycles), and $p+q+s=D$ with $D$ the spacetime dimensions. We want to compute the Borromean-like braiding of the magnetic operators that are supported on contractible cycles of dimension $n_1=D-p-1$, $n_2=D-q-1$, $n_3=D-s-1$ so $\sum n_i=2D-3$. The discussion is similar to e.g. \cite{Putrov:2016qdo} in $D=3,p=q=s=1$.

Let us show that there is a nontrivial link of three magnetic operators satisfying the following two properties:
\begin{itemize}
\item Any two of three magnetic operators are unlinked, while all three defects together form a non-trivially linked configuration.
\item The correlation function of the magnetic defects $M^{\rd},M^{\bl}, M^{\gr}$ evaluated on this Borromean link takes the value $(-1)$.
\end{itemize}

For illustration of a link, let us consider Euclidean spacetime, and construct a nontrivial link of three magnetic operators $M^{\rd}, M^{\bl}, M^{\gr}$ with the correlation function $(-1)$, which becomes an invariant of the link. These operators have $(D-p-1),(D-q-1),(D-s-1)$ dimensions respectively.

For convenience, we assign the coloring to the $D$ spacetime coordinates $x_1,\dots, x_D$ the colors: $x_j$ with $\red{1\le j\le p}$ are red, $\blue{p+1\le j\le p+q}$ are blue, $\green{p+q+1\le j\le D}$ are green. Then consider three hypercubes:
\begin{align}
\begin{split}
    \tilde C^{\rd}&:= \{x^{\rd} = 0, -L\le x^{\bl}\le L, -L\le x^{\gr}\le L \} \\
    \tilde C^{\bl}&:= \{x^{\bl} = 0, -L\le x^{\rd}\le L, -L\le x^{\gr}\le L \} \\
    \tilde C^{\gr}&:= \{x^{\gr} = 0, -L\le x^{\rd}\le L, -L\le x^{\bl}\le L \} \\
\end{split}
\end{align}
where $x^{\rd}$ is for all red coordinates, and same for other colors.
Then define the boundary of these hypercubes: 
\begin{align}
    \tilde \Sigma^{\rd} = \partial \tilde C^{\rd}~, \ \tilde \Sigma^{\bl} = \partial \tilde C^{\bl}~, \ \tilde \Sigma^{\gr} = \partial \tilde C^{\gr}~,   \end{align}
which are closed $(D-p-1),(D-q-1),(D-s-1)$ cycles respectively.

Then we want to consider a link of three operators 
\begin{align}
    M^{\rd}(\tilde\Sigma^{\rd})~, \ M^{\bl}(\tilde\Sigma^{\bl})~, \ M^{\gr}(\tilde\Sigma^{\gr})
\end{align}
But in order to form a link, we need to resolve the overlap of three membranes $\tilde\Sigma^{\rd}, \tilde\Sigma^{\bl}, \tilde\Sigma^{\gr}$. Let us consider the overlap of two cycles $\tilde\Sigma^{\rd}, \tilde\Sigma^{\bl}$. Its overlap $\tilde\Sigma^{\rd}\cap \tilde\Sigma^{\bl}$ is given by $\tilde\Sigma^{\rd\bl}:=\partial \tilde C^{\rd\bl}$, with
\begin{align}
    \tilde C^{\rd\bl} := \{x^{\rd} = x^{\bl} = 0, -L\le x^{\gr}\le L \}
\end{align}
Similarly, $\tilde \Sigma^{\bl}\cap \tilde\Sigma^{\gr}=\tilde\Sigma^{\bl\gr}=\partial \tilde C^{\bl\gr}$, $\tilde \Sigma^{\rd}\cap \tilde \Sigma^{\gr}=\tilde \Sigma^{\rd\gr}=\partial \tilde C^{\rd\gr}$, with
\begin{align}
\begin{split}
    \tilde C^{\bl\gr} &:= \{x^{\bl} = x^{\gr} = 0, -L\le x^{\rd}\le L \} \\
    \tilde C^{\rd\gr} &:= \{x^{\rd} = x^{\gr} = 0, -L\le x^{\bl}\le L \} \\
\end{split}
\end{align}

Then, we resolve the overlap of the three operators by perturbing the supports $\tilde \Sigma^{\rd}, \tilde \Sigma^{\bl},\tilde \Sigma^{\gr}$ as follows:
\begin{itemize}
    \item Along the overlap $\tilde \Sigma^{\rd\bl}$, we perturb the operator $M^{\rd}(\tilde\Sigma^{\rd})$ in the direction of $x^{\gr}$ so that the absolute value of $x^{\gr}$ slightly increases. Namely, $M^{\rd}(\tilde\Sigma^{\rd})$ is perturbed so that its support restricted at $x^{\bl}=0$ is $\partial \tilde C'^{\rd\bl}$ with
    \begin{align}
    \tilde C'^{\rd\bl} := \{x^{\rd} = x^{\bl} = 0, -L-\epsilon\le x^{\gr}\le L+\epsilon \}
\end{align}
with a small positive $\epsilon$. 

    \item To resolve the overlap at $\tilde \Sigma^{\bl\gr}$,  $M^{\bl}(\tilde \Sigma^{\bl})$ is perturbed so that its support restricted at $x^{\gr}=0$ is $\partial \tilde C'^{\bl\gr}$ with
    \begin{align}
    \tilde C'^{\bl\gr} := \{x^{\bl} = x^{\gr} = 0, -L-\epsilon\le x^{\rd}\le L+\epsilon \}
\end{align}
with a small positive $\epsilon$.

\item To resolve the overlap at $\tilde \Sigma^{\rd\gr}$,  $M^{\gr}(\tilde\Sigma^{\gr})$ is perturbed so that its support restricted at $x^{\rd}=0$ is $\partial \tilde C'^{\rd\gr}$ with
    \begin{align}
    \tilde C'^{\rd\gr} := \{x^{\rd} = x^{\gr} = 0, -L-\epsilon\le x^{\bl}\le L+\epsilon \}
\end{align}
with a small positive $\epsilon$.
\end{itemize}

Let us again denote the operators after the above perturbations by $M^{\rd}(\tilde \Sigma^{\rd}),  M^{\bl}(\tilde \Sigma^{\bl}),  M^{\gr}(\tilde \Sigma^{\gr})$. In (2+1)D, these three operators correspond to the configurations as shown in Fig.~\ref{fig:borromean}.
 In general dimensions, these three operators  are nontrivially linked as a whole. Indeed, one cannot shrink one of the membranes $M^{\rd}(\tilde\Sigma^{\rd})$ without intersecting with $M^{\bl}(\tilde\Sigma^{\bl})$, similar for other colors (related by cyclic permutation of colors $\rd\to\bl\to\gr\to\rd$). Meanwhile, once we eliminate one of the operators (say we eliminate $M^{\gr}(\tilde \Sigma^{\gr})$), then the configuration of two operators $M^{\rd},M^{\bl}$ are unlinked; one can topologically shrink $M^{\bl}$ to a point without intersecting with $M^{\rd}$.
Therefore this generalizes the Borromean ring in (2+1)D.

The above magnetic operators source the holonomy of $\Z_2$ gauge fields, such that 
\begin{align}
    \red{a^p} = \text{PD}(\tilde C^{\rd})~, \  \blue{b^q} = \text{PD}(\tilde C^{\bl})~, \ \green{c^s} = \text{PD}(\tilde C^{\gr})~.
\end{align}
These three hypercubes $\tilde C^{\rd},\tilde C^{\bl},\tilde C^{\gr}$ have nontrivial triple intersection, so the correlation function evaluates as
\begin{align}
\begin{split}
    \langle M^{\rd}(\tilde \Sigma^{\rd})M^{\bl}(\tilde \Sigma^{\bl})M^{\gr}(\tilde \Sigma^{\gr})\rangle &= (-1)^{\int \text{PD}(\tilde C^{\rd})\cup\text{PD}(\tilde C^{\bl})\cup\text{PD}(\tilde C^{\gr})} \\
    &= -1~.
    \end{split}
\end{align}
This correlation function is an invariant of the generalized Borromean ring.

We note that although we constructed the generalized Borromean ring using specific hypercubes, this choice is not essential: such link can be defined on a general CW complex, since the hypercubes can be topologically deformed while preserving the correlation function.

\begin{figure*}[htb]
\centering	\includegraphics[width=0.4\textwidth]{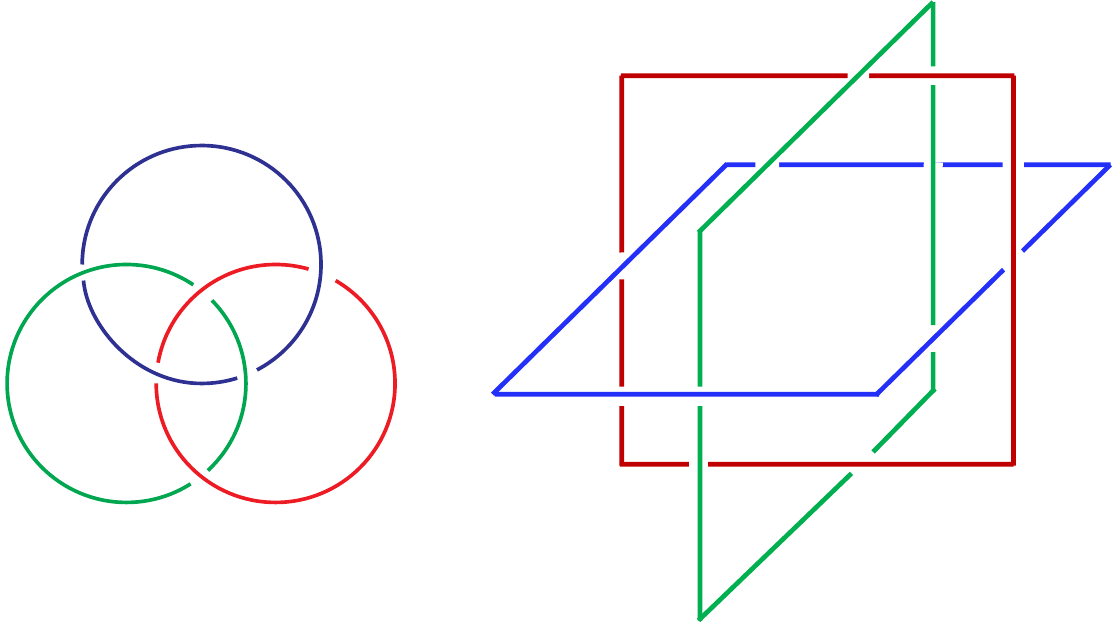}
    \caption{Left: The Borromean ring in (2+1)D. Right: It is convenient to illustrate a Borromean ring by three intersecting squares, which can be generalized to generic dimensions. Three squares are extended in $xy, yz, zx$ directions and have triple intersection.}
    	\label{fig:borromean}
\end{figure*}

\subsubsection{Operator-based derivation}

Here, we will present another approach, similar to the discussion in \cite{Witten:1989wf,Inamura:2023ldn} by taking time slices of the higher-dimensional analogue of Borromean ring process.

In the spacetime coordinate $\{x^{\rd},x^{\bl},x^{\gr}\}$ introduced above, suppose that the time direction is labeled red and the red coordinates decompose into $\{x^{\rd}\}=\{t, x_{\text{sp}}^{\rd}\}$; the spatial coordinate is given by $\{x_{}^{\rd},x^{\bl},x^{\gr}\}$. 

Then we consider three operators:
\begin{align}
    M^{\rd}(C_{}^{\rd})~, \quad M^{\bl}(C_{}^{\bl})~, \quad M^{\gr}(C_{}^{\gr})~,
\end{align}
where we define the operator supports as
\begin{align}
\begin{split}
    C_{}^{\rd}&:= \{x_{\text{sp}}^{\rd} = 0, -L\le x^{\bl}\le L, -L\le x^{\gr}\le L \} \\
    C_{}^{\bl}&:= \{x^{\bl} = 0, -L\le x_{\text{sp}}^{\rd}\le L, -L\le x^{\gr}\le L \} \\
    C_{}^{\gr}&:= \{x^{\gr} = 0, -L\le x_{\text{sp}}^{\rd}\le L, -L\le x^{\bl}\le L \} \\
\end{split}
\end{align}
and
\begin{align}
    \Sigma_{}^{\rd} := \partial C_{}^{\rd}~. 
\end{align}

Note that $M^{\bl}(C_{}^{\bl}),  M^{\gr}(C_{}^{\gr})$ have boundaries, where the magnetic excitations are created. Meanwhile, $M^{\rd}(\Sigma_{}^{\rd})$ is a symmetry operator that preserves the stabilizer space. They are explicitly represented by Clifford operators as
\begin{align}\label{eqn:magneticoprprojection}
    \begin{split}
        M^{\rd}(C_{}^{\rd}) &= P^{\rd}\left[\prod_{\sigma^*\in C_{}^{\rd}} X^{\rd}_{\sigma^*}\prod_{\sigma_q,\sigma_s\in C_{}^{\rd}}(1+Z^{\bl}_{\sigma_q})(1+Z^{\gr}_{\sigma_s})\right] P^{\rd} \\
        M^{\bl}(C_{}^{\bl}) &=P^{\bl}\left[\prod_{\sigma^*\in C_{}^{\bl}}X^{\bl}_{\sigma^*} \prod_{\sigma_p,\sigma_s\in C_{}^{\bl}}(1+Z^{\rd}_{\sigma_p})(1+Z^{\gr}_{\sigma_s})\right]P^{\bl} \\
        M^{\gr}(C_{}^{\gr}) &=P^{\gr}\left[\prod_{\sigma^*\in C_{}^{\gr}}X^{\gr}_{\sigma^*}\prod_{\sigma_p,\sigma_q\in C_{}^{\gr}}(1+Z^{\rd}_{\sigma_p})(1+Z^{\bl}_{\sigma_q})\right] P^{\gr}~,\\
    \end{split}
\end{align}
where $P^{\rd}$ is a projection onto the stabilizer subspace along $C_{}^{\rd}$, and $P^{\bl},P^{\gr}$ are projections onto the stabilizer subspace for stabilizers within the interior of $C_{}^{\bl},C_{}^{\gr}$ respectively.

The Borromean ring invariant is then evaluated by a sequence of operators acting on the stabilizer state $\ket{\psi}$,
\begin{align}
    \bra{\psi} M^{\bl}(C_{}^{\bl})^\dagger M^{\gr}(C_{}^{\gr})^\dagger M^{\rd}(\Sigma_{}^{\rd}) M^{\gr}(C_{}^{\gr})M^{\bl}(C_{}^{\bl})\ket{\psi}
\end{align}
The combined operation of magnetic operators $ M^{\gr}(C_{}^{\gr})M^{\bl}(C_{}^{\bl})$ generates the configuration of the gauge fields $\blue{\hat{b}^q},\green{\hat{c}^s}$ such that $(-1)^{\int_{C_{}^{\rd}}\blue{\hat{b}^q}\cup \green{\hat{c}^s}}=-1$. Note that the product of $X^{\rd}$ operators in the expression of $M^{\rd}(\Sigma_{}^{\rd})$ is identical to the product of $CZ$ operators within the stabilizer subspace,
\begin{align}
    \prod_{\sigma^*\in  \Sigma_{}^{\rd}} X^{\rd}_{\sigma^*} = \prod_{\substack{\sigma_q,\sigma_s\in C_{}^{\rd} \\ \sigma_q\cup\sigma_s\neq 0 }} CZ^{\bl,\gr}_{\sigma_{q},\sigma_{s}}
\end{align}
that evaluates the sum $(-1)^{\int_{C_{}^{\rd}}\blue{\hat{b}^q}\cup \green{\hat{c}^s}}$.
Therefore we get
\begin{align}
\begin{split}
    &\text{Arg}\;\bra{\psi} M^{\bl}(C_{}^{\bl})^\dagger M^{\gr}(C_{}^{\gr})^\dagger M^{\rd}(\Sigma_{}^{\rd}) M^{\gr}(C_{}^{\gr})M^{\bl}(C_{}^{\bl})\ket{\psi} \\
    &= \text{Arg}\left(-\bra{\psi} M^{\bl}(C_{}^{\bl})^\dagger M^{\gr}(C_{}^{\gr})^\dagger  M^{\gr}(C_{}^{\gr})M^{\bl}(C_{}^{\bl})\ket{\psi} \right)\\
    &=\pi~.
\end{split}
\end{align}
We note that inside the big bracket in the second line is the negative of the norm of $M^{\gr}(C_{}^{\gr})M^{\bl}(C_{}^{\bl})\ket{\psi}$, which is negative as long as we consider the state $\ket{\psi}$ such that the state $M^{\gr}(C_{}^{\gr})M^{\bl}(C_{}^{\bl})\ket{\psi}$ is not zero.
The above therefore diagnoses the Borromean ring invariant on the stabilizer code.

\section{Discussion and Outlook}
\label{sec:outlook}

The current paper focus on the construction of the non-Abelian qLDPC codes and the logical action via the gauging measurement protocol or  the corresponding spacetime path integral.  We have not investigated the decoding algorithm and the fault-tolerance proof of the gauging protocol.  In the special case that the skeleton classical code in the twisted hypergraph-product code construction is a repetition code which gives rise to a 2D topological code,  a `just-in-time' decoder has been developed in Ref.~\cite{Davydova:2025ylx}, along with a proof of the fault-tolerant threshold.  In the follow-up work, we will devise a decoder for the general twisted 2D HGP code using a similar strategy as the just-in-time decoder which first closes the non-Abelian magnetic flux defects and then then cleans up the Abelian electric charge excitations.    The proof strategy will need to be modified for the case that the skeleton classical codes are asymptotically good expander codes.

Next, it should be straightforward to generalize the twisted HGP code construction to a twisted (gauged) balanced product code \cite{Breuckmann:2021_balanced}.  This will further enhance the code distance beyond $\Omega(\sqrt{n})$. 

For the Clifford-stabilizer codes constructed in  Refs.~\cite{Kobayashi:2025cfh,Hsin:2025zgn,Warman:2025hov} in the context of topological codes, one can directly find  transversal non-Clifford gates using automorphism symmetries of the corresponding gauge theory in the presence of boundaries.   In the twisted 2D HGP codes studied here, the non-Clifford logical operation is constructed instead using the spacetime path integral, and there are no boundaries in these codes. It would be interesting to also find transversal non-Clifford gates in the  twisted 2D HGP codes.

While the twisted 2D HGP code we construct here is a Clifford stabilizer code, it is straightforward to generalize it to a twisted $N$-dimensional HGP code as a higher Clifford-hierarchy stabilizer code, such as the models discussed in \cite{Kobayashi:2025cfh,Hsin:2025zgn} with $(N+1)$-fold  intersection instead of triple intersection in the path integral.
The corresponding TQFT has the topological action of the form
\begin{align}
    S&=\sum_{i}\pi\int_{\tilde \L} a_i^{m_i}\cup d\tilde a^{l_i} +\pi \int_{\tilde \L} a_1^{m_1}\cup \cdots \cup a^{m_{N+1}}_{N+1}~,
\end{align}
where the spacetime dimension of the Poincare CW complex $\tilde{\L}$ is $D=m_i+l_i+1=\sum_{j=1}^{N+1} m_j$.
The theory corresponds to a Clifford-hierarchy stabilizer code consisting of stabilizer generators of the form $\prod Z$ and $\prod X (\text{C}^{N-1}\text{Z})$ \cite{Kobayashi:2025cfh,Hsin:2025zgn,Warman:2025hov}. The theory also has non-Abelian excitations that obey non-Abelian fusion and braiding, in particular $N$-body analogue of Borromean ring braiding.

One can then generalize the logical action through a spacetime path integral similar to Eq.~\eqref{eq:Z_component} or Eq.~\eqref{eq:Z_component_skeleton} involving an $N$-fold intersection which then performs the gauging measurement of addressable logical $\text{C}^{N-1}\text{Z}$ gates and creates a magic state fountain to produce $\text{C}^{N-1}\text{Z}$ magic states in parallel.   This will effectively implement a logical gate in the $(N+1)^\text{th}$-level of Clifford hierarchy in the twisted (gauged) $N$-dimensional HGP code, which boosts the computational power of the untwisted $N$-dimensional HGP code that only admits  transversal gates in the $N^\text{th}$-level of hierarchy. We call this approach a `\textit{gauging booster}'. 

We note that for the twisted 2D HGP code the gauging measurement needs $O(d)$ rounds of error corrections, which can be considered as a space-time tradeoff when compared to the single-shot transversal CCZ gates in the 3D HGP code.  Nevertheless, for twisted HGP codes with dimension $N \ge 3$, it is expected that the codes will admit single-shot state preparation similar to their untwisted counterpart \cite{golowich2025constant} due to the confinement of the flux excitations.    In that case, the gauging booster may enhance the computational power just in a single-shot protocol  without the additional  $O(d)$ time overhead.  

We conjecture that the gauging booster will allow native logical action in arbitrary Clifford hierarchies implemented via the  twisted (gauged) 2D product codes (including HGP and balanced product) in the form of Clifford-hierarchy stabilizer codes, which is in the same spirit as the recent results in Ref.~\cite{Warman:2025hov} in the context of 2D topological codes.

Finally, it would be interesting to explore a wider class of non-Abelian qLDPC codes and the associated combinatorial TQFTs, including the more general string-net models \cite{levin2005} and quantum double models \cite{kitaev2003}.

\section*{Acknowledgment}

We thank Ben Brown for insightful discussions about the error correction details of the gauging protocol and comments on the manuscript. G.Z. also thanks Virgile Guemard, Louis Golowich, Andreas Bauer, and Julio  Magdalena de la Fuente for helpful discussions. G.Z. is supported by the U.S. Department of Energy, Office of Science, National Quantum
Information Science Research Centers, Co-design Center
for Quantum Advantage (C2QA) under contract number
DE-SC0012704. 
R.K. is supported by the U.S. Department of Energy through grant number DE-SC0009988 and the Sivian Fund.
P.-S.H. is supported by Department of Mathematics King’s College
London. G.Z. and P.-S.H. thank Simons Foundation for hosting the Collaboration on Global Categorical Symmetries Annual Meeting 2025, during which part of the work is completed.

\appendix

\begin{appendix}

\section{Derivation of the Clifford stabilizer models by gauging a higher-form SPT}\label{app:gauging_derivation}

In this appendix, we present the derivation of the Clifford stabilizer codes presented in Sec.~\ref{sec:code_construction}.  This is also equivalent to construct the parent Hamiltonian whose ground space is the code space $\C$.  We follow the procedure in Ref.~\cite{Hsin2024_non-Abelian} which studies a general class of higher-form gauge theories called \textit{cubic theory} and is generalized to our non-Abelian qLDPC construction on a CW complex.   We start by constructing the Hamiltonian of a class of symmetry-protected topological (SPT) phases protected by higher-form symmetries, and then gauge these symmetries minimally.

\subsection{SPT Hamiltonian}
The SPT phase we consider here has $\red{\ZZ_2^{(p-1)}} \otimes \blue{\ZZ_2^{(q-1)}} \otimes \green{\ZZ_2^{(s-1)}}$ higher form symmetries, i.e., with ($p-1$)-form,  ($q-1$)-form and ($s-1$)-form symmetries.  For spacetime dimension $D$, this family of models satisfy the constraint $p+q+s=D$.

We start with the $D$-dimensional spacetime complex $\tilde{\L} \equiv \L \otimes I_t$, which can either be a CW complex (more specifically a Poincar\'e complex) or the triangulation associated with the constructed manifold $\M$ obtained from the generalized code-to-manifold mapping introduced in Sec.~\ref{sec:code_to_manifold}.  We start with three independent copies of lattice models, denoted by \red{red} (\rd), \blue{blue} (\bl) and \green{green} (\gr).  The qubits in copy \rd \ are placed on $(p-1)$-cells, qubits in copy \bl \ on $(q-1)$-cells, and qubits in copy \gr \ on $(s-1)$-cells. The corresponding Pauli operators are denoted by $\{\tilde X^{i},\tilde Y^{i},\tilde Z^{i}\}$ with the index $i=\rd, \bl, \gr$.   Where we have used tilde to denote the operators for the ungauged SPT models to distinguish them with the Pauli operators in the gauged model introduced later.   

We now introduce the operator-valued cochain formalism developed in Refs.~\cite{Hsin2024_non-Abelian, zhu2023non, Hsin2024:classifying, zhu2025topological, zhu2025topological}.  We consider an operator-valued $k$-cochain  $\hat{\lambda}^k_i$ with eigenvalues on each $k$-cell being 0 or 1, which physically represents the $
\ZZ_2$-valued ``matter'' fields transformed under $\ZZ_2$ global symmetry.   We can relate the operator-valued cochain on each $k$-cell $\sigma_k$ to the Pauli-$Z$ operator as
\be\label{eq:Pauli_relation}
(-1)^{\hat{\lambda}^k_i(\sigma_k)}=\tilde Z^{i}(\sigma_k),
\ee
or equivalently
\be
\hat{\lambda}^k_i(\sigma_k) = [1-Z^{i}(\sigma_k)]/2.
\ee

The topological response action of the higher-form SPT can be expressed as
\be\label{sec:SPT_action}
S_\text{SPT} = \pi \int_{\tilde{\L}}  \red{A^p} \cup  \blue{B^q} \cup \green{C^s}. 
\ee
Here, the $p$-, $q$- and $s$-cochains $\red{A^p}$,  $\blue{B^q}$ and $\green{C^s}$ represent the background gauge fields, which are classical variables.  This theory is a higher-form generalization of the group cohomology SPT in Ref.~\cite{Chen:2013foa}, and can be considered as a higher-form analog of ``type-III cocycle" SPT phases \cite{deWildPropitius:1995cf, Yoshida:2015cia}. In particular, Yoshida provides detailed lattice model construction of higher-form SPTs in Ref.~\cite{Yoshida:2015cia} (see also \cite{Tsui:2019ykk,Chen:2021xks} for lattice models of other higher-form SPTs).

Here we follow the gauging procedure introduced in Refs.~\cite{Yoshida:2015cia, barkeshli2023codimension}, where both papers have detailed lattice construction.  We start with a trivial Hamiltonian of the Ising model in the paramagnetic phase:
\begin{align}
    H^0=-\sum_{\sigma_{p-1}} 
    X^{\rd}_{\sigma_{p-1}}
    -\sum_{\sigma_{q-1}} \tilde X^{\bl}_{\sigma_{q-1}}-\sum_{\sigma_{s-1}} \tilde X^{\gr}_{\sigma_{s-1}},
\end{align}
where $\sigma_{p-1}$, $\sigma_{q-1}$ and $\sigma_{s-1}$ represent $(p-1)$-cells, $(q-1)$-cells and $(s-1)$-cells.  We construct the following entangler \cite{Yoshida:2015cia}:
\begin{equation} \label{eq:entangler}
    U= (-1)^{\int_{{\L}} \hat{\lambda}^{p-1}_{\rd}\cup d\hat{\lambda}^{q-1}_{\bl}\cup d\hat{\lambda}^{s-1}_{\gr}},
\end{equation}
where the sum in the exponent is over the entire spatial complex $\L$ of dimension $\mathsf{d}=D-1$.  
Here, for a given $\Z_2$ $k$-cochain $\alpha$, its coboundary $d\alpha$ is a $\Z_2$ $(k+1)$-cochain expressed as
\begin{align}
    d\alpha(\sigma_{k+1}) = \sum_{\sigma_k\in\partial \sigma_{k+1}} \alpha(\sigma_k).
    \label{eq:coboundary}
\end{align}
The operator-valued cochain can be expanded as
$\hat{\lambda}^{k}_i $$ = $$ \sum_{{\sigma}_{k}}\hat{N}_i^k   \tilde{\sigma}_{k}$, where $\hat{N}_i^k \equiv \hat{\lambda}^{k}_i(\sigma_k)$ is the quantum operator with its eigenvalue being  ${N}_i^k \in \{0,1\}$ and has the relation with the Pauli operator as: $(-1)^{\hat{N}_i^k}= \tilde{Z}^i(\sigma_k)$, which is equivalent to Eq.~\eqref{eq:Pauli_relation}.  The classical variable $\tilde{\sigma}_{k}$ is an indicator $k$-cochain that takes value 1 on a single $k$-cell $\sigma_k $ and zero otherwise. 
We can hence rewrite the above entangler in Eq.~\eqref{eq:entangler} as quantum gates:
\begin{align}\label{eq:entangler_re-express}
 U =& (-1)^{\int_{\L} \hat{N}^p_{\rd} \tilde{\sigma}_{p-1}\cup  \hat{N}^q_{\bl} d\tilde{\sigma}_{q-1} \cup \hat{N}^s_{\gr} d \tilde{\sigma}_{s-1}}    \cr
 =& \prod_{{\sigma}_{p-1}, {\sigma}_{q-1}, {\sigma}_{s-1}}  {\big[\text{CCZ}^{{\rd},{\bl},{\gr}}_{\sigma_{p-1}, \sigma_{q-1}, \sigma_{s-1}}\big]}^{\int_{\L} \tilde{\sigma}_{p-1}\cup  d\tilde{\sigma}_{q-1} \cup d \tilde{\sigma}_{s-1}} \cr
\equiv &  \prod_{\substack{{\sigma}_{p-1}, {\sigma}_{q-1}, {\sigma}_{s-1}: \\  \int_{\L} \tilde{\sigma}_{p-1}\cup  d\tilde{\sigma}_{q-1} \cup d \tilde{\sigma}_{s-1} \neq 0 } }  \text{CCZ}^{{\rd},{\bl},{\gr}}_{\sigma_{p-1}, \sigma_{q-1}, \sigma_{s-1}},  
\end{align}
where we have used the identity 
\be
(-1)^{\hat{N}^p_{\rd} \cdot \hat{N}^q_{\bl} \cdot \hat{N}^s_{\gr}}= \text{CCZ}^{{\rd},{\bl},{\gr}}_{\sigma_{p-1}, \sigma_{q-1}, \sigma_{s-1}} ,
\ee
since by definition we have
\begin{align}
&\text{CCZ}^{{\rd},{\bl},{\gr}}_{\sigma_{p-1}, \sigma_{q-1}, \sigma_{s-1}}  \ket{{N}^p_{\rd}, {N}^q_{\bl}, {N}^s_{\gr}} \cr
=& (-1)^{{N}^p_{\rd} \cdot {N}^q_{\bl} \cdot {N}^s_{\gr}} \ket{{N}^p_{\rd}, {N}^q_{\bl}, {N}^s_{\gr}}.
\end{align}
Note that in the second line of Eq.~\eqref{eq:entangler_re-express},  the evaluation of the cup product sum in the exponent with value 1 or 0 determines whether the CCZ gate is applied or not, leading to the conditions in the product in the third line of Eq.~\eqref{eq:entangler_re-express}: the CCZ is only acting on the triples of cells $\{\sigma_{p-1}, \sigma_{q-1}, \sigma_{s-1}\}$ satisfying the triple intersection condition on their indicator cochains 
\be
\int_{\L} \tilde{\sigma}_{p-1}\cup  d\tilde{\sigma}_{q-1} \cup d \tilde{\sigma}_{s-1}=1.
\ee

One can then obtain the higher-form SPT Hamiltonian by conjugating the entangler $U$ on the trivial paramagnetic Hamiltonian:  
\begin{align}\label{eq:SPT_Hamiltonian}
     & H_\text{SPT}  = UH^0 U^\dag  \cr
=&    -\sum_{\sigma_{p-1}} \tilde X^{\rd}_{\sigma_{p-1}}(-1)^{\int \tilde{\sigma}_{p-1}\cup d\hat{\lambda}_{q-1}^{\bl}\cup d\hat{\lambda}_{s-1}^{\gr}} \cr
    &-\sum_{\sigma_{q-1}} \tilde X^{\bl}_{\sigma_{q-1}}(-1)^{\int  d\hat{\lambda}_{p-1}^{\rd}\cup \tilde{\sigma}_{q-1}\cup d\hat{\lambda}_{s-1}^{\gr}} \cr
    &-\sum_{\sigma_{s-1}} \tilde X^{\gr}_{\sigma_{s-1}}
    (-1)^{\int d\hat{\lambda}_{p-1}^{\rd}\cup d\hat{\lambda}_{q-1}^{\bl}\cup  \tilde{\sigma}_{s-1}}  \cr
    =& -\sum_{\sigma_{p-1}} \tilde X^{\rd}_{\sigma_{p-1}}\prod_{{\sigma}_{q-1}, {\sigma}_{s-1}} {\text{CZ}^{\bl, \gr}_{\sigma_{q-1}, \sigma_{s-1}}}^{\int \tilde{\sigma}_{p-1}\cup d\tilde{\sigma}_{q-1}\cup d\tilde{\sigma}_{s-1}} \cr
    &-\sum_{\sigma_{q-1}} \tilde X^{\bl}_{\sigma_{q-1}} \prod_{{\sigma}_{p-1},  {\sigma}_{s-1}} {\text{CZ}^{\rd, \gr}_{\sigma_{p-1}, \sigma_{s-1}}}^{\int  d\tilde{\sigma}_{p-1}\cup \tilde{\sigma}_{q-1}\cup d\tilde{\sigma}_{s-1}} \cr
    &-\sum_{\sigma_{s-1}} \tilde X^{\gr}_{\sigma_{s-1}}
   \prod_{{\sigma}_{p-1}, {\sigma}_{q-1}}  {\text{CZ}^{\rd, \bl}_{\sigma_{p-1}, \sigma_{q-1}}}^{\int d\tilde{\sigma}_{p-1}\cup d\tilde{\sigma}_{q-1}\cup  \tilde{\sigma}_{s-1}}.  \cr
\end{align}
Note that we have derived the second equality by re-expressing the operator-valued cochain as $\hat{\lambda}^{k}_i $$ = $$ \sum_{\sigma_{k}}\hat{N}^k_i  \tilde{\sigma}_{k}$ as before, along with the identity $(-1)^{ \hat{N}^k_i \hat{N}^{k'}_j}= \text{CZ}^{i,j}_{\sigma_k, \sigma_{k'}}$, where $i,j \in \{\rd, \bl, \gr\}$ and $k, k' \in \{p-1, q-1, s-1\}$.   
By taking the exponent as the condition in the product as before, we arrive at the final form of the SPT Hamiltonian:
\begin{align}
    H_\text{SPT}   =& -\sum_{\sigma_{p-1}} \tilde X^{\rd}_{\sigma_{p-1}}\prod_{\substack{{\sigma}_{q-1}, {\sigma}_{s-1}: \\ \int \tilde{\sigma}_{p-1}\cup d\tilde{\sigma}_{q-1}\cup d\tilde{\sigma}_{s-1} \neq 0}} {\text{CZ}^{\bl, \gr}_{\sigma_{q-1}, \sigma_{s-1}}} \cr
    &-\sum_{\sigma_{q-1}} \tilde X^{\bl}_{\sigma_{q-1}} \prod_{\substack{{\sigma}_{p-1},  {\sigma}_{s-1}: \\\int  d\tilde{\sigma}_{p-1}\cup \tilde{\sigma}_{q-1}\cup d\tilde{\sigma}_{s-1} \neq 0}} {\text{CZ}^{\rd, \gr}_{\sigma_{p-1}, \sigma_{s-1}}} \cr
    &-\sum_{\sigma_{s-1}} \tilde X^{\gr}_{\sigma_{s-1}}
   \prod_{\substack{{\sigma}_{p-1}, {\sigma}_{q-1}: \\ \int d\tilde{\sigma}_{p-1}\cup d\tilde{\sigma}_{q-1}\cup  \tilde{\sigma}_{s-1}\neq 0}}  {\text{CZ}^{\rd, \bl}_{\sigma_{p-1}, \sigma_{q-1}}}.  \cr
\end{align}
Note that the above SPT Hamiltonian is a Clifford-stabilizer Hamiltonian which consists of dressed $X$-stabilizers, i.e., a local product of $X$-operators dressed by a product of CZ operators.  These are nothing but the parent Hamiltonian of the higher-form generalization of the cluster states widely studied by the quantum information community.  In particular, the special 1-form case of $p=q=s=1$ corresponds to the usual  cluster state. 

The higher-form SPT is protected by a $\red{\Z_2^{(p-1)}}$$\times$$\blue{\Z_2^{(q-1)}}$$\times$$\green{\Z_2^{(s-1)}}$ symmetry, where $\Z_2^{(k)}$ represents $\ZZ_2$ $k$-form symmetry. Each symmetry is generated by a transversal operator as the following:  
\begin{align}
    \begin{split}
        \red{\Z_2^{(p-1)}}: & \quad \prod_{ \sigma_{p-1}\subset \gamma^*_{\mathsf{d}-(p-1)}} \tilde{X}^{\rd}_{\sigma_{p-1}}, \\
        \blue{\Z_2^{(q-1)}}: & \quad \prod_{ \sigma_{q-1}\subset \gamma^*_{\mathsf{d}-(q-1)}} \tilde{X}^{\bl}_{\sigma_{q-1}}, \\
         \green{\Z_2^{(s-1)}}: & \quad \prod_{ \sigma_{s-1}\subset \gamma^*_{\mathsf{d}-(s-1)}} \tilde{X}^{\gr}_{\sigma_{s-1}}. \\
         \label{eq:globalsym}
    \end{split}
\end{align}
One can verify that when conjugated by the above symmetry generators, the  $H_\text{SPT}$ remains invariant.

Note that the above symmetries are generalized global symmetries which do not act on the entire system as the usual 0-form global symmetry does.  Instead, they act on a subcomplex of the dual complex $\L^*$ corresponding to cycles with codimension $p-1$, $q-1$ and $s-1$, i.e., $\gamma^*_{\mathsf{d}-(p-1)}$, $\gamma^*_{\mathsf{d}-(q-1)}$ and $\gamma^*_{\mathsf{d}-(s-1)}$, which are $\mathsf{d}-(p-1)$-cycle, $\mathsf{d}-(q-1)$-cycle and $\mathsf{d}-(s-1)$-cycle on the dual spatial complex $\L^*$ respectively, where $\mathsf{d}=D-1$ is the spatial dimension.  We also note that conventionally the $k$-form symmetry is defined on a codimension-$k$ submanifold, while here we have generalized the definition to a codimension-$k$ subcomplex $\L^*_k \subset \L^*$.   Note that the fact that $\L$ is a Poincar\'e complex which carries the Poincar\'e duality makes this generalization more straightforward.

\subsection{Constructing twisted qLDPC codes via gauging}

We now present the gauging procedure following the treatment in Refs.~\cite{Yoshida:2015cia, barkeshli2023codimension}. 
In general, there are two equivalent notions of gauging:
\begin{enumerate}
\item 
Promote the static background gauge fields $A$ to dynamical gauge fields $a$ with quantum fluctuation. 
\item
Promote the global symmetries to local gauge symmetries.
\end{enumerate}

We can understand notion 1 as promoting the background gauge fields $\red{A^p}$,  $\blue{B^q}$ and $\green{C^s}$ in the SPT action $S_\text{SPT}$ from Eq.~\eqref{sec:SPT_action} to dynamical gauge fields $\red{a^p}$,  $\blue{b^q}$ and $\green{c^s}$ in the action of the twisted gauge theory in Eq.~\eqref{eq:non-Abelian_action}.  The gauging procedure effectively applies the following mapping:
\begin{align}
& S_\text{SPT} = \pi \int_{\tilde{\L}}  \red{A^p} \cup  \blue{B^q} \cup \green{C^s}  \cr
&\rightarrow   S= \pi \int_{\tilde{\L}} \red{a^p} \cup \blue{b^q} \cup \green{c^s} + \text{BF terms}.
\end{align}

Notion 2 is the one we will pursue in this paper to construct the Clifford stabilizer models of the non-Abelian qLDPC codes. We will promote the global symmetries in Eq.~\eqref{eq:globalsym} to local gauge symmetries by introducing new degrees of freedom—$\mathbb{Z}_2$ gauge fields—to the microscopic model and modifying the Hamiltonian by imposing the Gauss law.

We begin with introducing $\ZZ_2$ gauge fields and the corresponding gauge qubits on the $p$-, $q$- and $s$-cells in copy $\rd$, $\bl$ and $\gr$ respectively.   The associated Pauli operators of the gauge qubits are denoted by $X'^i,Y'^i,Z'^i$, where $i= \rd, \bl, \gr$.    The three types of $\ZZ_2$ gauge fields $\red{\hat{a}'^p}, \blue{\hat{b}'^q}, \green{\hat{c}'^s}$ on the $\rd$, $\bl$ and $\gr$ copies represented by operator-valued $p$-, $q$- and $s$-cochains are related to the Pauli-$Z$ operators as:
\begin{align}
    \begin{split}
        (-1)^{\red{\hat{a}'^p}(\sigma_p)} = Z'^{\rd}_{\sigma_p}, \ (-1)^{\blue{\hat{b}'^q}(\sigma_q)} = Z'^{\bl}_{\sigma_q}, \  (-1)^{\green{\hat{c}'^s}(\sigma_s)} = Z'^{\gr}_{\sigma_s}, \\
    \end{split}
\end{align}

The key idea of gauging with notion 2 is that the matter fields $\hat{\lambda}^k_i$ in Eq.~\eqref{eq:Pauli_relation} which were transformed under the $\ZZ_2$ global symmetries in Eq.~\eqref{eq:globalsym} before gauging are now transformed under local $\ZZ_2$ gauge symmetries.  The gauge transformations are specified by $
\ZZ_2$ cochains $\mu^{p-1}$, $\nu^{q-1}$, $\rho^{s-1}$.  Under the $\ZZ_2$ gauge transformation, the matter fields are hence transformed by: 
\begin{align}
    \hat{\lambda}_{\rd}^{p-1}&\rightarrow \hat{\lambda}_{\rd}^{p-1}-\mu^{p-1}~,\cr 
    \lambda_{\bl}^{q-1}&\rightarrow \lambda_{\bl}^{q-1}-\nu^{q-1}~,\cr
    \lambda_{\gr}^{s-1}&\rightarrow \lambda_{\gr}^{s-1}-\rho^{s-1}~,
    \label{eq:gaugetrans1}
\end{align}
while the gauge fields are transformed correspondingly as
\begin{align}
    \red{\hat{a}'^p} &\rightarrow \red{\hat{a}'^p}+d\mu^{p-1}~,\cr
    \blue{\hat{b}'^q} &\rightarrow \blue{\hat{b}'^q}+d\nu^{q-1}~,\cr
    \green{\hat{c}'^s} &\rightarrow \green{\hat{c}'^s}+d\rho^{s-1}~,
    \label{eq:gaugetrans2}
\end{align}
where $d$ represents the coboundary operator.

We choose the generators of the above gauge transformations by setting $\mu^{p-1}= \tilde{\sigma}_{p-1}$, $\nu^{q-1}= \tilde{\sigma}_{q-1}$, and $\rho^{s-1}= \tilde{\sigma}_{s-1}$, where  $ \tilde{\sigma}_{p-1}$, $\tilde{\sigma}_{q-1}$ and $\tilde{\sigma}_{s-1}$ are indicator cochains as introduced before.  Such gauge transformations only act on matter fields $\hat{\lambda}_{\rd}^{p-1}$, $\hat{\lambda}_{\bl}^{q-1}$, and $\hat{\lambda}_{\gr}^{s-1}$ at a single site supported on cells $\sigma_{p-1}$,  $\sigma_{q-1}$, and $\sigma_{s-1}$ respectively.  Therefore, the gauge transformations are  generated by the following Gauss's law operators:
\begin{align}\label{eq:Gauss_law}
    G^{\rd}(\sigma_{p-1}) = &\tilde X^{\rd}_{\sigma_{p-1}}\prod_{\sigma_{p-1}\subset \partial \sigma_p} X'^{\rd}_{\sigma_p}~,\cr
    G^{\bl}(\sigma_{q-1}) =&\tilde X^{\bl}_{\sigma_{q-1}}\prod_{\sigma_{q-1}\subset \partial \sigma_q} X'^{\bl}_{\sigma_q}~,\cr
    G^{\gr}(\sigma_{s-1}) = &\tilde X^{\gr}_{\sigma_{s-1}}\prod_{\sigma_{s-1}\subset \partial \sigma_s} X'^{\gr}_{\sigma_s}~.
\end{align}
One can easily verify that they indeed generate the gauge transformation in Eqs.~\eqref{eq:gaugetrans1} and \eqref{eq:gaugetrans2}.
As an example, one can see that $G^{\rd}(\sigma_{p-1})$ acts by the following transformation:  $\hat{\lambda}_{\rd}^{p-1}\to \hat{\lambda}_{\rd}^{p-1}-\tilde{\sigma}_{p-1}$, $\red{\hat{a}'^p}\to \red{\hat{a}'^p} + d\tilde{\sigma}_{p-1}$.

Now the gauging is carried out by imposing the Gauss's law constraint $G^{\rd}=1, G^{\bl}=1, G^{\gr}=1$ for each cell $\sigma_{p-1}, \sigma_{q-1}, \sigma_{s-1}$. 
These Gauss's law constraints characterize the gauge invariant ``physical'' Hilbert space $\mathcal{H}_{\text{phys}}$ spanned by the gauge invariant states $G^i\ket{\text{phys}} = \ket{\text{phys}}$. 

Next, we minimally couple the Hamiltonian $H_{\text{SPT}}$ with the $\Z_2$ gauge fields so that it commutes with the Gauss's law operators and acts within the physical Hilbert space $\mathcal{H}_{\text{phys}}$. One can achieve this by the following 
replacement of the coboundaries in Eq.~\eqref{eq:SPT_Hamiltonian}:
\begin{align}
d\hat{\lambda}_{\rd}^{p-1} \rightarrow d\hat{\lambda}_{\rd}^{p-1}+\red{\hat{a}'^p} \equiv  \red{\hat{a}^p},   \cr
d\hat{\lambda}_{\bl}^{q-1} \rightarrow d\hat{\lambda}_{\bl}^{q-1}+\blue{\hat{b}'^q} \equiv  \blue{\hat{b}^q},  \cr
d\hat{\lambda}_{\gr}^{s-1} \rightarrow d\hat{\lambda}_{\gr}^{s-1}+\green{\hat{c}'^s} \equiv \green{\hat{c}^s}. 
\end{align}
Here,  we have introduced the new gauge fields $\red{\hat{a}^p}$,  $\blue{\hat{b}^q}$ and $\green{\hat{c}^s}$, which are related to $\red{\hat{a}'^p}$,  $\blue{\hat{b}'^q}$ and $\green{\hat{c}'^s}$ by the gauge transformations with $\mu^{p-1} = \hat{\lambda}^{p-1}, \ \nu^{q-1} = \hat{\lambda}^{q-1}, \rho^{s-1} = \hat{\lambda}^{s-1}$ in Eq.~\eqref{eq:gaugetrans2} that eliminates the matter fields $\hat{\lambda}$.

Due to the introduction of the new gauge fields, we  define the new Pauli operators on each cell $\sigma_p,\sigma_q,\sigma_s$ accordingly:
\begin{align}
    \begin{split}
        X^{\rd}_{\sigma_p} &= X'^{\rd}_{\sigma_p}~, \\
        Z^{\rd}_{\sigma_p} &= (-1)^{\red{\hat{a}^p}}= (-1)^{[d\hat{\lambda}_{p-1} + \red{\hat{a}'^p}](\sigma_p)}~, \\
        X^{\bl}_{\sigma_q} &= X'^{\bl}_{\sigma_q}~, \\
        Z^{\bl}_{\sigma_q} &= (-1)^{\blue{\hat{b}^q}}= (-1)^{[d\hat{\lambda}_{q-1} + \blue{\hat{b}'^q}](\sigma_q)}~, \\
        X^{\gr}_{\sigma_s} &= X'^{\gr}_{\sigma_s}~, \\
        Z^{\gr}_{\sigma_s} &= (-1)^{\green{\hat{c}^s}} = (-1)^{[d\hat{\lambda}_{s-1} + \green{\hat{c}'^s}](\sigma_s)}~. \\
        \label{eq:gaugeinvariant_paulis}
    \end{split}
\end{align}
We further impose  the flat gauge fields condition, or equivalently the cocycle condition: $\red{d\hat{a}^p}$$=$$0, \blue{d\hat{b}^q}$$=$$0, \green{d\hat{c}^s}$$=$$0$, where $\red{\hat{a}^p}$, $\blue{\hat{b}^q}$ and $\green{\hat{c}^s}$ now become operator-valued cocycles instead of general cochains. We can hence derive the gauged SPT Hamiltonian: 
\begin{align}
    \tilde{H}&=H_\text{Gauss}+H_\text{Flux}~,
\end{align}
where the Gauss's law term is given by
\begin{align}\label{eq:Gauss_law_term}
    H_\text{Gauss}&=-\sum_{\sigma_{p-1}} \left(\prod_{\sigma_{p-1}\subset \partial \sigma_p} X^{\rd}_{\sigma_{p}}\right)(-1)^{\int \tilde{\sigma}_{p-1}\cup \blue{\hat{b}^q}\cup \green{\hat{c}^s}} \cr   
    &\quad -\sum_{\sigma_{q-1}}\left(\prod_{\sigma_{q-1}\subset \partial \sigma_q} X^{\bl}_{\sigma_{q}}\right)(-1)^{\int \red{\hat{a}^p}\cup \tilde{\sigma}_{q-1}\cup  \green{\hat{c}^s}}\cr 
    &\quad  -\sum_{\sigma_{s-1}}\left(\prod_{\sigma_{s-1}\subset \partial \sigma_s} X^{\gr}_{\sigma_{s}}\right)(-1)^{\int \red{\hat{a}^p}\cup \blue{\hat{b}^q}\cup \tilde{\sigma}_{s-1}}~,
\end{align}
Note that we have re-expressed $\tilde X$ term in $H_{\text{SPT}}$ in terms of a product of $X$ in cells of one dimension higher by using the Gauss law constraints $G^i=1$ according to Eq.~\eqref{eq:Gauss_law} as well as the identification of $X'$ and $X$ in Eq.~\eqref{eq:gaugeinvariant_paulis}.
In addition, the flux term
\begin{equation}
\begin{split}
    H_\text{Flux} = &-\sum_{\sigma_{p+1}}\left(\prod_{\sigma_{p}\subset \partial \sigma_{p+1}}Z_{\sigma_p}^{\rd}\right) 
     -\sum_{\sigma_{q+1}}\left(\prod_{\sigma_{q}\subset \partial \sigma_{q+1}}Z_{\sigma_q}^{\bl}\right)\\
 &-    \sum_{\sigma_{s+1}}\left(\prod_{\sigma_{s}\subset \partial \sigma_{s+1}}Z_{\sigma_{s}}^{\gr}\right)~,
 \label{eq:Hflux}
\end{split}
\end{equation}
energetically enforces the flat $\Z_2$ gauge fields (cocycle) condition $\red{d\hat{a}^p}=0, \ \blue{d\hat{b}^q}=0, \ \green{d\hat{c}^s}=0$ on each cell $\sigma_{p+1},\sigma_{q+1}, \sigma_{s+1}$ in the ground-state subspace of $\tilde{H}$ in analogy to the plaquette term in standard (1-form) toric code.

We then re-express the Gauss's law terms in Eq.~\eqref{eq:Gauss_law_term} by expanding the gauge fields as $\red{\hat{a}^p} $$ = $$ \sum_{\sigma_{p}}\hat{N}^p_{\rd}  \tilde{\sigma}_{p}$,   $\blue{\hat{b}^q} $$ = $$ \sum_{\sigma_{q}}\hat{N}^q_{\bl}  \tilde{\sigma}_{q}$, $\green{\hat{c}^s} $$ = $$ \sum_{\sigma_{s}}\hat{N}^s_{\gr}  \tilde{\sigma}_{s}$,  along with the identity $(-1)^{ \hat{N}^k_i \hat{N}^{k'}_j}= \text{CZ}^{i,j}_{\sigma_k, \sigma_{k'}}$, where $i,j \in \{\rd, \bl, \gr\}$ and $k, k' \in \{p, q, s\}$,  which turns the terms with gauge fields into CZ operators. 
The rewritten gauged SPT Hamiltonian $\tilde{H}=H_\text{Gauss}+H_\text{Flux}$ then gives rise to the Clifford stabilizer generators defined in Eq.~\eqref{eq:stabilizer_summary}  in Sec.\ref{sec:code_construction},  which we list again as follows:
\begin{align}
\tilde{A}^{\rd}_{\sigma_{p-1}} =& A^{\rd}_{\sigma_{p-1}} \prod_{\sigma_q,\sigma_s: \int\tilde{\sigma}_{p-1}\cup \tilde{\sigma}_q\cup \tilde{\sigma}_s\neq 0}\text{CZ}^{\bl,\gr}_{\sigma_q, \sigma_s},  \cr
\tilde{A}^{\bl}_{\sigma_{q-1}} =& A^{\bl}_{\sigma_{q-1}}\prod_{\sigma_p,\sigma_s: \int \tilde{\sigma}_p\cup \tilde{\sigma}_{q-1}\cup  \tilde{\sigma}_{s} \neq 0}\text{CZ}^{\rd, \gr}_{\sigma_p, \sigma_s}, \cr
 \tilde{A}^{\gr}_{\sigma_{s-1}} =& A^{\gr}_{\sigma_{s-1}}\prod_{\sigma_p,\sigma_q: \int \tilde{\sigma}_p\cup \tilde{\sigma}_q\cup \tilde{\sigma}_{s-1} \neq 0}\text{CZ}^{\rd,\bl}_{\sigma_p, \sigma_q},  \cr
B^{\rd}_{\sigma_{p+1}}=& \prod_{\sigma_{p}\subset \partial \sigma_{p+1}}Z_{\sigma_p}^{\rd}, \quad  B^{\bl}_{\sigma_{q+1}}=\prod_{\sigma_{q}\subset \partial \sigma_{q+1}}Z_{\sigma_q}^{\bl}   \cr 
B^{\gr}_{\sigma_{s+1}} =& \prod_{\sigma_{s}\subset \partial \sigma_{s+1}}Z_{\sigma_{s}}^{\gr}.
\end{align}

Note that the untwisted code $\C$ in Eq.~\eqref{eq:stabilizer_untwisted} can be obtained by directly gauging the trivial paramagnetic Hamiltonian $H^0$ instead of the SPT Hamiltonian $H_\text{SPT}$ with completely the same gauging procedure dsicussed above.  The CZ terms in the Gauss's law term $A$ go away because they were absent in $H^0$. 

\section{Subsystem-code distance}\label{app:subsystem}

Here is the proof of Lemma \ref{lemma:subsystem}:
\begin{proof}
Due to the intersection pairing
\[
|\eta_k \cap \eta'^k|\equiv \int_{\eta_k} \eta'^k= \delta_{\eta, \eta'},
\]
a logical \(Z\) operator supported on a \(k\)-cycle \(\eta_k\) anticommutes with exactly one logical \(X\) operator, namely the one supported on its conjugate \(k\)-cocycle \(\eta^k\). This canonical anticommutation relation is
\begin{equation}
\lo{Z}_{\eta_k}\,\lo{X}_{\eta^k} = -\,\lo{X}_{\eta^k}\,\lo{Z}_{\eta_k}.
\end{equation}
Consequently, an \(X\)-type error can flip the eigenvalue of \(\lo{Z}_{\eta_k}\) only if its support wraps nontrivially around the conjugate  cocycle \(\eta^k\). Dually, a \(Z\)-type error must be supported on the cycle \(\eta_k\) in order to flip the eigenvalue of \(\lo{X}_{\eta^k}\).

It follows that the \(Z\)-distance of the subsystem code is set by the shortest nontrivial representative among the basis \(k\)-cycles \(\{\eta_k\}\),
\begin{equation}
d_Z = \min\{|\eta_k|\},
\end{equation}
while the \(X\)-distance is determined by the smallest basis \(k\)-cocycle \(\{\eta'^k\}\),
\begin{equation}
d_X = \min\{|\eta'^k|\}.
\end{equation}
The overall distance of the subsystem code is therefore
\begin{equation}
d = \min(d_Z, d_X)
  = \min\left(\min\{|\eta_k|\}, \min\{|\eta'^k|\}\right).
\end{equation}

\end{proof}

\end{appendix}

\bibliography{mybib_merge, biblio.bib}

@misc{GSD,
note={We calculated the ground state degeneracy (GSD) on a computer using the method described in Appendix B of \refcite{ShirleyTwisted},
  which summarizes parts of \refcite{stabilizerStandardForm,GottesmanThesis}.
  In practice, we calculate the GSD for several examples of system sizes and level coefficients.
  We then guess an analytic form for the GSD, which we check for many examples (using the computer program).}
}

@misc{Cordova:2024mqg,
    author = "Cordova, Clay and Costa, Davi B. and Hsin, Po-Shen",
    title = "{Non-Invertible Symmetries as Condensation Defects in Finite-Group Gauge Theories}",
    eprint = "2412.16681",
    archivePrefix = "arXiv",
    primaryClass = "cond-mat.str-el",
    month = "12",
    year = "2024"
}

@article{Choi:2022zal,
    author = "Choi, Yichul and Cordova, Clay and Hsin, Po-Shen and Lam, Ho Tat and Shao, Shu-Heng",
    title = "{Non-invertible Condensation, Duality, and Triality Defects in 3+1 Dimensions}",
    eprint = "2204.09025",
    archivePrefix = "arXiv",
    primaryClass = "hep-th",
    reportNumber = "YITP-SB-2022-16, MIT/CTP-5423, YITP-SB-2022-16, MIT/CTP-5423",
    doi = "10.1007/s00220-023-04727-4",
    journal = "Commun. Math. Phys.",
    volume = "402",
    number = "1",
    pages = "489--542",
    year = "2023"
}

@article{Johnson_Freyd_2022,
   title={On the Classification of Topological Orders},
   volume={393},
   ISSN={1432-0916},
   url={http://dx.doi.org/10.1007/s00220-022-04380-3},
   DOI={10.1007/s00220-022-04380-3},
   number={2},
   journal={Communications in Mathematical Physics},
   publisher={Springer Science and Business Media LLC},
   author={Johnson-Freyd, Theo},
   year={2022},
   month=apr, pages={989–1033} }

@article{Kong_2020,
   title={Algebraic higher symmetry and categorical symmetry: A holographic and entanglement view of symmetry},
   volume={2},
   ISSN={2643-1564},
   url={http://dx.doi.org/10.1103/PhysRevResearch.2.043086},
   DOI={10.1103/physrevresearch.2.043086},
   number={4},
   journal={Physical Review Research},
   publisher={American Physical Society (APS)},
   author={Kong, Liang and Lan, Tian and Wen, Xiao-Gang and Zhang, Zhi-Hao and Zheng, Hao},
   year={2020},
   month=oct }

@misc{Kong:2024ykr,
    author = "Kong, Liang and Zhang, Zhi-Hao and Zhao, Jiaheng and Zheng, Hao",
    title = "{Higher condensation theory}",
    eprint = "2403.07813",
    archivePrefix = "arXiv",
    primaryClass = "cond-mat.str-el",
    month = "3",
    year = "2024"
}

@article{Witten:1989wf,
    author = "Witten, Edward",
    title = "{Gauge Theories and Integrable Lattice Models}",
    reportNumber = "IASSNS-HEP-89-11",
    doi = "10.1016/0550-3213(89)90232-0",
    journal = "Nucl. Phys. B",
    volume = "322",
    pages = "629--697",
    year = "1989"
}

@misc{Inamura:2023ldn,
    author = "Inamura, Kansei and Wen, Xiao-Gang",
    title = "{2+1D symmetry-topological-order from local symmetric operators in 1+1D}",
    eprint = "2310.05790",
    archivePrefix = "arXiv",
    primaryClass = "cond-mat.str-el",
    month = "10",
    year = "2023"
}

@misc{Hsin:2025zgn,
    author = "Hsin, Po-Shen and Kobayashi, Ryohei",
    title = "{Automorphism in Gauge Theories: Higher Symmetries and Transversal Non-Clifford Logical Gates}",
    eprint = "2511.15783",
    archivePrefix = "arXiv",
    primaryClass = "cond-mat.str-el",
    month = "11",
    year = "2025"
}

@misc{Kobayashi:2025cfh,
    author = "Kobayashi, Ryohei and Zhu, Guanyu and Hsin, Po-Shen",
    title = "{Clifford Hierarchy Stabilizer Codes: Transversal Non-Clifford Gates and Magic}",
    eprint = "2511.02900",
    archivePrefix = "arXiv",
    primaryClass = "quant-ph",
    month = "11",
    year = "2025"
}

@article{PhysRevX.8.021074,
  title = {Classification of $\mathbf{(}3+1\mathbf{)}\mathrm{D}$ Bosonic Topological Orders: The Case When Pointlike Excitations Are All Bosons},
  author = {Lan, Tian and Kong, Liang and Wen, Xiao-Gang},
  journal = {Phys. Rev. X},
  volume = {8},
  issue = {2},
  pages = {021074},
  numpages = {24},
  year = {2018},
  month = {Jun},
  publisher = {American Physical Society},
  doi = {10.1103/PhysRevX.8.021074},
  url = {https://link.aps.org/doi/10.1103/PhysRevX.8.021074}
}

@misc{Davydova:2025ylx,
    author = "Davydova, Margarita and Bauer, Andreas and de la Fuente, Julio C. Magdalena and Webster, Mark and Williamson, Dominic J. and Brown, Benjamin J.",
    title = "{Universal fault tolerant quantum computation in 2D without getting tied in knots}",
    eprint = "2503.15751",
    archivePrefix = "arXiv",
    primaryClass = "quant-ph",
    month = "3",
    year = "2025"
}

@misc{Hsin:2024nwc,
    author = "Hsin, Po-Shen and Kobayashi, Ryohei and Zhu, Guanyu",
    title = "{Classifying Logical Gates in Quantum Codes via Cohomology Operations and Symmetry}",
    eprint = "2411.15848",
    archivePrefix = "arXiv",
    primaryClass = "quant-ph",
    month = "11",
    year = "2024"
}

@article{Roumpedakis:2022aik,
    author = "Roumpedakis, Konstantinos and Seifnashri, Sahand and Shao, Shu-Heng",
    title = "{Higher Gauging and Non-invertible Condensation Defects}",
    eprint = "2204.02407",
    archivePrefix = "arXiv",
    primaryClass = "hep-th",
    reportNumber = "YITP-SB-2022-14",
    doi = "10.1007/s00220-023-04706-9",
    journal = "Commun. Math. Phys.",
    volume = "401",
    number = "3",
    pages = "3043--3107",
    year = "2023"
}

@article{Hsin:2024pdi,
    author = "Hsin, Po-Shen and Kobayashi, Ryohei and Zhu, Guanyu",
    title = "{Non-Abelian Self-Correcting Quantum Memory and Transversal Non-Clifford Gate Beyond the n1/3 Distance Barrier}",
    eprint = "2405.11719",
    archivePrefix = "arXiv",
    primaryClass = "quant-ph",
    doi = "10.1103/hkgn-hhqx",
    journal = "PRX Quantum",
    volume = "6",
    number = "4",
    pages = "040360",
    year = "2025"
}

@article{Turaev:1992hq,
    author = "Turaev, V. G. and Viro, O. Y.",
    title = "{State sum invariants of 3-manifolds and quantum 6
				j-symbols
			}",
    doi = "10.1016/0040-9383(92)90015-A",
    journal = "Topology",
    volume = "31",
    pages = "865--902",
    year = "1992"
}

@misc{walker2021universalstatesum,
      title={A universal state sum}, 
      author={Kevin Walker},
      year={2021},
      eprint={2104.02101},
      archivePrefix={arXiv},
      primaryClass={math.QA},
      url={https://arxiv.org/abs/2104.02101}, 
}

@book{568530b2-62d5-3d43-9e2f-adfdf424006b,
 ISBN = {9780691147338},
 URL = {http://www.jstor.org/stable/j.ctt12f4hf},
 author = {Daniel L. Stein and Charles M. Newman},
 publisher = {Princeton University Press},
 title = {Spin Glasses and Complexity},
 urldate = {2025-12-31},
 year = {2013}
}

@article{ANDERSON1970549,
title = {Localisation theory and the Cu-Mn problem: Spin glasses},
journal = {Materials Research Bulletin},
volume = {5},
number = {8},
pages = {549-554},
year = {1970},
issn = {0025-5408},
doi = {https://doi.org/10.1016/0025-5408(70)90096-6},
url = {https://www.sciencedirect.com/science/article/pii/0025540870900966},
author = {P.W. Anderson}
}

@article{PhysRevLett.43.1754,
  title = {Infinite Number of Order Parameters for Spin-Glasses},
  author = {Parisi, G.},
  journal = {Phys. Rev. Lett.},
  volume = {43},
  issue = {23},
  pages = {1754--1756},
  numpages = {0},
  year = {1979},
  month = {Dec},
  publisher = {American Physical Society},
  doi = {10.1103/PhysRevLett.43.1754},
  url = {https://link.aps.org/doi/10.1103/PhysRevLett.43.1754}
}

@article{Binder:1986zz,
    author = "Binder, K. and Young, A. P.",
    title = "{Spin glasses: Experimental facts, theoretical concepts, and open questions}",
    doi = "10.1103/RevModPhys.58.801",
    journal = "Rev. Mod. Phys.",
    volume = "58",
    pages = "801--976",
    year = "1986"
}

@misc{Freedman:2013zfj,
    author = "Freedman, M. H. and Hastings, M. B.",
    title = "{Quantum Systems on Non-$k$-Hyperfinite Complexes: A Generalization of Classical Statistical Mechanics on Expander Graphs}",
    eprint = "1301.1363",
    archivePrefix = "arXiv",
    primaryClass = "quant-ph",
    month = "1",
    year = "2013"
}

@inproceedings{Anshu:2022hsn,
    author = "Anshu, Anurag and Breuckmann, Nikolas P. and Nirkhe, Chinmay",
    title = "{NLTS Hamiltonians from Good Quantum Codes}",
    booktitle = "{55th Annual ACM Symposium on Theory of Computing}",
    eprint = "2206.13228",
    archivePrefix = "arXiv",
    primaryClass = "quant-ph",
    doi = "10.1145/3564246.3585114",
    month = "6",
    year = "2022"
}

@misc{Warman:2025hov,
    author = "Warman, Alison and Schafer-Nameki, Sakura",
    title = "{Transversal Clifford-Hierarchy Gates via Non-Abelian Surface Codes}",
    eprint = "2512.13777",
    archivePrefix = "arXiv",
    primaryClass = "quant-ph",
    month = "12",
    year = "2025"
}

@article{Bauer:2024qpc,
    author = "Bauer, Andreas",
    title = "{Low-overhead non-Clifford fault-tolerant circuits for all non-chiral abelian topological phases}",
    eprint = "2403.12119",
    archivePrefix = "arXiv",
    primaryClass = "quant-ph",
    doi = "10.22331/q-2025-03-25-1673",
    journal = "Quantum",
    volume = "9",
    pages = "1673",
    year = "2025"
}

@article{Bauer:2023awl,
    author = "Bauer, Andreas",
    title = "{Topological error correcting processes from fixed-point path integrals}",
    eprint = "2303.16405",
    archivePrefix = "arXiv",
    primaryClass = "quant-ph",
    doi = "10.22331/q-2024-03-20-1288",
    journal = "Quantum",
    volume = "8",
    pages = "1288",
    year = "2024"
}

@article{Bauer:2024alh,
    author = "Bauer, Andreas",
    title = "{x+y Floquet code: A simple example for topological quantum computation in the path-integral approach}",
    eprint = "2408.07265",
    archivePrefix = "arXiv",
    primaryClass = "quant-ph",
    doi = "10.1103/PhysRevA.111.032413",
    journal = "Phys. Rev. A",
    volume = "111",
    number = "3",
    pages = "032413",
    year = "2025"
}

@article{Hastings:2021ptn,
    author = "Hastings, Matthew B. and Haah, Jeongwan",
    title = "{Dynamically Generated Logical Qubits}",
    eprint = "2107.02194",
    archivePrefix = "arXiv",
    primaryClass = "quant-ph",
    doi = "10.22331/q-2021-10-19-564",
    journal = "Quantum",
    volume = "5",
    pages = "564",
    year = "2021"
}

@misc{Placke:2024wey,
    author = "Placke, Benedikt and Rakovszky, Tibor and Breuckmann, Nikolas P. and Khemani, Vedika",
    title = "{Topological Quantum Spin Glass Order and its realization in qLDPC codes}",
    eprint = "2412.13248",
    archivePrefix = "arXiv",
    primaryClass = "quant-ph",
    month = "12",
    year = "2024"
}

@book{Peskin:1995ev,
    author = "Peskin, Michael E. and Schroeder, Daniel V.",
    title = "{An Introduction to quantum field theory}",
    doi = "10.1201/9780429503559",
    isbn = "978-0-201-50397-5, 978-0-429-50355-9, 978-0-429-49417-8",
    publisher = "Addison-Wesley",
    address = "Reading, USA",
    year = "1995"
}

@misc{Gaiotto:2019xmp,
    author = "Gaiotto, Davide and Johnson-Freyd, Theo",
    title = "{Condensations in higher categories}",
    eprint = "1905.09566",
    archivePrefix = "arXiv",
    primaryClass = "math.CT",
    month = "5",
    year = "2019"
}

@article{Barkeshli:2022edm,
    author = "Barkeshli, Maissam and Chen, Yu-An and Hsin, Po-Shen and Kobayashi, Ryohei",
    title = "{Higher-group symmetry in finite gauge theory and stabilizer codes}",
    eprint = "2211.11764",
    archivePrefix = "arXiv",
    primaryClass = "cond-mat.str-el",
    doi = "10.21468/SciPostPhys.16.4.089",
    journal = "SciPost Phys.",
    volume = "16",
    pages = "089",
    year = "2024"
}

@article{Tsui:2019ykk,
    author = "Tsui, Lokman and Wen, Xiao-Gang",
    title = "{Lattice models that realize $\mathbb{Z}_n$-1 symmetry-protected topological states for even $n$}",
    eprint = "1908.02613",
    archivePrefix = "arXiv",
    primaryClass = "cond-mat.str-el",
    doi = "10.1103/PhysRevB.101.035101",
    journal = "Phys. Rev. B",
    volume = "101",
    number = "3",
    pages = "035101",
    year = "2020"
}

@article{Dijkgraaf:1989pz,
    author = "Dijkgraaf, Robbert and Witten, Edward",
    title = "{Topological Gauge Theories and Group Cohomology}",
    reportNumber = "THU-89-9, IASSNS-HEP-89-33",
    doi = "10.1007/BF02096988",
    journal = "Commun. Math. Phys.",
    volume = "129",
    pages = "393",
    year = "1990"
}

@article{Putrov:2016qdo,
    author = "Putrov, Pavel and Wang, Juven and Yau, Shing-Tung",
    title = "{Braiding Statistics and Link Invariants of Bosonic/Fermionic Topological Quantum Matter in 2+1 and 3+1 dimensions}",
    eprint = "1612.09298",
    archivePrefix = "arXiv",
    primaryClass = "cond-mat.str-el",
    doi = "10.1016/j.aop.2017.06.019",
    journal = "Annals Phys.",
    volume = "384",
    pages = "254--287",
    year = "2017"
}

@article{Yoshida:2015cia,
    author = "Yoshida, Beni",
    title = "{Topological phases with generalized global symmetries}",
    eprint = "1508.03468",
    archivePrefix = "arXiv",
    primaryClass = "cond-mat.str-el",
    reportNumber = "NSF-KITP-15-095",
    doi = "10.1103/PhysRevB.93.155131",
    journal = "Phys. Rev. B",
    volume = "93",
    number = "15",
    pages = "155131",
    year = "2016"
}

@phdthesis{deWildPropitius:1995cf,
    author = "de Wild Propitius, Mark Dirk Frederik",
    title = "{Topological interactions in broken gauge theories}",
    eprint = "hep-th/9511195",
    archivePrefix = "arXiv",
    school = "Amsterdam U.",
    year = "1995"
}

@article{Chen:2021xks,
    author = "Chen, Yu-An and Hsin, Po-Shen",
    title = "{Exactly solvable lattice Hamiltonians and gravitational anomalies}",
    eprint = "2110.14644",
    archivePrefix = "arXiv",
    primaryClass = "cond-mat.str-el",
    doi = "10.21468/SciPostPhys.14.5.089",
    journal = "SciPost Phys.",
    volume = "14",
    number = "5",
    pages = "089",
    year = "2023"
}

@article{gidney2024magic,
	author = {Gidney, Craig and Shutty, Noah and Jones, Cody},
	journal = {arXiv preprint arXiv:2409.17595},
	title = {Magic state cultivation: growing T states as cheap as CNOT gates},
	year = {2024}}

@article{he2025quantum,
	author = {He, Zhiyang and Vaikuntanathan, Vinod and Wills, Adam and Zhang, Rachel Yun},
	journal = {arXiv preprint arXiv:2502.01864},
	title = {Quantum Codes with Addressable and Transversal Non-Clifford Gates},
	year = {2025}}

@article{fu2025no,
	author = {Fu, Esther Xiaozhen and Zheng, Han and Li, Zimu and Liu, Zi-Wen},
	journal = {arXiv preprint arXiv:2507.16797},
	title = {No-go theorems for logical gates on product quantum codes},
	year = {2025}}

@article{gulshen2025quantum,
	author = {Gulshen, Kyle and Kaufman, Tali},
	journal = {arXiv preprint arXiv:2510.07864},
	title = {Quantum Tanner Color Codes on Qubits with Transversal Gates},
	year = {2025}}

@article{li2025transversal,
	author = {Li, Christine and Preskill, John and Xu, Qian},
	journal = {arXiv preprint arXiv:2510.07269},
	title = {Transversal dimension jump for product qLDPC codes},
	year = {2025}}

@article{tan2025single,
	author = {Tan, Shi Jie Samuel and Hong, Yifan and Lin, Ting-Chun and Gullans, Michael J and Hsieh, Min-Hsiu},
	journal = {arXiv preprint arXiv:2510.08552},
	title = {Single-Shot Universality in Quantum LDPC Codes via Code-Switching},
	year = {2025}}

@article{golowich2025constant,
	author = {Golowich, Louis and Chang, Kathleen and Zhu, Guanyu},
	journal = {arXiv preprint arXiv:2510.06760},
	title = {Constant-overhead addressable gates via single-shot code switching},
	year = {2025}}

@article{menon2025magic,
	author = {Menon, Varun and Bonilla-Ataides, J Pablo and Mehta, Rohan and Tan, Daniel Bochen and Lukin, Mikhail D},
	journal = {arXiv preprint arXiv:2508.10714},
	title = {Magic tricycles: efficient magic state generation with finite block-length quantum LDPC codes},
	year = {2025}}

@article{zhu2025transversal,
	author = {Zhu, Guanyu},
	journal = {arXiv preprint arXiv:2507.15056},
	title = {Transversal non-Clifford gates on qLDPC codes breaking the $\sqrt{N}$ distance barrier and quantum-inspired geometry with systolic freedom},
	year = {2025}}

@article{guemard2025lifting,
	author = {Guemard, Virgile},
	journal = {arXiv preprint arXiv:2505.14327},
	title = {Lifting a CSS code via its handlebody realization},
	year = {2025}}

@article{Hsin2024_non-Abelian,
	adsurl = {https://ui.adsabs.harvard.edu/abs/2024arXiv240511719H},
	archiveprefix = {arXiv},
	author = {{Hsin}, Po-Shen and {Kobayashi}, Ryohei and {Zhu}, Guanyu},
	doi = {10.48550/arXiv.2405.11719},
	eid = {arXiv:2405.11719},
	eprint = {2405.11719},
	journal = {arXiv e-prints},
	month = may,
	pages = {arXiv:2405.11719},
	primaryclass = {quant-ph},
	title = {{Non-Abelian Self-Correcting Quantum Memory and Single-shot Non-Clifford Gate beyond the $n^{1/3}$ Distance Barrier}},
	year = 2024}

@article{zhu2025topological,
	author = {Zhu, Guanyu},
	journal = {arXiv preprint arXiv:2501.19375},
	title = {A topological theory for qLDPC: non-Clifford gates and magic state fountain on homological product codes with constant rate and beyond the $N^{1/3}$ distance barrier},
	year = {2025}}

@article{swaroop2024universal,
	author = {Swaroop, Esha and Jochym-O'Connor, Tomas and Yoder, Theodore J},
	journal = {arXiv preprint arXiv:2410.03628},
	title = {Universal adapters between quantum LDPC codes},
	year = {2024}}

@article{williamson2024low,
	author = {Williamson, Dominic J and Yoder, Theodore J},
	journal = {arXiv preprint arXiv:2410.02213},
	title = {Low-overhead fault-tolerant quantum computation by gauging logical operators},
	year = {2024}}

@article{cross2024improved,
	author = {Cross, Andrew and He, Zhiyang and Rall, Patrick and Yoder, Theodore},
	journal = {arXiv preprint arXiv:2407.18393},
	title = {Improved QLDPC Surgery: Logical Measurements and Bridging Codes},
	year = {2024}}

@article{xu2024fast,
	author = {Xu, Qian and Zhou, Hengyun and Zheng, Guo and Bluvstein, Dolev and Ataides, J and Lukin, Mikhail D and Jiang, Liang},
	journal = {arXiv preprint arXiv:2407.18490},
	title = {Fast and parallelizable logical computation with homological product codes},
	year = {2024}}

@article{huang2023homomorphic,
	author = {Huang, Shilin and Jochym-O'Connor, Tomas and Yoder, Theodore J},
	journal = {PRX Quantum},
	number = {3},
	pages = {030301},
	publisher = {APS},
	title = {Homomorphic logical measurements},
	volume = {4},
	year = {2023}}

@article{Hsin2024:classifying,
	author = {Hsin, Po-Shen and Kobayashi, Ryohei and Zhu, Guanyu},
	doi = {10.48550/arxiv.2411.15848},
	eprint = {2411.15848},
	journal = {arXiv},
	title = {{Classifying Logical Gates in Quantum Codes via Cohomology Operations and Symmetry}},
	year = {2024}}

@book{Bravyi:2014bq,
	author = {Bravyi, Sergey and Hastings, Matthew B},
	doi = {10.1145/2591796.2591870},
	series = {the 46th Annual ACM Symposium},
	title = {{Homological product codes}},
	year = {2014}}

@article{scruby2024quantum,
	author = {Scruby, Thomas R and Pesah, Arthur and Webster, Mark},
	journal = {arXiv preprint arXiv:2408.13130},
	title = {Quantum rainbow codes},
	year = {2024}}

@article{lin2024transversal,
	author = {Lin, Ting-Chun},
	journal = {arXiv preprint arXiv:2410.14631},
	title = {Transversal non-Clifford gates for quantum LDPC codes on sheaves},
	year = {2024}}

@article{golowich2024quantum,
	author = {Golowich, Louis and Lin, Ting-Chun},
	journal = {arXiv preprint arXiv:2410.14662},
	title = {Quantum LDPC Codes with Transversal Non-Clifford Gates via Products of Algebraic Codes},
	year = {2024}}

@article{breuckmann2024cups,
	author = {Breuckmann, Nikolas P and Davydova, Margarita and Eberhardt, Jens N and Tantivasadakarn, Nathanan},
	journal = {arXiv preprint arXiv:2410.16250},
	title = {Cups and Gates I: Cohomology invariants and logical quantum operations},
	year = {2024}}

@article{guefficient22,
	author = {S. Gu and C. Pattison and E. Tang},
	journal = {arXiv:2206.06557},
	title = {An efficient decoder for a linear distance quantum LDPC code},
	year = {2022}}

@inproceedings{pkldpc22,
	address = {New York, NY, USA},
	author = {P. Panteleev and G. Kalachev},
	booktitle = {Proc. ACM STOC},
	commentbooktitle = {STOC 2022: Proceedings of the 54th Annual ACM SIGACT Symposium on Theory of Computing},
	doi = {10.1145/3519935.3520017},
	pages = {375---388},
	publisher = {Association for Computing Machinery},
	title = {Asymptotically good Quantum and locally testable classical LDPC codes},
	year = {2022}}

@inbook{lzdecoding23,
	author = {Anthony Leverrier and Gilles Z{\'e}mor},
	booktitle = {Proc. ACM SODA},
	commentbooktitle = {Proceedings of the 2023 Annual ACM-SIAM Symposium on Discrete Algorithms (SODA)},
	doi = {10.1137/1.9781611977554.ch45},
	pages = {1216-1244},
	title = {Efficient decoding up to a constant fraction of the code length for asymptotically good quantum codes},
	url = {https://epubs.siam.org/doi/abs/10.1137/1.9781611977554.ch45},
	year = {2023}}

@inproceedings{fiberbundlecode21,
	address = {New York, NY, USA},
	author = {M. Hastings and J. Haah and R. O'Donnell},
	booktitle = {Proc. ACM STOC},
	commentbooktitle = {STOC 2021: Proceedings of the 53rd Annual ACM SIGACT Symposium on Theory of Computing},
	doi = {10.1145/3406325.3451005},
	pages = {1276--1288},
	publisher = {Association for Computing Machinery},
	title = {Fiber bundle codes: breaking the $n^{1/2}\textrm{polylog}(n)$ barrier for Quantum LDPC codes},
	year = {2021}}

@article{lh22,
	author = {T. Lin and M. Hsieh},
	journal = {arXiv:2203.03581},
	title = {Good quantum LDPC codes with linear time decoder from lossless expanders},
	year = {2022}}

@inproceedings{dhlv23,
	address = {New York, NY, USA},
	author = {I. Dinur and M. Hsieh and T. Lin and T. Vidick},
	booktitle = {Proc. ACM STOC},
	commentbooktitle = {STOC 2023: Proceedings of the 55th Annual ACM Symposium on Theory of Computing},
	doi = {10.1145/3564246.3585101},
	pages = {905--918},
	publisher = {Association for Computing Machinery},
	title = {Good Quantum LDPC Codes with Linear Time Decoders},
	year = {2023}}

@article{cohen22,
	author = {L. Cohen and I. Kim and S. Bartlett and B. Brown},
	journal = {Sci. Adv.},
	number = {eabn1717},
	title = {Low-overhead fault-tolerant quantum computing using long-range connectivity},
	volume = {8},
	year = {2022}}

@article{hastingswr21,
	author = {M. Hastings},
	journal = {arXiv:2102.10030},
	title = {On quantum weight reduction},
	year = {2021}}

@inproceedings{Quantum_tanner,
	address = {Los Alamitos, CA, USA},
	author = {A. Leverrier and G. Zemor},
	booktitle = {Proc. IEEE FOCS},
	commentbooktitle = {2022 IEEE 63rd Annual Symposium on Foundations of Computer Science (FOCS)},
	doi = {10.1109/FOCS54457.2022.00117},
	pages = {872-883},
	publisher = {IEEE Computer Society},
	title = {Quantum Tanner codes},
	url = {https://doi.ieeecomputersociety.org/10.1109/FOCS54457.2022.00117},
	year = {2022}}

@article{gusingleshot23,
	author = {S. Gu and E. Tang and L. Caha and S. Choe and Z. He and A. Kubica},
	journal = {arXiv:2306.12470},
	title = {Single-shot decoding of good quantum LDPC codes},
	year = {2023}}

@article{Bravyi:2024wc,
	author = {Bravyi, Sergey and Cross, Andrew W. and Gambetta, Jay M. and Maslov, Dmitri and Rall, Patrick and Yoder, Theodore J.},
	date = {2024/03/01},
	doi = {10.1038/s41586-024-07107-7},
	id = {Bravyi2024},
	isbn = {1476-4687},
	journal = {Nature},
	number = {8005},
	pages = {778--782},
	title = {High-threshold and low-overhead fault-tolerant quantum memory},
	url = {https://doi.org/10.1038/s41586-024-07107-7},
	volume = {627},
	year = {2024}}

@article{zhu2023non,
	author = {Zhu, Guanyu and Sikander, Shehryar and Portnoy, Elia and Cross, Andrew W and Brown, Benjamin J},
	journal = {arXiv:2310.16982},
	title = {Non-Clifford and parallelizable fault-tolerant logical gates on constant and almost-constant rate homological quantum LDPC codes via higher symmetries},
	year = {2023}}

@article{barkeshli2023codimension,
	author = {Barkeshli, Maissam and Chen, Yu-An and Huang, Sheng-Jie and Kobayashi, Ryohei and Tantivasadakarn, Nathanan and Zhu, Guanyu},
	journal = {SciPost Physics},
	number = {4},
	pages = {065},
	title = {Codimension-2 defects and higher symmetries in (3+ 1) D topological phases},
	volume = {14},
	year = {2023}}

@article{breuckmann2022fold,
	author = {Breuckmann, Nikolas P and Burton, Simon},
	journal = {arXiv preprint arXiv:2202.06647},
	title = {Fold-transversal Clifford gates for quantum codes},
	year = {2022}}

@article{freedman:2020_manifold_from_code,
	author = {Freedman, Michael and Hastings, Matthew B},
	journal = {arXiv:2012.02249},
	title = {{Building manifolds from quantum codes}},
	year = {2020}}

@article{Breuckmann:2021_balanced,
	author = {Breuckmann, Nikolas P and Eberhardt, Jens N},
	journal = {arXiv:2012.09271},
	title = {{Balanced Product Quantum Codes}},
	year = {2020}}

@article{Tillich:2014_hyergraph_product,
	author = {Tillich, Jean-Pierre and Z{\'e}mor, Gilles},
	doi = {10.1109/TIT.2013.2292061},
	journal = {IEEE Transactions on Information Theory},
	number = {2},
	pages = {1193-1202},
	title = {Quantum LDPC Codes With Positive Rate and Minimum Distance Proportional to the Square Root of the Blocklength},
	volume = {60},
	year = {2014}}

@book{Hatcher:2001ut,
	address = {Cambridge; New York},
	author = {Hatcher, Allen},
	isbn = {9780521795401},
	publisher = {CUP},
	title = {Algebraic Topology},
	ty = {BOOK},
	u1 = {Background (Could Read)},
	u3 = {505479},
	year = {2001}}

@article{JochymOConnor:2021ih,
	author = {Jochym-O{\textquoteright}Connor, Tomas and Yoder, Theodore J},
	journal = {Physical Review Research},
	month = feb,
	number = {1},
	pages = {013118},
	title = {{Four-dimensional toric code with non-Clifford transversal gates}},
	volume = {3},
	year = {2021}}

@article{kitaev2003,
	author = {A.Yu. Kitaev},
	journal = {Annals Phys.},
	pages = {2-30},
	title = {Fault-tolerant quantum computation by anyons},
	volume = {303},
	year = {2003}}

@article{levin2005,
	author = {Levin, Michael A. and Wen, Xiao-Gang},
	doi = {10.1103/PhysRevB.71.045110},
	issue = {4},
	journal = {Phys. Rev. B},
	month = {Jan},
	numpages = {21},
	pages = {045110},
	publisher = {American Physical Society},
	title = {String-net condensation: A physical mechanism for topological phases},
	volume = {71},
	year = {2005}}

@article{Bravyi:2013dx,
	author = {Bravyi, Sergey and K{\"o}nig, Robert},
	journal = {Phys. Rev. Lett.},
	month = apr,
	number = {17},
	pages = {170503--5},
	title = {{Classification of Topologically Protected Gates for Local Stabilizer Codes}},
	volume = {110},
	year = {2013}}

@article{Chen:2013foa,
	author = {Chen, Xie and Gu, Zheng-Cheng and Liu, Zheng-Xin and Wen, Xiao-Gang},
	journal = {Physical Review B},
	month = apr,
	number = {15},
	pages = {155114--17},
	title = {{Symmetry protected topological orders and the group cohomology of their symmetry group}},
	volume = {87},
	year = {2013}}

\end{document}